\documentclass[oneside,10pt]{article}
\usepackage{amsfonts,amsmath,amssymb,mathrsfs}
\usepackage{dsfont}
\usepackage{srcltx}
\usepackage{graphicx}

\newcommand{\be}{\begin{equation}}
\newcommand{\ee}{\end{equation}}

\newcommand{\landau}{\mbox{\begin{scriptsize}$\mathcal{O}$\end{scriptsize}}}

\newcommand{\llaa}{\left\langle\hspace{-0.14cm}\left\langle}
\newcommand{\rraa}{\right\rangle\hspace{-0.14cm}\right\rangle}
\newcommand{\laa}{\langle\hspace{-0.08cm}\langle}
\newcommand{\raa}{\rangle\hspace{-0.08cm}\rangle}

\newcommand{\LZ}{L^2(\mathbb{R}^3,\mathbb{C})}
\newcommand{\LZN}{L^2(\mathbb{R}^{3N},\mathbb{C})}

\newcommand{\C}{\mathcal{K}(\phi)(\alpha(\Psi,\phi)+N^{-\eta})}
\newcommand{\Cphi}{(\|\phi\|_\infty+(\ln N)^{1/3}\|\nabla\phi\|_{6,loc})}
\newcommand{\CphiA}{(\|\phi\|_\infty+(\ln N)^{1/3}\|\nabla\phi\|_{6,loc}+\|\dot A\|_\infty)}

\newcommand{\kinf}{\mathcal{K}\in\mathcal{F}}

\newcommand{\Cgamma}{\mathcal{K}(\phi)(\Gamma(\Psi,\phi)+N^{-\eta})}

\newcommand{\dt}{\frac{\text{d}}{\text{d}t}}
\newcommand{\im}{\text{i}}

\newtheorem{theorem}{Theorem}[section]
\newtheorem{lemma}[theorem]{Lemma}
\newtheorem{notation}[theorem]{Notation}

\newtheorem{corollary}[theorem]  {Corollary}
\newtheorem{remark}[theorem]  {Remark}
\newtheorem{definition}[theorem] {Definition}
\newtheorem{proposition}[theorem]{Proposition}

\newenvironment{proof}{\emph{Proof:}}{\begin{flushright} $ \Box $ \end{flushright}}

\renewcommand{\phi}{\varphi}

\newcommand{\potdiff}{Z}

\newcommand{\as}{\gamma^a}
\newcommand{\bs}{\gamma^b}
\newcommand{\cs}{\gamma^c}
\newcommand{\ds}{\gamma^d}

\newcommand{\ajk}{\gamma_{j,k}^a}
\newcommand{\bjk}{\gamma_{j,k}^b}
\newcommand{\cjk}{\gamma_{j,k}^c}
\newcommand{\djk}{\gamma_{j,k}^d}
\newcommand{\ejk}{\gamma_{j,k}^e}
\newcommand{\fjk}{\gamma_{j,k}^f}

\begin{document}

\title{Derivation of the Time Dependent Gross Pitaevskii Equation with External Fields}


\author{Peter Pickl\footnote{
pickl@itp.phys.ethz.ch
Institute for Theoretical Physics, ETH H\"onggerberg, CH-8093
Z\"urich, Switzerland}}

%
%


%


\maketitle

\begin{abstract}

Using a new method \cite{pickl1} it is possible to derive mean field equations from the microscopic $N$ body Schr\"odinger evolution of interacting particles
 without using BBGKY
 hierarchies.

Recently this method was used to derive the Hartree equation for singular interactions \cite{knowles} and the Gross Pitaevskii equation without positivity condition on the interaction \cite{pickl2} where one had to restrict the  scaling behavior of the interaction.

In this paper more general scalings shall be considered assuming positivity of the interaction.

\end{abstract}
\newpage
\tableofcontents\newpage
\section{Introduction}


In this paper we analyze the dynamics of a Bose condensate
of $N$ interacting particles when the external trap --- described by
an external potential $A_t$ --- is changed, for example removed.

We are interested in solutions of the $N$-particle Schr\"odinger
equation \be\label{schroe} \im\dt \Psi_t = H\Psi_t \ee with
some symmetric (under exchange of any two variables) $\Psi_0$ we shall specify below and the Hamiltonian
\be\label{hamiltonian}
 H=-\sum_{j=1}^N \Delta_j+\sum_{1\leq  j< k\leq  N} V_{\beta}(x_j-x_k) +\sum _{j=1}^N A_t(x_j)
 \ee
acting on the Hilbert space $\LZN$, where $\beta\in\mathbb{R}$
stands for the scaling behavior of the interaction. Note, that $\Psi$ depends on $N$. For ease of notation this shall not be indicated (as well as for many other $N$-dependent objects). The
$V_{\beta}$  scale with the particle number in
such a way, that the total interaction energy is (like the total kinetic energy of the $N$ particles) of order one.

For the moment one may think of an interaction
which is given by $V_{\beta}(x)= N^{-1+3\beta} V(N^\beta x)$ for a
compactly supported, spherically symmetric, positive potential $V\in
L^\infty$. The interactions we shall choose below will be of a more
general form.

The $A_t$ describing the trap potential is a time
dependent external potential which we shall choose
--- in contrast to $V_{\beta}$
--- $N$-independent. Note, that $H$ conserves symmetry, i.e. for
any symmetric function $\Psi_0$ also $H\Psi_0$ and thus $\Psi_t$ is
symmetric.

Assume moreover that the initial wave function $\Psi_0$ is a
condensate in the sense that the reduced one particle marginal density
$$\mu^{\Psi_0}:=\int \Psi^*(\cdot,x_2,\ldots,x_N)\Psi(\cdot,x_2,\ldots,x_N)d^3x_2\ldots d^3x_N$$
converges to $|\phi_0\rangle\langle\phi_0|$ in operator norm.

 Under these and some additional technical assumptions
we shall show that also $\mu^{\Psi_t}$ will be a condensate, i.e.  that
there exist $L^2$ functions $\phi_t$ such that in operator norm
$$\lim_{N\to\infty}\mu^{\Psi_t}=|\phi_t\rangle\langle\phi_t|$$ uniform in $t$ on any
compact subset of $\mathbb{R}^+$ and --- under additional decay
conditions on $\phi_t$
--- uniform in $t\in \mathbb{R}^+$.

In addition we shall show that $\phi_t$ solves the differential
equation \be\label{meanfield}i\frac{d}{dt}
\phi_t=\left(-\Delta+A_t+\overline{V}_{\phi_t}\right)\phi_t\ee with $\phi_0$ as
above, where the ``mean field'' $\overline{V}_{\phi_t}$ depends on $\phi_t$
itself, so (\ref{meanfield}) is a non-linear equation.

 For different
regimes of $\beta$ different effective mean field potentials will
appear: For $\beta=0$ each particle feels
$N^{-1}\sum_{j=2}^{N}V(x_1-x_j)\approx \int V(x-y)|\phi_t|^2(y)d^3y$ interactions
as long as the particles are roughly $|\phi_t|^2$-distributed. Hence the mean field is given by $\overline{V}_{\phi_t}=V\star |\phi_t|^2$. This case is less involved than scalings
$0<\beta\leq  1$, thus it fits best to introduce the new method (see \cite{pickl1}).

For $0<\beta$ the interaction becomes $\delta$-like. To be able to
``average out'' the potential it is important to control the
microscopic structure of $\Psi_t$. Assuming that the energy of
$\Psi_t$ is small, the microscopic structure is --- whenever two
particles approach  --- roughly given by the zero energy scattering
state of the potential $V_{\beta}$. Let us for the moment call this zero energy scattering state $f^N_\beta(x)$.
Changing to coordinates $y=N^\beta x$ the zero energy scattering state satisfies
$$N^{2\beta}(-\Delta+\frac{1}{2}N^{-1+\beta}V(y))f^N_\beta(y)=0\;.$$

For $\beta=1$ the scaling of the potential is such that the zero
energy scattering state $f^N_\beta(x)$ of the potential $V_{\beta}$
just scales with $y$, i.e. $f_1^N(x)=f_1^1(Nx)$. Since $\int V_\beta(x) f^N_\beta(x)dx$ equals $8\pi$ times the scattering length of $V_\beta$ it follows that the mean field is given
by $2a|\phi_t|^2$, where $a/(4\pi)$ is the scattering length of $V$.

The microscopic structure formed by the wave functions enables us to generalize the interactions
 when  $\beta=1$ and $V$ is compactly supported:
Since the scattering length of the potential is always smaller than the radius of the support of the potentials,
the coupling constant of the interaction may grow arbitrarily fast in $N$ in that case. Hence we shall also
consider interactions of the form
\be\label{otherscalings}V_{1,\mu}(x)=N^\mu V(N^{-1}x)\ee
with $\mu> 2$. In this case the wave function avoids the interaction regions and still the scattering length and thus the effect of each interaction
is of order $N^{-1}$.

For $0<\beta<1$ the scaling is ``softer'' and the microscopic
structure disappears as $N\to\infty$. Thus the mean field is given
by $V_{\phi_t}=\|V\|_1|\phi_t|^2$. One can also argue, that for
``soft scalings'' the scattering length is in good approximation
given by the first order Born approximation and thus roughly the $L_1$-norm
of the interaction divided by $8\pi$.

Note that the cases $\beta<0$ and $\beta>1$ are of minor interest: In both cases the interaction becomes negligible. In the case $\beta<0$  the interactions are more or less constant over the support of $\Psi$, in the case $\beta>1$  the radius of the support (and with it the scattering length) of the interaction shrinks faster than $N^{-1}$.  Thus the effective interaction felt by each particle becomes negligible.

A proof for the cases   $0<\beta\leq  1$ without external fields based
on a hierarchical method  analogous to BBGKY hierarchies can be found in \cite{erdos1,erdos2}. The simpler, one dimensional case is treated in \cite{adami}. We
shall give an alternative proof in three dimensions including time dependent external
potentials and generalize to scalings of the form (\ref{otherscalings}). Furthermore we shall prove that the convergence holds uniform in time, assuming that $\phi_t$ shows sufficient decay behavior.

In recent years there has been a growing number of experiments with Bose Einstein condensates where the influence of the mean field has been analyzed (see for example \cite{koehl}). In many of these experiments the condensate propagates while an external field is present, for example in the well known atomic laser experiments the condensates are dropped in the gravitational field. Thus a theoretical understanding of the dynamics of Bose Einstein condensates in external fields is appreciated.

\section{Formulation of the Problem}

\subsection{The new method}

The method we shall  use in this paper is in details explained in  \cite{pickl1}. Heuristically speaking it is based on the idea of counting for each time $t$ the relative number of those particles which are
not in the state $\phi_t$  and estimating the time derivative of that value. To put that onto a rigorous level we need to define some projectors first.
\begin{definition}\label{defpro}
Let $\phi\in\LZ$.
\begin{enumerate}
\item For any $1\leq  j\leq  N$ the
projectors $p_j^\phi:\LZN\to\LZN$ and $q_j^\phi:\LZN\to\LZN$ are given by
\begin{align*} p_j^\phi\Psi=\phi(x_j)\int\phi^*(x_j)\Psi(x_1,\ldots,x)d^3x_j\;\;\;\forall\;\Psi\in\LZN
\end{align*}
and $q_j^\phi=1-p_j^\phi$.

We shall also use the bra-ket notation
$p_j^\phi=|\phi(x_j)\rangle\langle\phi(x_j)|$.
\item
For any $0\leq  k\leq  N$ we define the set $$\mathcal{A}_k:=\{(a_1,a_2,\ldots,a_N): a_j\in\{0,1\}\;;\;
\sum_{j=1}^N a_l=k\}$$ and the orthogonal projector $P_{k}^\phi$
acting on $\LZN$ as
$$P_{k}^\phi:=\sum_{a\in\mathcal{A}_k}\prod_{j=1}^N\big(p_{j}^{\phi}\big)^{1-a_j} \big(q_{j}^{\phi}\big)^{a_j}\;.$$
For
negative $k$ and $k>N$ we set $P_{k}^\phi:=0$.
\item
For any function $m:\mathbb{N}^2\to\mathbb{R}^+_0$ we define the
operator $\widehat{m}^{\phi}:\LZN\to\LZN$ as
\be\label{hut}\widehat{m}^{\phi}:=\sum_{j=0}^N m(j,N)P_j^\phi\;.\ee
We shall also need the shifted operators
$\widehat{m}^{\phi}_d:\LZN\to\LZN$ given by
$$\widehat{m}^{\phi}_d:=\sum_{j=d}^{N+d} m(j+d,N)P_j^\phi\;.$$
\end{enumerate}
\end{definition}



\subsection{Derivation of the Gross-Pitaevskii equation}

This paper deals with the case $0<\beta\leq  1$ only. Then (\ref{meanfield})
becomes the Gross Pitaevskii equation \be\label{GP} \im\frac{d}{dt}
\phi_t=\left(-\Delta +A_t\right) \phi_t+
2a|\phi_t|^2\phi_t:=h^{GP}\phi_t\;.\ee

Following \cite{pickl1} we shall define a functional $\alpha:\LZN\otimes\LZ\to\mathbb{R}^+$ such that
\begin{enumerate}
 \item

$\dt\alpha(\Psi_t,\phi_t)$ can be estimated by $\alpha(\Psi_t,\phi_t)+\landau(1)$,
giving good control of $\alpha(\Psi_t,\phi_t)$ via Gr\o nwall.

\item $\alpha(\Psi,\phi)\to0$
implies convergence the reduced one particle density matrix of $\Psi$ to
$|\phi\rangle\langle\phi|$ in trace norm. 
\end{enumerate}

In the case $\beta=0$ it turned out that the choice $$\alpha(\Psi,\phi)=\llaa\Psi,
\left(\widehat{n}^{\phi}\right)^j\Psi\rraa$$ (again $n(k,N)=\sqrt{k/N}$ and $\laa\cdot\raa$ is scalar product on $\LZN$) for arbitrary $j>0$ does the job (see for example \cite{pickl1} and \cite{knowles}, where the cases $j=2$ respectively $j=1$ are treated for different interactions).

Depending on the particular setting slight adjustments of the functional $\alpha$ are sometimes needed to
get sufficient control of $\dt\alpha(\Psi_t,\phi_t)$. When dealing with interactions which peak very fast as $N$ tends to infinity, adding
a functional which takes care of the smoothness of $\Psi$   proves to be helpful. Doing the estimates
it turns out that one needs that $\|\nabla_1 q^\phi_1\Psi\|$ is small (see Lemma \ref{hnorms} (d)). With the Gr\o nwall argument in mind the first idea one might have is to add precisely this term to
$\alpha$, but on the other hand the time derivative of  $\|\nabla_1 q^\phi_1\Psi\|$ is hard to control. Therefore we add
the difference of the energy per particle of $\Psi$ and the Gross-Pitaevskii-energy of $\phi$ to our functional. It is natural to assume that --- if $\mu^\Psi\to|\phi\rangle\langle\phi|$ ---
this difference is initially small and one expects that during time evolution the energy change per particle of $\Psi$ and the
energy change of $\phi$
are approximately the same.

Therefore we shall need the energy functional $\mathcal{E}:\LZN\to \mathbb{R}$
$$
\mathcal{E}(\Psi)=N^{-1}\laa\Psi,H\Psi\raa\;,$$
as well as the
Gross Pitaevskii energy functional $\mathcal{E}^{GP}:\LZ\to \mathbb{R}$  \begin{align}\label{energyfunct}
\mathcal{E}^{GP}(\phi):=&\langle\nabla\phi,\nabla\phi\rangle+\langle\phi,(A_t+a|\phi|^2)\phi\rangle
=\langle\phi,
(h^{GP}-a|\phi|^2)\phi\rangle\;.
\end{align}
Doing the estimates it turns out that $\|\nabla_1 q^\phi_1\Psi\|$ is small in terms of the energy difference plus $\laa\Psi,\widehat{n}^{\phi}\Psi\raa$.
Therefore we choose $\alpha$ in the following way:
\begin{definition}\label{defalpha}
Let 
 $n(k,N):=\sqrt{k/N}$. We define for any $N\in\mathbb{N}$ the functional
$\alpha:\LZN\times\LZ\to\mathbb{R}^+_0$
$$\alpha(\Psi,\phi):=\laa\Psi,\widehat{n}^\phi\Psi\raa+|\mathcal{E}(\Psi)-\mathcal{E}^{GP}(\phi)|\;.$$
\end{definition}

To get good control of
$\laa\Psi_t,\widehat{n}^{\phi_t}\Psi_t\raa$, the
solutions $\phi_t$ of the Gross Pitaevskii equation we shall
consider have to satisfy some additional conditions.

\begin{definition}\label{defGP}
We define the set of ``good'' solutions of the Gross-Pitaevskii equation
$$\mathcal{G}:=\{\phi_t:\im\frac{d}{dt} \phi_t=h^{GP}\phi_t;\;\|\phi_t\|_\infty+\|\Delta\phi_t\|<\infty\;\forall\;t\geq0\}\;.$$
\end{definition}
Furthermore we shall --- depending on $\beta$ --- need some
conditions on the interaction $V_{\beta}$. These conditions shall
include the potentials we used in the introduction, i.e. potentials
which scale like $V_\beta (x)= N^{-1+3\beta} V(N^\beta x)$ as well as scalings of the form (\ref{otherscalings}) for
compactly supported, spherically symmetric, positive potentials $V\in
 L^\infty$.

\begin{definition}\label{defpot}
Let $a>0$. For any $0<\beta\leq 1$ we define the auxiliary set
$$\mathcal{U}_\beta:=\{V_{\beta}\text{ pos. and spher. symm., } V_{\beta}(x)=0\;\forall\;x>R N^{-\beta}\text{ for some }R<\infty \}$$
as well as the set of potentials with appropriate scaling behavior for $0<\beta<1$
\begin{align*}\mathcal{V}_{\beta}:=\{V_{\beta}\in
\mathcal{U}_\beta:
&\lim_{N\to\infty}N^{1-3\beta}\|V_{\beta}\|_\infty<\infty;
 \\&\lim_{N\to\infty}N^{\eta}|\;\|NV_{\beta}\|_1-2a|<\infty\text{ for some
}\eta>0\}\;, \end{align*} and for $\beta=1$
$$\mathcal{V}_{1}:=\{V_{1}\in \mathcal{U}_1: \lim_{N\to\infty}N^{\eta}|4\pi N scat(V_{1})-a|<\infty \text{
for some
}\eta>0\}\;,$$ where $scat(V)$ is the scattering length of the
potential $V$.
\end{definition}

With these definitions we arrive at the main Theorem:

\begin{theorem}\label{theorem}

Let $0<\beta\leq 1$, let   $V_{\beta}\in\mathcal{V}_{\beta}$, let $A_t$ be an external potential with $\sup_{x\in\mathbb{R}^3, t\in\mathbb{R}}|\dot A_t|<\infty$. Let
$\phi_t\in\mathcal{G}$ and $\Psi_0$ be symmetric with $\|\Psi_0\|=1$.
 Then there exists a $\eta>0$ and constants $C_1,C_2<\infty$ such that
\be\label{theoremeq}\alpha(\Psi_t,\phi_t)\leq  C_1e^{C_2(\ln N)^{1/3}\int_0^t
\|\phi_s\|_\infty+\|\nabla\phi_s\|_{6,loc}+\|\dot{A}_s\|_\infty ds}\left(\alpha(\Psi_0,\phi_0)+  N^{-\eta}\right)\;,\ee
where  $\|\cdot\|_{6,loc}:\LZ\to\mathbb{R}^+$ is the ``local $L^6$-norm'' given by $$\|\phi\|_{6,loc}:=\sup_{x\in\mathbb{R}^3}\|\mathds{1}_{|\cdot-x|\leq  1} \phi\|_6\;.$$
\end{theorem}

\begin{remark}
\begin{enumerate}

\item Lieb, Seiringer and Yngvason have proven that for the ground state $\Psi^{gs}$ of a trapped Bose gas and the ground state $\phi^{gs}$ of the respective Gross-Pitaevskii energy functional $\mathcal{E}(\Psi^{gs})-\mathcal{E}^{GP}(\phi^{gs})\to0$ as $N\to\infty$ \cite{lsy}. In \cite{ls} Lieb and Seiringer show that $\mu^{\Psi^{gs}}\to|\phi^{gs}\rangle\langle\phi^{gs}|$. Hence for the ground state of a trapped Bose gas $\lim_{N\to\infty}\alpha(\Psi^{gs},\phi^{gs})=0$.

\item For all
$\eta>0$ one can find a $N>0$ such that $(\ln N)^{1/3}<\eta \ln N$. Thus $e^{(\ln N)^{1/3}}\leq  C e^{\eta\ln N}=CN^\eta$, so if $\alpha(\Psi_0,\phi_0)\leq  C N^{\eta}$ for some $\eta>0$  and if $\int_0^t \|\phi_s\|_\infty+\|\nabla\phi_s\|_{6,loc}+\|\dot{A}_s\|_\infty ds<\infty$
it follows that the right hand side of (\ref{theoremeq}) is small.

\item Using Sobolev $\|\nabla\phi_s\|_{6,loc}\leq \|\nabla\phi_s\|_{6}\leq \|\Delta\phi\|$. Thus $\|\nabla\phi_s\|_{6,loc}$ can be bounded by the square of the Gross-Pitaevskii Energy. 
     
     On the other hand $\|\nabla\phi_s\|_{6,loc}\leq \|\nabla\phi\|_\infty$. Since we are in the defocussing regime one expects when the potential is turned of  that $\|\phi\|_\infty$ and $\|\nabla\phi\|_\infty$ decay like $t^{-3/2}$. Whenever $\int_0^\infty \|\phi_s\|_\infty+\|\nabla\phi_s\|_{6,loc}+\|\dot{A}_s\|_\infty ds<\infty$ the right hand side of (\ref{theoremeq}) is small uniform in $t$.

\item

It has been shown in \cite{pickl1} that $$\lim_{N\to\infty}\laa\Psi,\widehat{n}^{\phi}\Psi\raa=0$$  implies weak convergence of
the reduced one particle density matrix of $\Psi$ against $|\phi\rangle\langle\phi|$ and vice versus. For other equivalent definitions of asymptotic 100\% condensation see \cite{michelangeli}.

\item

 The set $\mathcal{V}_1$ includes potentials with scalings of the form (\ref{otherscalings}).




\end{enumerate}

\end{remark}

\subsection{Skeleton of the Proof}

We shall prove the Theorem via  Gr\o nwall, so
our goal is to show that there exists a $\eta>0$ such that
\be\label{groen}\dt\alpha(\Psi_t,\phi_t)\leq  C(\alpha(\Psi_t,\phi_t)+N^{-\eta})\;. \ee
Therefore we shall define a functional $\alpha':\LZN\otimes\LZN\to\mathbb{R}$ such that $\dt\alpha(\Psi_t,\phi_t)\leq \alpha'(\Psi_t,\phi_t)$. It is convenient to split up $\alpha'=\alpha'_0+\alpha'_1+\alpha'_2$ and treat these summands separately (see Definition \ref{alphasplit} and Lemma \ref{ableitung}).  Then we will show that we can find a respective bound for
$\alpha'(\Psi_t,\phi_t)$. A nice feature of the method we use is that we can avoid propagation estimates on $\Psi_t$ to get (\ref{groen}): Similar as in in \cite{pickl1}
one can estimate the functional $\alpha'(\Psi,\phi)$ uniform in $\Psi$ and $\phi$ in terms of $\alpha(\Psi,\phi)$
and $N^{-\eta}$ times some polynomial in $\|\phi\|_\infty$, $\|\nabla\phi\|_{6,loc}$ and $\|\Delta\phi\|<\infty$. Under the assumption
$\phi_t\in\mathcal{G}$ we get (\ref{groen}).

The proof is organized as follows:
\begin{enumerate}

\item

The respective estimates of the $\alpha'_{0,1,2}(\Psi,\phi)$ shall be given in  section \ref{secb1}. The procedure is similar
as in \cite{pickl1}.
It turns out that
\begin{itemize}
\item We get good control of $\alpha'_0$ for all $0<\beta\leq 1$.
 \item We get sufficient control of $\alpha_1$ for $\beta<1/3$, only.
\item  For $\alpha_2$ some of the estimates are in terms of $\|\nabla_1q_1\Psi\|$ and some estimates require that $\beta<1$.
\end{itemize}

So the next step will be to show that $\|\nabla_1q_1\Psi\|$ is small:
For later reference we shall give in section \ref{secsmooth} an estimate of the interaction energy and an implicit estimate on $\|\nabla_1q_1\Psi\|$
which holds for all $0<\beta\leq  1$.
The result will be used in section \ref{secabl} to control $\|\nabla_1q_1\Psi\|$ in terms of $\alpha(\Psi,\phi)$
and $N^{-\eta}$ for some $\eta>0$ under the restriction $0<\beta<1$.

This enables us to finish the proof of the Theorem for $\beta<1/3$ using Gr\o nwall (section \ref{secproof1}).

\item
After that we   generalize the proof of the Theorem to the case $\beta<1$. We already have good control of $\alpha'_0$ and $\alpha'_2$. To make  $\alpha'_1$ controllable we use the microscopic structure  to adjust $\alpha$ in such a way that the respective adjusted
$\alpha'_1$ is controllable.
Therefore we need some estimates on the microscopic structure of the wave function.
These are given in section \ref{secmic}. In section \ref{secadj1} we adjust $\alpha$ and prove, that the adjustment in fact changes the respective $\alpha_1'$ such that it is controllable for all $0<\beta\leq  1$.
Then we complete in section \ref{secproof3} the proof of the Theorem for $0<\beta<1$.

\item Our final goal is to treat the case $\beta=1$. To be able to use our results of the previous sections, we have to generalize
our estimates on $\|\nabla_1q_1\Psi\|$ to the case $\beta=1$ first. It turns out that $\|\nabla_1q_1\Psi\|$ is in fact
not small for $\beta=1$: Some non-negligible part of the kinetic energy is used to build up the microscopic structure
in that case. Nevertheless we are able to control the kinetic energy of $q_1\Psi$ outside some small set around the
positions of the other particles (section \ref{secabl2}).

In the next section we show that this new estimate is in fact sufficient to recover our old estimates, in particular
Lemma \ref{hnorms} (d).

Similar as in (b) we now make another adjustment of $\alpha$ using again the microscopic structure. We adjust $\alpha$ in such a way that the respective $\alpha'_2$ is controllable also for $\beta=1$ (section \ref{secadj2}).

Finally we complete the proof of the Theorem (section \ref{secproof3}).

\end{enumerate}

\section{Preliminaries}

\begin{notation}
\begin{enumerate}
 \item
Throughout the paper hats $\;\widehat{\cdot}\;$ shall solemnly be
used in the sense of Definition \ref{defpro} (c). The label $n$ shall always be used for the function $n(k,N)=\sqrt{k/N}$.
\item
In the following we shall omit the upper index $\phi$ on $p_j$,
$q_j$, $P_j$, $P_{j,k}$ and $\widehat{\cdot}$. It shall be replaced exclusively in a few formulas where their $\phi$-dependence plays an important role.
\item
We shall need the operator $H^{GP}:=\sum_{j=1} h_j^{GP}$, where
$h_j^{GP}$ is the Gross Pitaevskii (\ref{GP}) operator acting on the
$j^{\text{th}}$ particle.

\item
In our estimates below we shall need  the operator norm $\|\cdot\|_{op}$ defined for any linear operator $f:\LZN\to\LZN$ by
$$\|f\|_{op}:=\sup_{\|\Psi\|=1}\|f\Psi\|\;.$$

\item
Constants appearing in estimates will generically be denoted by $C$. We shall not distinguish constants
 appearing in a sequence of estimates, i.e. in $X\leq  CY\leq  CZ$ the constants may differ.

\end{enumerate}

\end{notation}

First we need some properties of the objects defined in Definition \ref{defpro}

\begin{lemma}\label{kombinatorik}
\begin{enumerate}

\item For any weights $m,r:\mathbb{N}^2\to\mathbb{R}^+_0$ we have
that
$$\widehat{m}\widehat{r}\,=\widehat{mr}=\widehat{r}\,\widehat{m}\;\;\;\;\;\;\;\;\;\;\widehat{m}p_j=p_j\widehat{m}\;\;\;\;\;\;\;\;\;\;\widehat{m}P_{k}=P_{k}\widehat{m}\;.$$
\item Let $n:\mathbb{N}^2\to\mathbb{R}^+_0$ be given by $n(k,N):=\sqrt{k/N}$.
Then the square of $\widehat{n}$ (c.f. (\ref{hut}))
equals the relative particle number operator of particles not in the
state $\phi$, i.e.
\be\label{partnumber}\left(\widehat{n}\right)^2=N^{-1}\sum_{j=1}^Nq_j\;.\ee

\item For any weight $m:\mathbb{N}^2\to\mathbb{R}^+_0$ and any function $f:\mathbb{R}^6\to\mathbb{R}$ and any
$j,k=0,1,2$  $$\widehat{m} Q_j f(x_1,x_2)Q_k=
Q_j f(x_1,x_2)\widehat{m}_{j-k}Q_k\;,$$ where
$Q_0:=p_1 p_2$, $Q_1\in\{p_1q_2,q_1p_2\}$ and
$Q_2:=q_1q_2$.

\item For any weight $m:\mathbb{N}^2\to\mathbb{R}^+_0$ and any function $f:\mathbb{R}^6\to\mathbb{R}$
$$[f(x_1,x_2),\widehat{m}]=\left[f(x_1,x_2),p_1p_2(\widehat{m}-\widehat{m}_2)+p_1q_2(\widehat{m}-\widehat{m}_1)+q_1p_2(\widehat{m}-\widehat{m}_1)\right]$$

\item Let $f\in L^1$, $g\in L^2$, $h\in L^3$ with $h(x)=0$ for all $|x|>0$.
\begin{align}\label{kombeqa}\|p_j f(x_j-x_k)p_j\|_{op}
\leq&  \|f\|_1\|\phi\|_\infty^2\;,
\\ \label{kombeqb}\|g(x_j-x_k)p_j\|_{op}\leq&  \|g\|_2\|\phi\|_\infty
\;,
\\\label{kombeqc}\|h(x_j-x_k)\nabla p_j\|_{op}\leq&  \|h\|_3\|\nabla\phi\|_{6,loc}
\;.\end{align}

\end{enumerate}
\end{lemma}

\begin{proof}\begin{enumerate}
 \item
follows immediately from Definition \ref{defpro}, using that $p_j$
and $q_j$ are orthogonal projectors.

\item Note that $\cup_{k=0}^N\mathcal{A}_k=\{0,1\}^N$, so $1=\sum_{k=0}^N P_k$. Using also
$(q_k)^2=q_k$ and $q_k p_k=0$ we get \begin{align*}
N^{-1}\sum_{k=1}^Nq_k=N^{-1}\sum_{k=1}^Nq_k\sum_{j=0}^N
P_j= N^{-1}\sum_{j=0}^N\sum_{k=1}^Nq_k
P_j=N^{-1}\sum_{j=0}^Nj P_j\end{align*} and (b) follows.

\item
Using the definitions above we have  \begin{align*} \widehat{m}
Q_j f(x_1,x_2)Q_k
=&\sum_{l=0}^N m(l)P_l Q_jf(x_1,x_2)Q_k
\end{align*}
The number of projectors $q_j$ in $P_l Q_j$ in the coordinates $j=3,\ldots,N$ is equal to $l-j$. The $p_j$ and $q_j$ with $j=3,\ldots,N$ commute with $Q_jf(x_1,x_2)Q_k$. Thus
$P_l Q_jf(x_1,x_2)Q_k= Q_jf(x_1,x_2)Q_kP_{l-j+k}$ and
\begin{align*}
\widehat{m}
Q_j f(x_1,x_2)Q_k=& \sum_{l=0}^N  m(l) Q_jf(x_1,x_2)Q_kP_{l-j+k}
\\&\hspace{-3cm}= \sum_{l=k-j}^{N+k-j}  Q_jf(x_1,x_2)m(l+j-k)P_{l} Q_k
=Q_j f(x_1,x_2)\widehat{m}_{j-k}Q_k\;.
 \end{align*}

\item First note that
\begin{align}\label{multsev}
&\hspace{-1cm}[f(x_1,x_2),\widehat{m}]-\left[f(x_1,x_2),p_1p_2(\widehat{m}-\widehat{m}_2)+p_1q_2(\widehat{m}-\widehat{m}_1)
+q_1p_2(\widehat{m}-\widehat{m}_1)\right]%
\nonumber\\=&[f(x_1,x_2),q_1q_2\widehat{m}]+\left[f(x_1,x_2),p_1p_2\widehat{m}_2+p_1q_2\widehat{m}_1
+q_1p_2\widehat{m}_1\right]%
\;.
\end{align}
We shall show that the right hand side is zero.
Multiplying the right hand side with $p_1p_2$ from the left one gets
\begin{align*}
&p_1p_2f(x_1,x_2)q_1q_2\widehat{m}+p_1p_2f(x_1,x_2)p_1p_2\widehat{m}_2-p_1p_2\widehat{m}_2f(x_1,x_2)
\\&+p_1p_2f(x_1,x_2)p_1q_2\widehat{m}_1+p_1p_2f(x_1,x_2)q_1p_2\widehat{m}_1%
\end{align*}
Using (c) the latter is zero.
Multiplying (\ref{multsev}) with $p_1q_2$ from the left one gets
\begin{align*}
&p_1q_2f(x_1,x_2)q_1q_2\widehat{m}+p_1q_2f(x_1,x_2)p_1p_2\widehat{m}_2+
p_1q_2f(x_1,x_2)p_1q_2\widehat{m}_1\\&+p_1q_2f(x_1,x_2)q_1p_2\widehat{m}_1-p_1q_2\widehat{m}_1f(x_1,x_2)
\end{align*}
Using (c) the latter is zero. Also multiplying with $q_1p_2$ yields zero due to symmetry in interchanging
$x_1$ with $x_2$.
Multiplying (\ref{multsev}) with $q_1q_2$ from the left one gets
\begin{align*}
&q_1q_2f(x_1,x_2)\widehat{m}q_1q_2-q_1q_2\widehat{m}f(x_1,x_2)+
q_1q_2f(x_1,x_2)p_1p_2\widehat{m}_2+\\&q_1q_2f(x_1,x_2)p_1q_2\widehat{m}_1+q_1q_2f(x_1,x_2)q_1p_2\widehat{m}_1
\end{align*}
which is again zero, thus (\ref{multsev}).

\item To show (\ref{kombeqa}) we use the notation $p_j=|\phi(x_j)\rangle\langle\phi(x_j)|$
\begin{align}\nonumber p_j f(x_j-x_k)p_j=& |\phi(x_j)\rangle\langle \phi(x_j)|f(x_j-x_k)|\phi(x_j)\rangle\langle\phi(x_j)|
\\&\hspace{-3cm}=\;\;\;|\phi(x_j)\rangle\rangle(f\star|\phi|^2)(x_k)\langle\phi(x_j)|=p_j (f\star|\phi|^2)(x_k)\;.
\label{faltungorigin}
\end{align}
It follows that $$\|p_j f(x_j-x_k)p_j\|_{op}\leq  \|f\|_1\|\phi\|_\infty^2\;.$$ With Young we get (\ref{kombeqa}).

For (\ref{kombeqb}) we write
\begin{align*}
\|g(x_j-x_j)p_j\|_{op}^2=&\sup_{\|\Psi\|=1}\|g(x_j-x_j)p_j\Psi\|^2=
\\=&\sup_{\|\Psi\|=1}\laa \Psi,p_j  g(x_j-x_j)^2 p_j\Psi\raa\\\leq& \|p_j  g(x_j-x_j)^2p_j\|_{op}\;.
\end{align*}
With (\ref{kombeqa}) we get (\ref{kombeqb}).
For (\ref{kombeqc}) we have using Young's inequality
\begin{align*}
\|h(x_j-x_k)\nabla p_j\|_{op}\leq&  \sup_{y\in\mathbb{R}^3}|\langle \nabla\phi,h^2(\cdot-y)\nabla\phi\rangle|^{1/2}
\\\leq&  \sup_{y\in\mathbb{R}^3}\|h^2\|^{1/2}_{3/2}\|\mathds{1}_{|\cdot-y|\leq 1}|\nabla\phi|^2\|^{1/2}_{3}\\=&\|h\|_3\|\nabla\phi\|_{6,loc}\;.
\end{align*}
\end{enumerate}
\end{proof}
When doing the estimates we will encounter wave functions where some of the symmetry is broken (at this point the
reader should exemplarily think of the wave function
$V_\beta(x_1-x_2)\Psi$ which is not symmetric under exchange of the variables $x_1$ and $x_3$ for example). Therefore we want to formulate some of
our results for wave functions which are not symmetric under exchange of any two variables $x_j$, $x_k$. This leads to the following
definition

\begin{definition}
We define for any finite set $\mathcal{M}\subset\mathbb{N}$ the space $\mathcal{H}_{\mathcal{M}}\subset\LZN$ of
functions which are symmetric in all variables but those in $\mathcal{M}$
\begin{align*}\Psi\in \mathcal{H}_{M}\Leftrightarrow& \Psi(x_1,\ldots,x_j,\ldots,x_k,\ldots,x_N)=\Psi(x_1,\ldots,x_k,\ldots,x_j,\ldots,x_N)\\&\text{ for all } j,k\notin\mathcal{M}\;.
\end{align*}
and the operator norm $\|\cdot\|_{\mathcal{M}}$ on $\mathcal{H}_{\mathcal{M}}\to\LZN$ by
$$\|A\|_{\mathcal{M}}:=\sup_{\Psi,\chi\in\mathcal{H}_{\mathcal{M}};\|\Psi\|=\|\chi\|=1}\laa\chi,A\Psi\raa\;.$$

\end{definition}

With Definition \ref{defpro} we arrive directly at the following Lemma
based on combinatorics of the $p_j$ and $q_j$:
\begin{lemma}\label{kombinatorikb}

For any $f:\mathbb{N}^2\to\mathbb{R}^+_0$ and any finite set $\mathcal{M}_a\subset\mathbb{N}$ with $1\notin\mathcal{M}_a$
and any finite set $\mathcal{M}_b\subset\mathbb{N}$ with $1,2\notin\mathcal{M}_b$ there exists a $C>\infty$ such that
\begin{align}\label{komb1}
\left\| \widehat{f} q_1\Psi\right\|^2\leq& C
\|\widehat{f}\widehat{n}\Psi\|^2\;\;\;\;\;\;\;\;\;\;\text{ for any }\Psi\in\mathcal{H}_{\mathcal{M}_a},\;N>|\mathcal{M}_a|\\
\label{komb2} \left\| \widehat{f}
q_1 q_2\Psi\right\|^2\leq& C
\|\widehat{f}(\widehat{n})^2\Psi\|^2\;\;\;\;\;\text{ for any }\Psi\in\mathcal{H}_{\mathcal{M}_b},\;N>|\mathcal{M}_b|\;.\end{align}

\end{lemma}

\begin{proof}
Let $\Psi\in\mathcal{H}_{\mathcal{M}_a}$ for some finite set $1\in\mathcal{M}_a\subset\mathbb{N}$.
For (\ref{komb1}) we can write using symmetry of $\Psi$ and Lemma \ref{kombinatorik} (b)
\begin{align*}\|\widehat{f}\widehat{n}\Psi\|^2
=&\laa\Psi,(\widehat{f})^{2}(\widehat{n})^2\Psi\raa
=N^{-1}\sum_{k=1}^N\laa\Psi,(\widehat{f})^{2}q_k\Psi\raa
\\\geq&  N^{-1}\sum_{k\notin\mathcal{M}_a}\laa\Psi,(\widehat{f})^{2}q_k\Psi\raa
=\frac{N-|\mathcal{M}_a|}{N}\laa\Psi,(\widehat{f})^{2}q_1\Psi\raa
\\=&\frac{N-|\mathcal{M}_a|}{N}\|\widehat{f}
q_1\Psi\|^2\;.
 \end{align*}
Similarly we have for $\Psi\in\mathcal{H}_{\mathcal M_b}$
 \begin{align*} \|\widehat{f}
(\widehat{n})^2\Psi\|^2 =&\laa\Psi,(\widehat{f}
)^2(\widehat{n})^4\Psi\raa
\geq N^{-2}\sum_{j,k\notin\mathcal{M}_b}\laa\Psi,(\widehat{f} )^2q_j
q_k\Psi\raa
\\=&\frac{N-|\mathcal M|(N-|\mathcal M|-1)}{N^2}\laa\Psi,(\widehat{f} )^2q_1
q_2\Psi\raa+\frac{|\mathcal M|}{N^2}\laa\Psi,(\widehat{f}
)^2q_1\Psi\raa
\\\geq& \frac{N-|\mathcal M|(N-|\mathcal M|-1)}{N^2} \|\widehat{f} q_1q_2\Psi\|
 \end{align*}
and the Lemma follows.
\end{proof}
Our next step is to define the functionals $\alpha'_{j}$, $j=0,1,2$ which, as explained above, control the time derivative of $\alpha(\Psi_t,\phi_t)$.
\begin{definition}\label{alphasplit}
Using the notation
$$\potdiff_\beta(x_j,x_k):=V_{\beta}(x_j-x_k)-\frac{2a}{N-1}|\phi|^2(x_j)-\frac{2a}{N-1}|\phi|^2(x_k)$$
we define  functionals $\alpha_{0,1,2}':\LZN\to\mathbb{R}^+$ by
\begin{align}\label{fnochdao} \alpha_{0}'(\Psi,\phi)%
=&\left|\laa\Psi_t,\dot A\Psi_t\raa-\langle\phi,\dot A\phi\rangle\right|\\
\label{fnochda} \alpha_{1}'(\Psi,\phi)
=&2N(N-1)\Im\left(\laa\Psi ,\potdiff_\beta(x_1,x_2) p_1p_2(\widehat{n}-\widehat{n}_{2})
\Psi\raa\right) \\
\label{fnochda2}
\alpha_{2}'(\Psi,\phi)%
=&4N(N-1)\Im\left(\laa\Psi ,\potdiff_\beta(x_1,x_2) p_1q_2(\widehat{n}-\widehat{n}_1)
\Psi\raa\right) \;.\end{align}

\end{definition}

\begin{lemma}\label{ableitung}

For any solution of the Schr\"odinger equation $\Psi_t$ and any
solution of the Gross-Pitaevskii equation $\phi_t$ we have
\be\label{lemmaableitungeq}\left|\frac{d}{dt} \alpha(\Psi_t,\phi_t)\right|\leq  \sum_{j=0}^2\alpha'_j(\Psi_t,\phi_t)\;.\ee

\end{lemma}

\begin{proof} For the proof of the Lemma we shall restore the upper index $\phi_t$ to pay respect to the time dependence of
$\widehat{n}^{\phi_t}$.
We have for the time derivative of the first summand of $\alpha$
\begin{align*}\frac{d}{dt}\laa\Psi_t,\widehat{n}^{\phi_t}\Psi_t\raa=&-i\laa H\Psi_t
,\widehat{n}^{\phi_t}\;\Psi_t\raa+i\laa\Psi_t
,\widehat{n}^{\phi_t}\;H\Psi_t\raa\\&+i\laa \Psi_t
,[H^{GP}_t,\widehat{n}^{\phi_t}\;]\Psi_t\raa
\nonumber\\=&-i\laa\Psi_t
,[H-H^{GP}_t,\widehat{n}^{\phi_t}\;]\Psi_t\raa\\=&-iN(N-1)\laa\Psi_t
,[\potdiff_\beta(x_1,x_2),\widehat{n}^{\phi_t}\;]\Psi_t\raa
\;.
\end{align*}

Using Lemma \ref{kombinatorik} (d) it follows that the latter equals
\begin{align*}&-iN(N-1)\laa\Psi_t
,[\potdiff_\beta(x_1,x_2),p_1p_2(\widehat{n}^{\phi_t}-\widehat{n}^{\phi_t}_2)]\Psi_t\raa
\\
&-2iN(N-1)\laa\Psi_t
,[\potdiff_\beta(x_1,x_2),p_1q_2(\widehat{n}^{\phi_t}-\widehat{n}^{\phi_t}_1)]\Psi_t\raa\;.
\end{align*}
Since $\potdiff_\beta$ is selfadjoint this is $\alpha_{1}'+\alpha_{2}'$.

For the second summand of $\alpha$ we have
\begin{align}\label{enbound2}&
\nonumber\frac{d}{dt}\left(\mathcal{E}(\Psi_t)-\mathcal{E}^{GP}(\phi_t)\right)
=\laa\Psi_t,\dot A_t(x_1)\Psi_t\raa-\langle\phi_t,a\left(\frac{d}{dt}|\phi_t|^2\right)\phi_t\rangle
\\\nonumber&-\langle\phi_t,\dot A_t\phi_t\rangle-\langle\phi_t,[(h^{GP}-a|\phi_t|^2),h^{GP}]\phi_t\rangle
\\\nonumber=&\laa\Psi_t,\dot A_t(x_1)\Psi_t\raa-\langle\phi_t,\dot A_t\phi_t\rangle
+\langle\phi_t,[a|\phi_t|^2,h^{GP}]\phi_t\rangle\\&-\langle\phi_t,[a|\phi_t|^2,h^{GP}]\phi_t\rangle\;.
\end{align}
Hence \be\label{ablediff}\dt\left|\mathcal{E}(\Psi_t)-\mathcal{E}^{GP}(\phi_t)\right|\leq \alpha'_{0}(\Psi,\phi)\ee
which proves the Lemma.

\end{proof}

\section{Control of the $\alpha'$ for $\beta<1/3$}\label{secb1}

As a first step we shall prove the Theorem for  scalings $\beta<1/3$. It is not surprising that this case is special: $\beta<1/3$ means that
the mean distance of two particle is much smaller than the radius of the support of the interactions as $N\to\infty$. Thus we are in a regime
where for $|\Psi|^2$-typical  configurations   many of the interactions overlap and one arrives directly at a mean field picture.

Our goal is now to control the functionals $\alpha'_{0,1,2}$ in such a way, that we can conclude that $\alpha(\Psi_t,\phi_t)$ is small via Gr\o nwall. Remember that for the $\phi$'s we are interested in   $\|\phi\|_\infty$, and $\|\nabla\phi\|_{6,loc}\geq\|\Delta\phi\|$ are bounded (see Definition \ref{defGP}). Hence it is sufficient to estimate
$\alpha'_{j}$, $j=0,1,2$ in terms of $\alpha+N^{-\eta}$ for some $\eta>0$ times an arbitrary polynomial in $\|\phi\|_\infty$, $\|\nabla\phi\|_{6,loc}$ and $\|\Delta\phi\|$. So we define
\begin{definition}\label{defK}
The set $\mathcal{F}$ of functionals $\LZ\to\mathbb{R}^+$ is given by
\begin{align*}\mathcal{K}\in\mathcal{F}\Leftrightarrow&\text{ there exists a polynomial }K:(\mathbb{R}^+)^3\to\mathbb{R}^+\text{ such that }
\\&\mathcal{K}(\phi)=K(\|\phi\|_\infty,\|\Delta\phi\|)\;.\end{align*}
\end{definition}

And we want to show that for some $\kinf$, some $\eta>0$ and $j=0,1,2$
$$\alpha'_{j}\leq  \mathcal{K}(\phi)(\alpha(\Psi,\phi)+N^{-\eta})$$
uniform in $(\Psi,\phi)\in\LZN\otimes\LZ$.

Rewrite $\alpha_{1,2}'$ multiplying $\potdiff_\beta$ with  $1=p_1p_1+p_1q_2+q_1p_2+q_1q_1$ from the right.
Lemma \ref{kombinatorik} (c) shows that $$p_1p_2\potdiff_\beta(x_1,x_2) p_1p_2(\widehat{n}-\widehat{n}_{2})=p_1p_2(\widehat{n}-\widehat{n}_{2})\potdiff_\beta(x_1,x_2) p_1p_2$$
is selfadjoint and so is $p_1q_2\potdiff_\beta(x_1,x_2) p_1q_2(\widehat{n}-\widehat{n}_1)$.\\The operator $q_1p_2\potdiff_\beta(x_1,x_2) p_1q_2(\widehat{n}-\widehat{n}_1)$ is invariant under adjunction plus simultaneous exchange of the variable $x_1$ and $x_2$. Thus the sandwiches with $\Psi$ of the respective operators are real and using Lemma \ref{kombinatorik} (c)
\begin{align}
\label{alpha1} \alpha_{1}'(\Psi,\phi)
=&4N(N-1)\Im\left(\laa\Psi ,p_1q_2(\widehat{n}_{-1}-\widehat{n}_{1})\potdiff_\beta(x_1,x_2) p_1p_2
\Psi\raa\right) \\
\label{alpha1b} &+2N(N-1)\Im\left(\laa\Psi ,q_1q_2(\widehat{n}_{-2}-\widehat{n})\potdiff_\beta(x_1,x_2) p_1p_2
\Psi\raa\right) \\
\label{alpha2}
\alpha_{2}'(\Psi,\phi)%
=&4N(N-1)\Im\left(\laa\Psi ,p_1p_1\potdiff_\beta(x_1,x_2) p_1q_2(\widehat{n}-\widehat{n}_1)
\Psi\raa\right)
\\\label{alpha2b}&+4N(N-1)\Im\left(\laa\Psi ,q_1q_2(\widehat{n}_{-1}-\widehat{n})\potdiff_\beta(x_1,x_2) p_1q_2
\Psi\raa\right)
\;.
\end{align}

The task of this section is to estimate all the terms on the right hand sides of (\ref{alpha1}) and (\ref{alpha2})
as well as $\alpha_0'$. This will be done in Lemma \ref{hnorms} below, but let us first give some heuristic arguments why they are small

\begin{itemize}
  \item From a physical point of view (\ref{alpha1}) and (\ref{alpha2}) are the most important. Here we use that in leading order the interaction and the mean field cancel out. Note first that one of the mean field parts in $\potdiff_\beta$ is zero: $p_jq_j=0$, thus $p_1q_1a|\phi(x_1)|p_1p_2=0$.
For the interaction part in $\potdiff_\beta$ we use formula (\ref{faltungorigin}):
$p_1V_\beta(x_1-x_2)p_1=V_\beta\star|\phi|^2(x_2)p_1$.
Since $V_\beta$ is $\delta$-like and its integral is  $a/N$ the latter is $\approx a|\phi|^2(x_2)p_1$, cancelling out
most of the mean field part in $\potdiff_\beta$. Thus
 (\ref{alpha1}) and (\ref{alpha2}) are small.

\item
Since there is neither a $p_1$ nor $p_2$ on the left side of $V_\beta$ in (\ref{alpha1b})  the latter seems at first view to grow with $N$.
It is indeed not small for general non-symmetric normalized $\Psi$: If all the mass of $\Psi$ was concentrated in an area where $x_1\approx x_2$ (which is of course not possible for symmetric $\Psi$), then (\ref{alpha1b}) would in fact grow with $N$. So to estimate (\ref{alpha1b}) we have to use symmetry of $\Psi$. The trick is to estimate $$2N\Im\left(\laa (\widehat{n}_{-2}-\widehat{n})q_1\Psi ,\sum_{j=2}^Nq_j\potdiff_\beta(x_1,x_j) p_1p_j
\Psi\raa\right) $$
which is for symmetric $\Psi$ equal to (\ref{alpha1b}). Note that in view of Lemma \ref{kombinatorik} (b) $(\widehat{n}_{-2}-\widehat{n})q_1\Psi$ is of order $N^{-1}$. Using Cauchy Schwarz we have to control
\begin{align*}\|\sum_{j=2}^Nq_j\potdiff_\beta(x_1,x_j) p_1p_j
\Psi\|^2=\sum_{2\leq  j\leq  N}\laa\Psi, p_1p_j\potdiff_\beta(x_1,x_j)q_j\potdiff_\beta(x_1,x_j) p_1p_j\Psi\raa
\\+2\sum_{2\leq  k<j\leq  N}\laa\Psi, p_1p_k\potdiff_\beta(x_1,x_k)q_kq_j\potdiff_\beta(x_1,x_j) p_1p_j\Psi\raa
\end{align*}
The first line has only $N$ summands. Since $\|V_\beta\|_1$ is of order $N^{-1}$ this line is small if $\|V_\beta\|_\infty\ll1$ which is the case for $\beta<1/3$.

For the second line we can write
\begin{align*}\sum_{2\leq  k<j\leq  N}\laa \sqrt{\potdiff_\beta(x_1,x_k)}p_k\sqrt{\potdiff_\beta(x_1,x_j)}p_1q_j\Psi,\\ \sqrt{\potdiff_\beta(x_1,x_j)}p_j \sqrt{\potdiff_\beta(x_1,x_k)}p_1q_k\Psi\raa\;.
\end{align*} Now we have enough projectors $p$ on both sides of the interactions to be able to integrate them against $\phi$ (c.f.
  Lemma \ref{kombinatorik} (e)) and it is clear that it at least does not grow with $N$.
Below we shall show that for $\beta<1/3$ (\ref{alpha1b}) and (\ref{alpha2}) are in fact bounded by $\alpha+N^{-\eta}$ for some $\eta>0$.

\item To show that (\ref{alpha2b}) is small
one needs to use smoothness of $\Psi$.
We do so in the following way: We introduce a potential $U_{\beta_1,\beta}$ with moderate scaling behavior and the same $L^1$ norm as $V_\beta$. This is done in definition \ref{udef}. The scaling $\beta_1$ of $U_{\beta_1,\beta}$ will be chosen such that (\ref{alpha2b}) with $V_\beta$ replaced by $U_{\beta_1,\beta}$ can be controlled. The difference ---  (\ref{alpha2b}) with $V_\beta$ replaced by $V_\beta-U_{\beta_1,\beta}$ --- we integrate by parts. It follows that (\ref{alpha2b}) can be controlled in terms of $\|\nabla_1q_1\Psi\|$ which is small in view of Lemma \ref{potdiffer}.

\end{itemize}

Before we prove the Theorem let us first do the preparations needed to control (\ref{alpha2b}) as explained right above, i.e. introduce the smeared out interaction $U_{\beta_1,\beta}$.

\begin{definition}\label{udef}
For any $0\leq  \beta_1\leq  \beta\leq 1$ and any $V_{\beta}\in\mathcal{V}_\beta$  we define
$$U_{\beta_1,\beta}(\mathbf{x}):=\left\{
         \begin{array}{ll}
           \frac{3}{4\pi}\|V_{\beta}\|_1N^{3\beta_1}, & \hbox{for $x<N^{-\beta_1}$;} \\
           0, & \hbox{else.}
         \end{array}
       \right.
$$
and \be \label{defh}h_{\beta_1,\beta}(x):=\int |x-y|^{-1}
(V_{\beta}(y)-U_{\beta_1,\beta}(y))d^3y \ee \end{definition}
\begin{lemma}\label{ulemma}
For any $0\leq  \beta_1\leq  \beta<1$ and any $V_{\beta}\in\mathcal{V}_\beta$
\begin{align}\label{nulltez}\Delta h=&V_{\beta}-U_{\beta_1,\beta}\;,\hspace{1cm}&U_{\beta_1,\beta}&\in \mathcal{V}_{\beta_1}\;,\\
\label{erstez}\|h_{\beta_1,\beta}\|\leq&  CN^{-1-\beta_1/2}\;,
 &\|h_{\beta_1,\beta}\|_3&\leq  CN^{-1}(\ln N)^{1/3}\;,
\\\label{zweitez}\|\nabla h_{\beta_1,\beta}\|_1\leq& C N^{-1-\beta_1}\;,
 &\|\nabla h_{\beta_1,\beta}\|&\leq  CN^{-1+\beta/2}
\;. \end{align}

\end{lemma}

\begin{proof} The Lemma gives in fact a well known result of standard electrostatics: $V_\beta$ can be understood as a
given charge density, $U_{\beta_1,\beta}$ was defined in such a way, that the ``total charge'' is zero. Hence the
potential $h_{\beta,\beta_1}$ is constant
outside the support of $U_{\beta_1,\beta}$ in our case, by definition ($h_{\beta,\beta_1}(x)=0$ decays like $x^{-1}$ as $x\to\infty$) this
constant is zero.

The first statement of the Lemma is almost trivial
 \begin{align*}\Delta h(x)=&\int \Delta |x-y|^{-1}
(V_{\beta}(y)-U_{\beta_1,\beta}(y))d^3y
\\=&V_{\beta}(x)-U_{\beta_1,\beta}(x)\;.
\end{align*}
By definition $U_{\beta_1,\beta}\in\mathcal{U}_{\beta_1}$,  $\|U_{\beta_1,\beta}\|_1=\|V_{\beta}\|_1$ and
$\lim_{N\to\infty}N^{1-3\beta_1}\|U_{\beta_1,\beta}\|_\infty<\infty$. Since $V_\beta\in\mathcal{V}_\beta$ we get that
$\lim_{N\to\infty}N^{1+\eta}(\|U_{\beta_1,\beta}\|_1-a/N)<\infty$ implying $U_{\beta_1,\beta}\in\mathcal{V}_{\beta_1}$
which completes the proof of line (\ref{nulltez}).
Since $\|U_{\beta_1,\beta}\|_1=\|V_{\beta}\|_1$ and
$\Delta h(x)=0$ for $x>N^{-\beta_1}$ it follows that
$h=0$ for $x>N^{-\beta_1}$. Furthermore $|h_{\beta_1,\beta}(x)|<CN^{-1}|x|^{-1}$ and $|\nabla h_{\beta_1,\beta}(x)|<CN^{-1}|x|^{-2}$,
 implying the two equations in line (\ref{erstez}) as well as the first equation in (\ref{zweitez}).

To get the second equation in (\ref{zweitez}) we write for $\nabla h_{\beta_1,\beta}$
\begin{align*} \nabla h_{\beta_1,\beta}(x)=&\int \mathds{1}_{|x-y|<N^{-\beta}}\frac{x-y}{|x-y|^{3}}
(V_{\beta}(y)-U_{\beta_1,\beta}(y))d^3y
\\&+\int \mathds{1}_{|x-y|>N^{-\beta}}\frac{x-y}{|x-y|^{3}}
(V_{\beta}(y)-U_{\beta_1,\beta}(y))d^3y
\end{align*}
Using Young's inequality it follows that
\begin{align*}
 \|\nabla h_{\beta_1,\beta}\|_{\infty}\leq& \left\| \mathds{1}_{|\cdot|<N^{-\beta}}\frac{(\cdot)}{|\cdot|^{3}}\right\|_1
\left(\|V_{\beta}\|_\infty+\|U_{\beta_1,\beta}\|_\infty\right)
\\&+\left\|\mathds{1}_{|\cdot|>N^{-\beta}}\frac{(\cdot)}{|\cdot|^{3}}\right\|_\infty\left(\|V_{\beta}\|_1+\|U_{\beta_1,\beta}\|_1\right)
\\\leq& CN^{-\beta} N^{-1+3\beta}+CN^{2\beta}N^{-1}\;.
\end{align*}
Since  $|\nabla h_{\beta_1,\beta}|<CN^{-1}|x|^{-2}$
\begin{align*}
\|\nabla h_{\beta_1,\beta}\|^2\leq& C\int_{N^{-\beta}}^\infty N^{-2}|x|^{-4}dx + CN^{-3\beta}\|\nabla h_{\beta_1,\beta}\|_\infty^2
\\\leq& C N^{-2+\beta}+C N^{-2+\beta}\;.
\end{align*}
\end{proof}

We now arrive at the central point of this section which is estimating $\alpha'_0$ (\ref{fnochdao}), $\alpha'_1$ ((\ref{alpha1}) and (\ref{alpha1b})) and $\alpha'_2$ ((\ref{alpha2}) and (\ref{alpha2b})) following the strategy explained above.

\begin{lemma}\label{hnorms}

Let $\mathcal{M}\subset\mathbb{N}$ with $1,2\notin\mathcal{M}$ and $m:\mathbb{N}^2\to\mathbb{R}^+$ with $m\leq  n^{-1}$.
\begin{enumerate}

\item Let $0<\beta\leq 1$,  $f\in L^\infty$. Then there exists a $C<\infty$ such that
$$\left|\laa\Psi,f(x_1)\Psi\raa-\langle\phi,f(x)\phi\rangle\right|\leq 2\|f\|_\infty C\alpha(\Psi,\phi)$$
for any   $\Psi\in\mathcal{H}_{\mathcal{M}}$.

\item Let $0<\beta\leq  1$, $V_\beta\in\mathcal{U}_\beta$ with   $\lim_{N\to\infty}N^{\eta}(N\| V_{\beta}f_{\beta_1,\beta}\|_1-a\|\leq \infty$ for some $\eta>0$.  Then
there exists a  $\mathcal{K}\in\mathcal{F}$ and a $\eta>0$ such that  $$N\|p_1p_2
\potdiff_\beta(x_1,x_2) q_1p_2\widehat{m}\|_{\mathcal{M}}\leq  \mathcal{K}(\phi)\|\phi\|_\infty N^{-\eta}
$$

\item Let $0<\beta<1/3$, $V_\beta\in\mathcal{V}_\beta$.  Then
there exists a  $\mathcal{K}\in\mathcal{F}$ and a $\eta>0$ such that  $$N|\laa\Psi
p_1p_2\potdiff_\beta(x_1,x_2)\widehat{m}q_1q_2\chi\raa| \leq  \mathcal{K}(\phi)\|\phi\|_\infty
((\laa\Psi,\widehat{n}\Psi\raa\laa\chi,\widehat{n}\chi\raa)^{1/2}+N^{-\eta})\;.
$$
for   any $\Psi,\chi\in\mathcal{H}_{\mathcal M}$.

\item  Let $0<\beta<1$, $V_\beta\in\mathcal{V}_\beta$. Then  there exists a  $\mathcal{K}\in\mathcal{F}$ and a $\eta>0$ such that
\begin{align*} N|\laa\Psi p_1q_2\potdiff_\beta(x_1,x_2)\widehat{m}q_1q_2\Psi\raa|
\leq& \mathcal{K}(\phi)(\|\phi\|_\infty+(ln N)^{1/3}\|\nabla\phi\|_{6,loc})
\\&(\laa\Psi,\widehat{n}\Psi\raa+\|\nabla_1q_1\Psi\|^2+N^{-\eta})\;.
\end{align*}
for any 
symmetric $\Psi$ (i.e. $\Psi\in\mathcal{H}_{\emptyset}$).

\end{enumerate}

\end{lemma}

\begin{proof}
\begin{enumerate}
\item Using $1=p_1+p_2$ and Lemma \ref{kombinatorikb}
\begin{align*}
 |\laa\Psi &,f(x_1)\Psi \raa-\langle\phi ,f(x)\phi \rangle |\nonumber\\\nonumber=&
\big|\laa\Psi ,p_1f(x_1)p_1\Psi \raa-\langle\phi ,f(x)\phi \rangle
+2\Re\left(\laa\Psi ,q_1f(x_1)p_1\Psi \raa\right)\\&+\laa\Psi ,q_1f(x_1)q_1\Psi \raa\big|
\\\nonumber\leq& (1-\|p_1\Psi \|^2)\langle\phi ,f(x)\phi \rangle
+2\left|\Re\left(\laa\Psi ,q_1\widehat{n}^{-1/2}f(x_1)\widehat{n}^{1/2}_1p_1\Psi \raa\right)\right|
\\&+\laa\Psi ,q_1f(x_1)q_1\Psi \raa
\\\leq& \alpha(\Psi,\phi)\|f\|_\infty
+C\alpha(\Psi,\phi)\|f\|_\infty+\alpha(\Psi,\phi)\|f\|_\infty
\;.
\end{align*}

\item

In view of \ref{kombinatorikb}
\begin{align*}
N \laa\Psi, p_1p_2
\potdiff_\beta(x_1,x_2) q_1p_2\widehat{m}\Psi\raa
\leq&  N\|p_1p_2
\potdiff_\beta(x_1,x_2) q_1p_2\|_{op}\|q_1\widehat{m}\Psi\|
\\\leq&  CN\|p_1p_2
\potdiff_\beta(x_1,x_2) q_1p_2\|_{op}\;.\end{align*}
$\|p_1p_2
\potdiff_\beta(x_1,x_2) q_1p_2\|_{op}$ can be estimated using $p_1q_1=0$ and (\ref{faltungorigin}):
\begin{align*}& \|p_1p_2
(V_\beta(x_1-x_2)-\frac{2a}{N-1}|\phi(x_1)|^2-\frac{2a}{N-1}|\phi(x_2)|^2) q_1p_2\|_{op}\\=&\|p_1p_2
(V_\beta(x_1-x_2)-\frac{2a}{N-1}|\phi(x_1)|^2
p_2\|_{op}\\\leq& \|p_1p_2\left((V_\beta\star|\phi|^2)(x_1)-\frac{2a}{N-1}|\phi(x_1)|^2\right)\|_{op}
\end{align*}
We introduce $a_N:=\|V_\beta\|_1$ Since $V_\beta\in\mathcal{V}_\beta$ we can find a $\eta>0$ such that
\begin{align*}
\leq& \|p_1\left((V_\beta-a_N\delta)\star|\phi|^2\right)(x_1)\|_{op}+CN^{-\eta}\|\phi\|^2_\infty
\\\leq& \|\phi\|_{\infty}\|(V_\beta-a_N\delta)\star|\phi|^2\|+CN^{-\eta}\|\phi\|^2_\infty\;.
\end{align*}
Let $h\in L^\infty$ be given by $$\Delta h(x)=V_\beta(x)-a_N\delta(x)\;.$$ As above (see Lemma \ref{ulemma}) we have that
$h(x)=0$ for $x>R N^{-\beta}$, where $R N^{-\beta}$ is the radius of the support of $V_\beta$ and that $\|\nabla h\|_1\leq  N^{-1-\beta}$. Partial integration and Young's inequality give that
\begin{align*}
\|(V_\beta-a_N\delta)\star|\phi|^2\|=&\|(\nabla h)\star(\nabla|\phi|^2)\|
\\\leq& \|\nabla h\|_1\|\nabla|\phi|^2\|\leq 2\|\nabla h\|_1\|\nabla\phi\|\;\|\phi\|_\infty\;.
\end{align*}
Hence \begin{align}\label{formelfuerkin}
\|p_1p_2\left((V_\beta\star|\phi|^2)(x_1)-\frac{2a}{N-1}|\phi(x_1)|^2\right)p_2\|_{op}\leq
C\|\phi\|_\infty N^{-\eta}(\|\nabla\phi\|\;+\|\phi\|_\infty)
\end{align}
for some $\eta>0$ and (b) follows.

\item


Let us first find an upper bound for $$\|\sum_{j\in\mathcal{M}} q_jf(x_1,x_j)\widehat{r}\,p_1
p_j\Psi\|^2$$ for general $r:\mathbb{N}^2\to\mathbb{R}^+$, $f\in L^2\cap L^1$ and $\Psi\in\mathcal{H}_{\mathcal M}$ with
$1,2\notin\mathcal M$.
Using Lemma \ref{kombinatorikb}
\begin{align}\label{symback}\|\sum_{j\notin\mathcal{M}} &q_jf(x_1-x_j)\widehat{r}\,p_1
p_j\Psi\|^2
\\\nonumber=&\sum_{j\neq k\notin\mathcal{M}}\laa \widehat{r}\,p_1p_j\Psi,f(x_1-x_j)q_jq_kf(x_1-x_k)
p_kp_1
\widehat{r}\,\Psi\raa
\\\nonumber&+\sum_{j\notin\mathcal{M}}\laa \widehat{r}\,\Psi,p_1  p_j
f(x_1-x_j)q_jf(x_1-x_j)p_1 p_j\widehat{r}\,\Psi\raa
\\\nonumber\leq& \sum_{j\neq k\notin\mathcal{M}}\laa q_k\widehat{r}\,\Psi,p_1\sqrt{f}(x_1-x_k)
p_j
\sqrt{f}(x_1-x_j)\\\nonumber&\hspace{2cm}\sqrt{f}(x_1-x_k)p_k\sqrt{f}(x_1-x_j)p_1
q_j\widehat{r}\,\Psi\raa
\\\nonumber&+\sum_{j\notin\mathcal{M}}\laa \widehat{r}\,\Psi,p_1  p_j
f(x_1-x_j)f(x_1-x_j)p_1 p_j\widehat{r}\,\Psi\raa
\\\nonumber\leq& N^2\|\sqrt{f}(x_1-x_2)p_1\|_{op}^4\;\|
q_2\widehat{r}\,\Psi\|^2
%
+N\|f^2\|_1\|\phi  \|_\infty^2\|\widehat{r}\,\|^2_{op}
\\\nonumber\leq& N^2\|\phi  \|_\infty^4\|f\|_1^2\;\|
\widehat{n}\widehat{r}\,\Psi\|^2
+N\|f\|^2\|\phi  \|_\infty^2\sup_{1\leq  k\leq  N}|r(k,N)^2|\;.
 \end{align}
We now come back to (c).
We define for
some $\varepsilon>0$ we shall specify below the functions
$m^{a,b}:\mathbb{N}^2\to\mathbb{R}^+$ by $$m^a(k,N):=m(k,N)\;\;\text{ for }\;\;
k<N^{1-\varepsilon}\;;\;\;\;\;\;m^a(k,N)=0\;\;\text{ for }\;\;k\geq N^{1-\varepsilon}$$ and
$m^b=m-m^a$. Note also that $p_2|\phi(x_1)|^2q_2=p_1|\phi(x_2)|^2q_1=0$. It follows that (c) is bounded by
\be\label{zweifreunde}N|\laa\Psi ,p_1  p_2
V_\beta(x_1-x_2) \widehat{m}^aq_1 q_2 \chi\raa|+N|\laa\Psi
,p_1  p_2 V_\beta(x_1-x_2) \widehat{m}^bq_1 q_2 \chi\raa| \;.\ee Defining also $g:\mathbb{N}^2\to\mathbb{R}^+$ by
$g(k,N)=1$ for $k<N^{1-\varepsilon}$, $g(k,N)=0$ for $k\geq N^{1-\varepsilon}$ we
have that $m^a=m^ag$ and the first summand in (\ref{zweifreunde}) equals
\begin{align}N\laa\Psi & ,\widehat{g}_{2}p_1 p_2 V_\beta(x_1-x_2) q_1
q_2 \widehat{m}^a\chi\raa
\nonumber\\=&
N(N-|\mathcal{M}|-1)^{-1}\laa\Psi ,\sum_{j\notin\mathcal{M}}\widehat{g}_{2}p_1 p_j
V_\beta(x_1-x_j) q_1 q_j \widehat{m}^a\chi\raa
 \nonumber\\\leq&
\|\sum_{j\notin\mathcal{M}} \widehat{g}_{2}q_jV_\beta(x_1-x_j)p_1
p_j\Psi\|\;\|\widehat{m}^a q_1 \chi\|\;.
 \end{align}
In view of (\ref{symback}) and Lemma \ref{kombinatorikb} the latter is bounded by
\begin{align*} N\|\phi&  \|_\infty^2\|V_\beta\|_1\;\|
\widehat{n}\widehat{g}_{2}\chi\|
+N^{1/2}\|V_\beta\|\;\|\phi  \|_\infty\sup_{1\leq  k\leq  N}|g(k,N)|
\\\leq& CN^{-\varepsilon/2}\|\phi\|_\infty^2\;
+CN^{-1/2+3\beta/2}\|\phi\|_\infty\;.
\end{align*}
 Using Lemma \ref{kombinatorikb} the second summand in (\ref{zweifreunde}) is bounded by \begin{align}
N\laa\Psi &,p_1 p_2 f(x_1,x_2) q_1 q_2
\widehat{m}^b\chi\raa
\nonumber\\=&
N(N-|\mathcal{M}|-1)^{-1}\laa\Psi ,\sum_{j\notin\mathcal{M}}(\widehat{m}^b_{2})^{1/2}p_1 p_j
f(x_1,x_j) q_1 q_j (\widehat{m}^b)^{1/2}\chi\raa
 \nonumber\\\leq&
C\|\sum_{j\notin\mathcal{M}} q_jf(x_1,x_j)(\widehat{m}^b_{2})^{1/2}p_1
p_j\Psi\|\;\|(\widehat{m}^b)^{1/2} q_1 \chi\|
\nonumber\\\label{l233}\leq&
C\|\sum_{j\notin\mathcal{M}} q_jf(x_1,x_j)(\widehat{m}^b_{2})^{1/2}p_1
p_j\Psi\|\;\sqrt{\alpha(\chi,\phi)}\;.
 \end{align}
Since $m(k,N)<\sqrt{N/k}$ one has $\sup_{1\leq  k\leq
N}|(m^b(k,N))^{1/2}|=N^{\varepsilon/4}$ and
$$\|
(\widehat{m}^b_{2})^{1/2}\widehat{n}\Psi\|^2\leq  C\|\widehat{n}^{1/2}_{2}\Psi\|^2\leq  C\alpha(\Psi,\phi)+CN^{-1/2}\;.$$
Thus (\ref{symback}) and Lemma \ref{kombinatorikb} imply that the second summand in (\ref{zweifreunde}) is bounded by
\begin{align*}
C\sqrt{\alpha(\chi,\phi)}&\left(N\|\phi  \|_\infty^2\|f\|_1\;\sqrt{\alpha(\Psi,\phi)}
+N^{1/2}\|f\|\;\|\phi  \|_\infty N^{\varepsilon/4}\right)
\\\leq& C\|\phi  \|_\infty^2\sqrt{\alpha(\chi,\phi)\alpha(\Psi,\phi)}
+CN^{-1/2-3\beta/2+\varepsilon}\|\phi\|_\infty
\end{align*}
Summarizing we have that
 \begin{align*} N|\laa\Psi ,p_1 p_2
f(x_1,x_2) \widehat{m}q_1 q_2 \Psi\raa|\leq&
C\|\phi
\|_\infty^2\sqrt{\alpha(\chi,\phi)\alpha(\Psi,\phi)}\\&\hspace{-0.4cm}+C\|\phi
\|_\infty N^{(-1+3\beta)/2+\varepsilon/4}+
C\|\phi
\|_\infty^2N^{-\varepsilon/2}\;.
 \end{align*}
Choosing $0<\varepsilon<(-1+3\beta)/2$ and
$\eta<\min\{-1+3\beta+2\varepsilon,\varepsilon/2\}$ (c) follows.
\item Let $U_{\beta_1,\beta}$ be given by definition \ref{udef}.
 As a first step we show that
for $0\leq \beta_1<\beta<1$ and for $h_{\beta_1,\beta}$ given by Definition \ref{udef} there exists a
$\mathcal{K}\in\mathcal{F}$ and a $\eta>0$ such that
\begin{align}\label{hnormd}N\laa\Psi ,& p_1q_2(V_\beta(x_1-x_2)-U_{\beta_1,\beta}(x_1-x_2))\widehat{m}q_1q_2\Psi\raa\\
\nonumber\leq&
\mathcal{K}(\phi)(\|\phi\|_\infty+\|\nabla\phi\|_{6,loc}(\ln N)^{1/3})\\&\nonumber(\laa\Psi,\widehat{n}\Psi\raa+N^{-\beta_1}
\|\nabla_1q_1\Psi\|^2+N^{-\eta})
\;.
\end{align}
Lemma \ref{ulemma} and integration by parts gives
\begin{align}\nonumber N\laa\Psi, & p_1q_2(V_\beta(x_1-x_2)-U_{\beta_1,\beta}(x_1-x_2))\widehat{m}q_1q_2\Psi\raa\\=&N|\laa\Psi ,q_1
p_2\widehat{m}_1 (\nabla_1
h_{\beta_1,\beta}(x_1-x_2))q_2\nabla_1q_1
\Psi\raa|\label{sa}
\\&+N|\laa\nabla_1q_1  \Psi ,p_2
(\nabla_1 h_{\beta_1,\beta}(x_1-x_2)) q_1 q_2\widehat{m} \Psi\raa|\label{snewb}
\;. \end{align}
For $(\ref{sa})$ we write
\begin{align*}
(\ref{sa})\leq& N(N-1)^{-1}|\laa\Psi ,\sum_{k=2}^Nq_1  p_k\widehat{m}_1 (\nabla_1
h_{\beta_1,\beta}(x_1-x_k))q_k\nabla_1q_1
\Psi\raa|
\\\leq& C\|\sum_{k=2}^N q_k(\nabla_1
h_{\beta_1,\beta}(x_1-x_k))q_1  p_k\widehat{m}_1\Psi \|\;\|\nabla_1q_1\Psi\|\;.
\end{align*}
For the first factor we have
\begin{align}\nonumber
\|\sum_{k=2}^N& q_k(\nabla_1
h_{\beta_1,\beta}(x_1-x_k))q_1  p_k\widehat{m}_1\Psi \|^2
\\\label{snewc}\leq& 2|\laa\Psi, \sum_{2\leq  k<j\leq  N} q_1  p_j\widehat{m}_1(\nabla_1
h_{\beta_1,\beta}(x_1-x_j)) q_j\\\nonumber&\hspace{4cm}q_k(\nabla_1
h_{\beta_1,\beta}(x_1-x_k))q_1  p_k\widehat{m}_1\Psi \raa|
\\&\label{sd}+|\laa\Psi, \sum_{2\leq  k\leq  N} q_1  p_k\widehat{m}_1(\nabla_1
h_{\beta_1,\beta}(x_1-x_k)) \\\nonumber&\hspace{4cm}q_k(\nabla_1
h_{\beta_1,\beta}(x_1-x_k))q_1  p_k\widehat{m}_1\Psi \raa|
\;.
\end{align}
Using that $\nabla_1
h_{\beta_1,\beta}(x_1-x_k)=-\nabla_k
h_{\beta_1,\beta}(x_1-x_k)$ and integrating by parts gives
\begin{align*}
(\ref{snewc})&\\\leq& N^2|\laa\Psi, q_3q_1  p_2\widehat{m}_1(\nabla_3
h_{\beta_1,\beta}(x_1-x_3))(\nabla_2
h_{\beta_1,\beta}(x_1-x_2)) q_2q_1  p_3\widehat{m}_1\Psi \raa|
\\\leq&
N^2|\laa\nabla_2 \nabla_3q_3q_1  p_2\widehat{m}_1\Psi,
h_{\beta_1,\beta}(x_1-x_2)
h_{\beta_1,\beta}(x_1-x_3) q_2q_1  p_3\widehat{m}_1\Psi \raa|
\\&+N^2|\laa\nabla_3q_3q_1  p_2\widehat{m}_1\Psi,
h_{\beta_1,\beta}(x_1-x_2)
h_{\beta_1,\beta}(x_1-x_3) \nabla_2q_2q_1  p_3\widehat{m}_1\Psi \raa|
\\&+
N^2|\laa q_3q_1 \nabla_2 p_2\widehat{m}_1\Psi,
h_{\beta_1,\beta}(x_1-x_2)
h_{\beta_1,\beta}(x_1-x_3) q_2q_1  \nabla_3p_3\widehat{m}_1\Psi \raa|
\\&+N^2|\laa q_3q_1  p_2\widehat{m}_1\Psi,
h_{\beta_1,\beta}(x_1-x_2)
h_{\beta_1,\beta}(x_1-x_3) \nabla_2\nabla_3q_2q_1  p_3\widehat{m}_1\Psi \raa|
\\\leq&
N^2\| \nabla_3q_3q_1  \widehat{m}_1\Psi\|\;\|h_{\beta_1,\beta}(x_1-x_2)\nabla_2p_2\|_{op}
\|h_{\beta_1,\beta}(x_1-x_3)p_3\|_{op}
\| q_2q_1  \widehat{m}_1\Psi \|
\\&+N^2\|\nabla_3q_3q_1\widehat{m}_1\Psi\|\;\|  p_2
h_{\beta_1,\beta}(x_1-x_2)\|_{op}
\|h_{\beta_1,\beta}(x_1-x_3)p_3\|_{op} \|\nabla_2q_2q_1  \widehat{m}_1\Psi \|
\\&+
N^2\|q_3q_1 \widehat{m}_1\Psi\|\;
\|h_{\beta_1,\beta}(x_1-x_2)\nabla_2 p_2\|_{op}
\|h_{\beta_1,\beta}(x_1-x_3)\nabla_3p_3\|_{op} \|q_2q_1  \widehat{m}_1\Psi \|
\\&+N^2\|q_3q_1 \widehat{m}_1\Psi\|
 \;\|h_{\beta_1,\beta}(x_1-x_2)p_2\|_{op}
\|h_{\beta_1,\beta}(x_1-x_3)\nabla_3p_3\|_{op} \|\nabla_2q_2q_1  \widehat{m}_1\Psi \|
\;.
\end{align*}
Note that \begin{align*} q_1 \nabla_2q_2\widehat{m} \Psi=&q_1 p_2\nabla_2q_2\widehat{m} \Psi+q_1 q_2\nabla_2q_2\widehat{m} \Psi
\\=&q_1 p_2\widehat{m}_1\nabla_2q_2 \Psi+q_1 \widehat{m}q_2\nabla_2q_2 \Psi\;.
\end{align*}

With  Lemma \ref{kombinatorikb} it follows that
\be\label{nablamitm}\|q_1 \nabla_2q_2\widehat{m} \Psi\|\leq  C \|\nabla_2q_2 \Psi\|\;.\ee
This and Lemma \ref{kombinatorik} (e) give
$$(\ref{snewc})\leq
C(N^{2}\|\nabla_1q_1\Psi\|^2\|h\|_2^2\|\phi\|_\infty^2+N^2\laa\Psi,\widehat{n}\Psi
\raa\|q_1\Psi\|^2\|h\|_3^2\|\nabla\phi\|_{6,loc}^2)\;.$$
with Lemma \ref{ulemma} it follows that
$$(\ref{snewc})\leq
C(N^{-\beta_1}\|\phi\|_\infty^2\|\nabla_1q_1\Psi\|^2+\laa\Psi,\widehat{n}\Psi
\raa\|q_1\Psi\|^2\|\nabla\phi\|_{6,loc}^2(\ln N)^{2/3})\;.$$

For $(\ref{sd})$ we have
\begin{align*}
(\ref{sd})\leq& N
\|\widehat{m}_1q_1\Psi\|^2\|   p_2(\nabla_1
h_{\beta_1,\beta}(x_1-x_2))\|_{op}^2
\\\leq& C N^{1-2+\beta}\|\phi\|_\infty^2\;.
\end{align*}
It follows that $(\ref{sa})$ is bounded by the right hand side of (\ref{hnormd}).

To control $(\ref{snewb})$ we use once more that $$\nabla_1 h_{\beta_1,\beta}(x_1-x_2)=-\nabla_2 h_{\beta_1,\beta}(x_1-x_2)$$
and integrate by parts
\begin{align*}
(\ref{snewb})\leq& |\laa\nabla_2p_2\nabla_1q_1  \Psi ,
h_{\beta_1,\beta}(x_1-x_2) q_1 q_2\widehat{m} \Psi\raa|
\\&+|\laa\nabla_1q_1  \Psi ,p_2
 h_{\beta_1,\beta}(x_1-x_2) q_1 \nabla_2q_2\widehat{m} \Psi\raa|
\\\leq&  \|\nabla_1q_1  \Psi\|
\|h_{\beta_1,\beta}(x_1-x_2)\nabla_2 p_2\|_{op}\| q_1 q_2\widehat{m} \Psi\|
\\&+\|\nabla_1q_1  \Psi\|\;\|p_2
 h_{\beta_1,\beta}(x_1-x_2)\|_{op} \|q_1 \nabla_2q_2\widehat{m} \Psi\|
\\\leq& C
\|\nabla_1q_1\Psi\|\left(\|q_1\Psi\|\;(\ln N)^{1/3}\|\nabla\phi\|_{6,loc}+N^{-\beta_1}
\|\phi\|_\infty\|\nabla_1q_1\Psi\| \right)\;.
\end{align*}
and (\ref{hnormd}) follows.

Now (\ref{hnormd}) together with Lemma \ref{ulemma} gives
\begin{align*}
&\hspace{-0.8cm}|N\laa\Psi p_1q_2\potdiff_\beta(x_1,x_2)\widehat{m}q_1q_2\chi\raa| \leq  CN|\laa\Psi ,p_1q_2(U_{0,\beta}(x_1-x_2)
\widehat{m}q_1q_2
\Psi\raa|
\\&\hspace{1.7cm}+CN|\laa\Psi ,p_1q_2\frac{a}{N-1}|\phi(x_1)|^2\widehat{m}q_1q_2
\Psi\raa|
\\&\hspace{1.7cm}+CN|\laa\Psi ,p_1q_2(V_\beta(x_1-x_2)-U_{0,\beta}(x_1-x_2))\widehat{m}q_1q_2
\Psi\raa|
\\\leq&
CN\|q_2\Psi\|\;\|p_1U_{0,\beta}(x_1-x_2)\|_{op} \|\widehat{m}q_1q_2
\Psi\|
\\& +C\|q_2\Psi\|\;\|\phi\|_\infty^2\|\widehat{m}q_1q_2
\Psi\|
\\& +\mathcal{K}(\phi)(\|\phi\|_\infty+\|\nabla\phi\|_{6,loc(\ln N)^{1/3}})
(\laa\Psi,\widehat{n}\Psi\raa+\|\nabla_1q_1\Psi\|^2+N^{-1/2})\;.
\end{align*}
Since by definition $N\|U_{0,\beta}\|\leq  C$ we get with Lemma \ref{kombinatorikb} that\\
$N\|p_1U_{0,\beta}(x_1,x_2)\|_{op}\leq  C\|\phi\|_\infty$ and (d) follows.

\end{enumerate}

\end{proof}

\subsection{Controlling the smoothness of $\Psi$}\label{secsmooth}

To get good control for the term in Lemma \ref{hnorms} (d) we need in addition a bound on $\|\nabla_1q_1\Psi\|$ in terms
of $\alpha(\Psi,\phi)+\landau(1)$. $\|\nabla_1q_1\Psi\|$ is a part of the kinetic energy of $\Psi$ so the idea is to
show that the other contributions to the energy $\mathcal E(\Psi)$ are in leading order cancelled out by $\mathcal{E}^{GP}$
and thus $\|\nabla_1q_1\Psi\|\approx \mathcal E(\Psi)-\mathcal{E}^{GP}\leq \alpha(\Psi,\phi)$.
In the next Lemma we estimate as a first step
 $\|\nabla_1q_1\Psi\|$ in terms of $\alpha$ and the difference between interaction and effective mean field. Note that
the following Lemma holds for any
$0\leq  \beta\leq  1$ and shall be useful when we generalize to $\beta=1$. We shall use this estimate in the following
section to control $\|\nabla_1q_1\Psi\|$ in terms of $\alpha(\Psi,\phi)+\landau(1)$ restricting to $0<\beta<1$.
For later reference we shall also show that the total interaction energy $\|\sqrt{V_\beta(x_1-x_2)}\Psi\|^2$ stays
bounded.

\begin{lemma}\label{totalE}
For any
$0< \beta\leq  1$ there exists a $\mathcal{K}\in\mathcal{F}$ such that
\begin{align}
\label{smooth1}N\|\sqrt{V_\beta(x_1-x_2)}\Psi\|^2
\leq& \alpha(\Psi,\phi)+\|\nabla\phi\|^2+2\|A\|_\infty+2a\|\phi\|_\infty^2
\\\|\nabla_1q_1\Psi\|^2
\leq& \laa\Psi,(2a|\phi(x_1)|^2-(N-1)V_\beta(x_1-x_2))\Psi\raa\nonumber
\\\label{smooth2}&+\mathcal{K}(\phi)\alpha(\Psi,\phi)
\;.\end{align}
uniform in $(\Psi,\phi)\in \LZN\otimes\LZ$.
\end{lemma}
\begin{proof}

Using
symmetry of $\Psi$
\begin{align}\label{ensplit}
\mathcal{E}(\Psi)-\mathcal{E}^{GP}(\phi)=&\|\nabla_1\Psi\|^2-\|\nabla\phi\|^2+
(N-1)\laa\Psi,V_\beta(x_1-x_2)\Psi\raa\\\nonumber&+\laa\Psi,A(x_1)\Psi\raa-\langle\phi,
2a|\phi|^2+A(x_1)\phi\rangle
\;.
\end{align}
Since
$|\mathcal{E}(\Psi)-\mathcal{E}^{GP}(\phi)|\leq \alpha(\Psi,\phi)$ $$(N-1)\|\sqrt{V_\beta(x_1-x_2)}\Psi\|^2\leq
\alpha(\Psi,\phi)+\|\nabla\phi\|^2+2\|A\|_\infty+2a\|\phi\|_\infty^2$$
and (\ref{smooth1}) follows.

For (\ref{smooth2}) note that \begin{align*}|\laa \nabla_1 q_1\Psi,\nabla_1
p_1\Psi\raa|=&|\laa
q_1\Psi,\widehat{n}^{1/2}_1\widehat{n}^{-1/2}_1 \Delta_1p_1\Psi\raa|
\\=|\laa
q_1\Psi,\widehat{n}^{-1/2}_1\Delta_1p_1
\widehat{n}^{1/2}_{2}\Psi\raa|\leq&
\|\Delta\phi\|\;\|\widehat{n}^{-1/2}_1q_1\Psi\|\;\|\widehat{n}^{1/2}_{2}p_1\Psi\|\;.
\end{align*} Thus with Lemma \ref{kombinatorikb} and using that
$\sqrt{\frac{k+2}{N}}\leq  3\sqrt{\frac{k}{N}}$ we
get that \be \label{mixedder} |\laa \nabla_1 q_1\Psi,\nabla_1
p_1\Psi\raa|\leq  C\|\Delta\phi\|\alpha(\Psi,\phi)\;, \ee
implying
$$\|\nabla_1 \Psi\|^2-\|\nabla_1 p_1\Psi\|^2\geq\|\nabla_1
q_1\Psi\|^2-C\|\Delta\phi\|\alpha(\Psi,\phi)\;.$$ Since
\be\label{ekinmin}\|\nabla_1
p_1\Psi\|^2=\|p_1\Psi\|^2\;\|\nabla\phi\|^2=(1-\|q_1\Psi\|^2)\|\nabla\phi\|^2\ee
it follows that \be\label{onehand}\|\nabla_1 \Psi\|^2-\|\nabla\phi\|^2\geq\|\nabla_1
q_1\Psi\|^2-C\|\Delta\phi\|\alpha(\Psi,\phi)-\|\nabla\phi\|^2\alpha(\Psi,\phi)\;.\ee

Using this
and Lemma \ref{hnorms} (a) setting $f=A+2a|\phi|^2$ we get with (\ref{ensplit})
\begin{align*} \mathcal{E}(\Psi)-\mathcal{E}^{GP}(\phi)
\geq&  \|\nabla_1
q_1\Psi\|^2+\laa\Psi,((N-1)V_\beta(x_1-x_2)-2a|\phi(x_1)|^2)\Psi\raa
\\&-C(\|A\|_\infty+2a\|\phi\|^2_\infty+\|\Delta\phi\|+\|\nabla\phi\|^2)\alpha(\Psi,\phi)\;.
\end{align*}
Since $$\|\nabla\phi\|^2=\langle\nabla\phi,\nabla\phi\rangle=-\langle\phi,\Delta\phi\rangle\leq \|\Delta\phi\|$$
we get (b).
\end{proof}

\subsection{Controlling $\|\nabla_1q_1\Psi\|$ for $\beta<1$}\label{secabl}

With Lemma \ref{totalE} we have found a bound on $\|\nabla_1q_1\Psi\|$ for all $0<\beta\leq  1$ in terms of
effective mean field minus interaction.
Note that in view of Lemma \ref{totalE} (a) the interaction is bounded. The mean-field term is bounded by
 $\laa\Psi,|\phi|^2\Psi\raa\leq  \|\phi\|_\infty^2$ and we have for any $0< \beta\leq  1$
\be\label{ablabsch}
\|\nabla_1q_1\Psi\|^2\leq \mathcal{K}(\phi)(\alpha(\Psi,\phi)+1)\;.
\ee
But --- as explained above --- we want to show that $\|\nabla_1q_1\Psi\|$ is small using that the effective mean field
cancels out the
leading order of the interaction.
\begin{lemma}\label{potdiffer}
Let  $0<\beta<1$, $V_\beta\in\mathcal{V}_\beta$ and $m:\mathbb{N}^2\to\mathbb{R}^+$ with $m\leq  n^{-1}$.
Then there exists a $\eta>0$ and a $\mathcal{K}\in\mathcal{F}$ such that for
   any $\Psi\in\mathcal{H}_{\emptyset}$ and any $\phi\in \LZ$
\begin{enumerate}

 \item
$$N|\laa\Psi,p_1p_2V_\beta(x_1-x_2) (1-p_1p_2)\Psi\raa|\leq  \C\;.
$$

\item \begin{align*}
\laa
\Psi(2a|\phi(x_1)|^2)-(N-1)V_{\beta}(x_1-x_2)\Psi\raa
\leq \C
\;. \end{align*}

\item
\be\label{kinenzit}\|\nabla_1q_1\Psi\|^2\leq  \C\;,\ee


\end{enumerate}
\end{lemma}

\begin{proof}
\begin{enumerate}
\item
Since $1-p_1p_1=p_1q_2+q_1p_2+q_1q_2$ it suffices to show that
\begin{align} N|\laa\Psi,p_1p_2V_\beta(x_1-x_2) p_1q_2\Psi\raa|\leq& \C\label{splitt1}
\\N|\laa\Psi,p_1p_2V_\beta(x_1-x_2) q_1q_2\Psi\raa|\leq& \C\label{splitt2}
\end{align}
For (\ref{splitt1}) we use Lemma \ref{kombinatorik} (c) and (e) as well as Lemma \ref{kombinatorikb} to get
\begin{align*}
(\ref{splitt1})=&N|\laa\Psi,\widehat{n}_1^{-1/2}p_1p_2V_\beta(x_1-x_2) \widehat{n}_2^{1/2}p_1q_2\Psi\raa|
\\\leq& N\|\widehat{n}_1^{-1/2}\Psi\|\;\|p_1V_\beta(x_1-x_2)p_1\|_{op}\|\widehat{n}_2^{1/2}q_2\Psi\|
\\\leq& C \|\phi\|_\infty^2\alpha(\Psi,\phi)\;.
\end{align*}
For $0< \beta<1/3$ (\ref{splitt2}) is Lemma \ref{hnorms} (c) for the special case $\widehat{m}=1$  (recall that $p_1
|\phi|^2(x_2)p_1=p_2 |\phi|^2(x_1)p_2=0$).
To generalize to $0<\beta<1$ we use Definition \ref{udef}
\begin{align*}
N|\laa\Psi,p_1p_2V_\beta (x_1-x_2)q_1q_2\Psi\raa|\;\leq \; N|\laa\Psi,p_1p_2U_{1/4,\beta} (x_1-x_2)q_1q_2\Psi\raa|&
\\+N|\laa\Psi,p_1p_2(V_\beta-U_{1/4,\beta}) (x_1-x_2)q_1q_2\Psi\raa|&
\end{align*}
Since $U_{1/4,\beta}\in\mathcal{V}_{1/4}$ (Lemma \ref{ulemma}) the first summand has the right bound (Lemma \ref{hnorms} (c)).
To finish the proof of the Lemma we verify the following formula, which shall be also of use later on.

 \be\label{l234}N|\laa\Psi ,p_1  p_2 (\Delta
h_{1/4,\beta})(x_1-x_2) q_1 q_2 \Psi\raa| \leq  C(\|\phi
 \|_\infty+\|\nabla\phi  \|_\infty) N^{-\eta} \ee to get (\ref{splitt2}) in full generality.
Integrating by parts we get
\begin{align}\nonumber& N|\laa\Psi ,p_1 p_2 (\Delta h_{1/4,\beta})q_1 q_2 \Psi\raa|
\\\label{inta}&\hspace{1.5cm}\leq
N|\laa\Psi ,p_1  p_2  (\nabla_1 h_{1/4,\beta}(x_1-x_2))\nabla_1q_1 q_2
\Psi\raa|
\\&\label{intb}\hspace{2cm}+N|\laa\nabla_1p_1  p_2\Psi ,
(\nabla_1 h_{1/4,\beta}(x_1-x_2)) q_1 q_2 \Psi\raa|
\;. \end{align}
 To control $(\ref{inta})$ we use similar ideas as in the proof of Lemma \ref{hnorms} \begin{align*}
(\ref{inta})\leq& C\|\sum_{j=2}^{N}q_j
(\nabla_1 h_{1/4,\beta}(x_1-x_j)) p_1  p_j\Psi \|\;\|\nabla_1q_1\Psi\|\;. \end{align*} The second factor is bounded (see
(\ref{ablabsch})). For the first factor
we write \begin{align}\nonumber \|&\sum_{j=2}^{N}q_j (\nabla_1 h_{1/4,\beta}(x_1-x_j)) p_1
p_j\Psi \|^2\\&=\label{intc}
\sum_{j\neq k\neq 1}\laa\Psi, p_1  p_k (\nabla_1
h_{1/4,\beta}(x_1-x_k))q_kq_j (\nabla_1 h(x_1-x_j)) p_1 p_j\Psi\raa
\\\label{intd}&+\sum_{j=2}^{N}\laa\Psi, p_1  p_j (\nabla_1
h_{1/4,\beta}(x_1-x_j))^2 p_1 p_j\Psi\raa
\;.
 \end{align}
$(\ref{intc})$ can be estimated using Lemma \ref{kombinatorik} (e) and Lemma \ref{ulemma} \begin{align*}
(\ref{intc})=&
\sum_{j\neq k\neq 1}\laa\Psi, p_1  q_j
 (\nabla_1 h(x_1-x_j))p_j p_k(\nabla_1h_{1/4,\beta}(x_1-x_k))p_1q_k \Psi\raa
\\\leq& N^2\|p_1(\nabla_1 h(x_1-x_j))p_j\|_{op}^2
=N^2\|\phi\|_\infty^4\|\nabla_1
h_{1/4,\beta}\|_1^2
\\\leq& CN^2\|\phi\|_\infty^4 N^{-5/2}=C\|\phi\|_\infty^4 N^{-1/2}\;.
\end{align*}

$(\ref{intd})$ is bounded by
$$N\|\nabla h_{1/4,\beta}(x_1-x_j)p_1\|_{op}^2\leq  N\|\nabla h_{1/4,\beta}(x_1-x_j)\|^2\|\phi\|_\infty^2\leq
CN^{-1+\beta}\|\phi\|_\infty^2\;.$$
It follows that
$(\ref{inta})$ is bounded by $\mathcal{K}(\phi)N^{-\eta}$ for some $\eta>0$.

Integration by parts yields for $(\ref{intb})$
\begin{align*}
(\ref{intb})\leq& N|\laa\Delta_1p_1  p_2\Psi ,
 h_{1/4,\beta}(x_1-x_2) q_1 q_2 \Psi\raa|
\\&+N|\laa\nabla_1p_1  p_2\Psi ,
 h_{1/4,\beta}(x_1-x_2) \nabla_1q_1 q_2 \Psi\raa|
\\\leq& N\|\Delta_1p_1\Psi\| \|p_2 h_{1/4,\beta}(x_1-x_2)\|_{op}
\\&+N\|\nabla_1p_1\Psi\| p_2  h_{1/4,\beta}(x_1-x_2)\|_{op} \|\nabla_2
q_2 \Psi\|
\\\leq& N\|\Delta\phi\|\;\|\phi\|_\infty\;\|h_{1/4,\beta}\|
\\&+N\|\nabla\phi  \|\;\|\phi  \|_\infty\;\|h_{1/4,\beta}\|\;\|\nabla_2
q_2 \Psi\|
 \end{align*}
with Lemma \ref{ulemma} and (\ref{ablabsch}) we get (\ref{l234}) and (\ref{splitt2}) follows.

\item

We get using   selfadjointness of the
multiplication operators \begin{align}
\nonumber\laa
\Psi,&\left(2a|\phi(x_1)|^2-(N-1)V_{\beta}(x_1-x_2)\right)\Psi\raa
\\\label{zeile1}=&\laa p_1p_2\Psi,\left(2a|\phi(x_1)|^2-(N-1)V_{\beta}(x_1-x_2)\right)p_1p_2\Psi\raa
\\\label{zeile2}&+2\Re\laa p_1p_2\Psi,\left(2a|\phi(x_1)|^2-(N-1)V_{\beta}(x_1-x_2)\right)(1-p_1p_2)\Psi\raa
\\&\label{positiverterm}-(N-1)\laa (1-p_1p_2)\Psi,V_{\beta}(x_1-x_2)(1-p_1p_2)\Psi\raa
\\\label{zeile3}&+2a\laa (1-p_1p_2)\Psi,|\phi(x_1)|^2(1-p_1p_2)\Psi\raa
\;. \end{align}
(\ref{zeile1}) is controlled by formula (\ref{formelfuerkin}).

Using symmetry of $\Psi$ and $p_2|\phi(x_1)|^2q_2=0$, the absolute value of (\ref{zeile2})
 is bounded by
 \begin{align*} &\hspace{-1cm}2\left|\llaa
p_1p_2\widehat{n}_2^{1/2}\Psi,2a|\phi(x_1)|^2\widehat{n}_1^{-1/2}q_1p_2\Psi\rraa\right|
\\&+2(N-1)\left|\llaa
p_1p_2\Psi,2V_{\beta}(x_1-x_2)(1-p_1p_2)\Psi\rraa\right|
\end{align*}
The first line is controlled by Lemma \ref{kombinatorik} (e) with Lemma \ref{kombinatorikb} and thus bounded by $\|\phi\|_{\infty}^2\alpha(\Psi,\phi)$. The second line is controlled by by part (a)
of the Lemma. Thus
we can find a $\kinf$ such that
$$(\ref{zeile2})\leq  \C\;. $$
Positivity of $V_{\beta}$ implies
that line (\ref{positiverterm}) is negative.
(\ref{zeile3}) is bounded by $$\|(p_1q_2+q_1p_1+q_1q_1)\Psi\|^2\|\phi\|^2_\infty\leq  C\|\phi\|^2_\infty\alpha(\Psi,\phi)
$$
 and we get (b).
\item follows from (b) with
  Lemma \ref{totalE}.
\end{enumerate}
\end{proof}

\subsection{Proof of the Theorem for $0<\beta<1/3$}\label{secproof1}

With Lemma \ref{hnorms} and Lemma \ref{potdiffer} we can now estimate $\alpha'_j$ for $j=0,1,2$ for $\beta<1/3$.
 We arrive directly at the following Corollary.
\begin{corollary}\label{alphaallg}
Let $0<\beta<1/3$, $V_\beta\in \mathcal{V}_\beta$. Then there
exists a $\kinf$ and a $\eta>0$ such that for any symmetric $\Psi$, any $\phi$ and any $j=0,1,2$
\be\label{alphacontrol}|\alpha'_j(\Psi,\phi)|\leq  \CphiA\C\;.\ee

\end{corollary}
The  Theorem follows for $0<\beta<1/3$ via Gr\o nwall.


\section{Generalizing to $1/3\leq \beta< 1$}\label{secb2}

 For $\beta>1/3$
the radius of the interactions is much smaller than the mean distance of the particles, so the interactions do not
overlap
for typical configurations any more. Still the interaction of our $N$-body system can be approximated by
an effective mean field, let us explain why:
Whenever two or more particles come very
close the wave function is  affected on a microscopic length scale
by the interaction of the particles.
Neglecting three particle interactions the microscopic structure can be constructed from the zero energy scattering
states of
$V_\beta$.
This can be made more clear with the following
heuristic argument: In principle one could control the time evolution
of $\Psi_t$ by generalized eigenfunction expansion. The relevant
eigenfunctions are on a microscopic scale
approximately given by this zero
energy scattering state.

 By this microscopic structure
the effect of the interactions is smeared out. For the smeared out effective interactions  the old mean field argument
holds. Therefore the first step when generalizing the Theorem is
to say something about the microscopic structure of the wave functions.

Below we shall use our estimates on the microscopic structure in two
places. On the one hand it makes a better control on $\Psi_t$ and
thus a better control of $\alpha(\Psi_t,\phi_t)$ for
$\beta\geq1/3$ possible. On the other hand  for $\beta=1$ the
interaction energy of $\Psi_t$ can only be controlled in a suitable
way when the microscopic structure of $\Psi_t$ is known.

\subsection{Microscopic Structure}\label{secmic}

For technical reasons we shall make a smooth spacial cutoff of the
zero energy scattering state. We do so by defining --- depending on $V_\beta$ --- a
potential
$W_{\beta_1}\in\mathcal{U}_{\beta_1}$ with softer scaling behavior $\beta_1<\beta$ in such a way that the potential
$V_\beta-W_{\beta_1}$
has scattering length zero, i.e. the zero energy scattering
state of  $V_\beta-W_{\beta_1}$ is outside the support of  $W_{\beta_1}$ equal to one.

\begin{definition}\label{microscopic}
Let $0<\beta_1<\beta\leq  1$, $V_{\beta}\in\mathcal{V}_{\beta}$ and $a_N/(4\pi)$ be the scattering length of $V_\beta/2$.
We define the potential $W_{\beta_1}$ via
$$W_{\beta_1}(x):=\left\{
  \begin{array}{ll}
    a_N N^{3\beta_1}, & \hbox{ for $N^{-\beta_1}<x< R_{\beta_1}$;} \\
    0, & \hbox{else.}
  \end{array}
\right.$$
$R_{\beta_1}$ is the minimal value which ensures that the
scattering length of $V_{\beta}-W_{\beta_1}$ is zero.

The respective zero energy scattering state shall be denoted by
$f_{\beta_1,\beta}$, i.e.
\be\label{fdef}\left(-\Delta+\frac{1}{2}(V_{\beta}-W_{\beta_1})\right)f_{\beta_1,\beta}=0\;,\ee
we shall also need $$g_{\beta_1,\beta}=1-f_{\beta_1,\beta}\;.$$
\end{definition}

\begin{lemma}\label{defAlemma}
For any $0<\beta_1<\beta\leq 1$,
$V_{\beta}\in\mathcal{V}_{\beta}$

\begin{enumerate}

\item
$$W_{\beta_1}f_{\beta_1,\beta}\in\mathcal{V}_{\beta_1}\;,\hspace{1cm}
 \lim_{N\to\infty}N^{\eta}|N\| V_{\beta}f_{\beta_1,\beta}\|_1-a|<\infty \;.$$

\item \begin{align*}\|g_{\beta_1,\beta}\|_1&\leq  C N^{-1-2\beta_1}\;,\hspace{1.5cm}\|g_{\beta_1,\beta}\|\leq  C N^{-1-\beta_1/2}\;,\\
 \|g_{8/9,1}\|_{3}&\leq  CN^{-1}(\ln N)^{1/3}\end{align*}


\item
For any
$x_2\in\mathbb{R}^3$ and any $\Psi\in\LZN$
$$\|\mathds{1}_{|x_1-x_2|\leq  R_{\beta_1}}\nabla_1\Psi\|^2+\frac{1}{2}\laa\Psi,
(V_{\beta_1}-W_{\beta_1})(x_1-x_2)\Psi\raa\geq0\;.$$

\end{enumerate}
\end{lemma}

\begin{proof}

Let $j_{\beta}$ be the zero energy scattering state of the
potential $V_{\beta}/2$.

Before we prove the different points of the Lemma, let us give some properties of $f_{\beta_1,\beta}$ first.

Since $V_{\beta}$ is
positive and has compact support of radius $r$ it follows, that
$1>j_{\beta}(x)\geq 1-a_N/(4\pi x)$.

Note, that the potential $W_{\beta_1}$ is
zero inside the Ball around zero of radius $N^{-\beta_1}$, hence
$f_{\beta_1,\beta}$ is inside this Ball a multiple of
$j_{\beta}$, i.e. there exists a $K_{\beta_1}$ such that
$$K_{\beta_1}f_{\beta_1,\beta}(x)=j_{\beta}(x)\text{ for
} x<N^{-\beta_1}\;,$$ in particular the derivative $d_x K_{\beta_1}f_{\beta_1,\beta}(x)$ is positive for $x=N^{-\beta_1}$.

For $x>N^{-\beta_1}$ the $f_{\beta_1,\beta}$ ``sees'' a negative potential, namely $-W_{\beta_1}$. Due to spherical symmetry $f_{\beta_1,\beta}$ is in that region a linear combination of an in- and an outgoing spherical wave with momentum $k_0=\sqrt{a_NN^{3\beta_1}}$.

By definition $R_{\beta_1}$ is the minimal value which ensures that the
scattering length of $V_{\beta}-W_{\beta_1}$ is zero, i.e. the minimal value which ensures that $f_{\beta_1,\beta}$ is constant for $x>R_{\beta_1}$.  In other words $R_{\beta_1}$   is the minimal value satisfying
$$d_{|x|}f_{\beta_1,\beta}(|x|)\big|_{|x|=R_{\beta_1}}=0\;.$$
It follows that $f_{\beta_1,\beta}$ is a positive function and that $d_{|x|}f_{\beta_1,\beta}\geq 0$.

Finally we have to control the constant $K_{\beta_1}$ which ensures that $f_{\beta_1,\beta}(x)=1$ for $x>R_{\beta_1}$: Since
$W_{\beta_1}$ is positive, that
$K_{\beta_1}d_xf_{\beta_1,\beta}\leq
d_xj_{\beta}$ and $K_{\beta_1}f_{\beta_1,\beta}\leq  j_{\beta}$.

Since $f_{\beta_1,\beta}(x)=1$ for $x>R_{\beta_1}$ and
$\lim_{x\to\infty}j_{\beta}(x)=1$ we get that
$K_{\beta_1}\leq  1$. On the other hand we have, since \be\label{fabschink}1\geq f(N^{-\beta_1})=j_{\beta}(N^{-\beta_1})/K_{\beta_1}\geq(1-a_N/(4\pi N^{-\beta_1}))/K_{\beta_1}\ee that
\be\label{fabschink2}(1-a_N/(4\pi N^{-\beta_1}))\leq  K_{\beta_1}\leq  1\;.\ee

\begin{enumerate}
 \item Using that $ f_{\beta_1,\beta}\geq j_{\beta}$ it follows that
\be\label{gbound}|g_{\beta_1,\beta}(x)|\leq  a_N/(4\pi x)\;.\ee Thus \begin{align*}
 N\|g_{\beta_1,\beta}V_\beta\|_1=&N\int_{0}^{N^{-\beta}}|g_{\beta_1,\beta}(x)V_\beta(x)|x^2dx+N\int_{N^{-\beta}}^{\infty}|g_{\beta_1,\beta}(x)V_\beta(x)|x^2dx
 \\\leq& C NN^{-1+3\beta}N^{-1-2\beta}+CNN^{-1+\beta}N^{-1}=CN^{\beta-1}\;.
 \end{align*}
Since $V_{\beta}\in\mathcal{V}_\beta$ it follows that there exists a $\eta>0$ such  that
$\lim_{N\to\infty}N^{\eta}|\|V_\beta\|-a|<\infty$ and we get the second statement in (a).

The scattering length of the potential $V_{\beta}-W_{\beta_1}$ is zero. Thus $\int (V_\beta(x)-W_{\beta_1}(x))f_{\beta_1,\beta}(x)dx=0$ and also $\lim_{N\to\infty}N^{\eta}(N\| W_{\beta_1}f_{\beta_1,\beta}\|_1-2a\|\leq \infty$. This implies in particular that $R_{\beta_1}$ is of order $N^{-\beta_1}$ thus $W_{\beta_1}f_{\beta_1,\beta}\in\mathcal{V}_{\beta_1}$.

\item

Since $g_{\beta_1,\beta}(x)=0$ for $x>N^{-\beta_1}$  it follows
that \begin{align*}\|g_{\beta_1,\beta}\|_1\leq&  \frac{1}{4\pi}a_N\int_{0}^{R_{\beta_1}}
|x|^{-1}d^3x\leq  CN^{-1-2\beta_1}
\\\|g_{\beta_1,\beta}\|^2\leq&
\frac{1}{16\pi^2}
a_N^2\int_{0}^{R_{\beta_1}} |x|^{-2}d^3x\leq  CN^{-2-\beta_1}
\\\|g_{\beta_1,\beta}\|_{3}^{3}\leq&  \frac{1}{64 \pi^3}a_N^{3}\int_{0}^{R_{\beta_1}}
|x|^{-3}d^3x\leq  CN^{-3}\ln N
\;.  \end{align*}

\item

To prove (c) we first show that for any $n\in \mathbb{N}$ and any subset $X_n\subset\mathbb{R}^3$ with $|X_n|=n$ which is
 such that the supports of the potentials $W_{\beta_1}(\cdot-x)$ are pairwise disjoint for any $x\in X_n$ the operator
$$H^{X_n}:=-\Delta+\sum_{x_k\in X_n} (V_{\beta_1}(\cdot-x_k)-W_{\beta_1}(\cdot-x_k))$$
is nonnegative.

This one can see in the following way: For any such $X_n$ the zero energy scattering state of $H^{X_n}$ is given by
$$F^{X_n}_{\beta_1,\beta}:=\prod_{x_k\in X_n}f_{\beta_1,\beta}(\cdot-x_k)\;.$$
By construction the $f_{\beta_1,\beta}$ are positive, so is $F^{X_n}_{\beta_1,\beta}$.
Assume now that $H^{X_n}$ is not nonnegative, i.e. that there exists a ground state $\Psi\in L^2$ of $H^{X_n}$ of
negative energy $E$. Since the phase of the ground state can be chosen such that  the ground state is  positive we get
\be\label{contra}\laa F^{X_n}_{\beta_1,\beta},H^{X_n}\Psi\raa=\laa F^{X_n}_{\beta_1,\beta},E\Psi\raa<0\;.\ee

On the other hand we have since $F^{X_n}_{\beta_1,\beta}$ is the zero energy scattering state $$\laa
F^{X_n}_{\beta_1,\beta},H^{X_n}\Psi\raa=\laa H^{X_n}F^{X_n}_{\beta_1,\beta},\Psi\raa=0\;.$$
This contradicts (\ref{contra}) and the nonnegativity of $H^{X_n}$ follows.

Having shown that the $H^{X_n}$ are nonnegative we prove (c) by contradiction. Assume that there exists a $\Psi\in
\mathcal{D}(H)$ such that
$$\|\mathds{1}_{|x|\leq  R_{\beta_1}}\nabla_1\Psi\|+\laa\Psi,(V_{\beta_1}(x_1)-W_{\beta_1}(x_1))\Psi\raa=E<0\;.$$ Since $V_{\beta_1}$ and $W_{\beta_1}$ are spherically symmetric we can assume that $\Psi$ is spherically symmetric and $\Psi(x)=1$ for $|x|> R_{\beta_1}$. We shall construct
now a set of points $X_n$ and a $\chi\in L^2$ such that $\laa\chi,H^{X_n}\chi\raa<0$, contradicting to nonnegativity of $H^{X_n}$.

For any $R>0$ let
$$\xi_R(x):=\left\{
         \begin{array}{ll}
           R^2/x^2, & \hbox{for $x>R$;} \\
           1, & \hbox{else.}
         \end{array}
       \right.
$$
Let now $X_n$ be a subset $X_n\subset\mathbb{R}^3$ with $|X_n|=n$ which is such that the supports of the potentials
$W_{\beta_1}(\cdot-x_k)$ lie within the Ball around zero with radius $R$ and are pairwise disjoint for any $x_k\in
X_n$. Since we are in three dimensions we can choose a $n$ which is of order $R^3$.

Let now $\chi_R:=\xi_R \prod_{x_k\in X_n} \Psi(x-x_k)$. The energy inside the ball $B_R(0)$ equals $E n$, thus it is negative
and of order $R^3$. Outside the ball we have only kinetic energy
$4\int_{x>R} R^4/x^6d^3x$ which  is of order $R$. Choosing $R$ large enough we can find a $X_n$
such that
$$\laa\chi_R,H^{X_n}\chi_R\raa$$ is negative, contradicting nonnegativity of $H^{X_n}$.

\end{enumerate}

\end{proof}

\subsection{First adjustment of the functionals}\label{secadj1}

As mentioned in the introduction we shall use the control on the
microscopic structure, i.e. the control of the zero energy
scattering state of $(V_{\beta}-W_{\beta_1})/2$ with some potential $W_{\beta_1}$ with softer scaling $\beta_1<\beta$
to get a control of $\Psi$ when $\beta$
increases. The first idea one might have is to divide $\Psi$ through
a function which approximates the microscopic structure (e.g. the
product $\prod_{j\neq k}f_{\beta_1,\beta}$ for some suitable
$0<\beta_1<\beta$), but this is not what we shall do. One reason is that $\left(\prod_{j\neq
k}f_{\beta_1,\beta}\right)^{-1}$ gets very large when many particles get very close.

Instead of dividing $\Psi$ through its microscopic structure we
equip the projectors with the respective microscopic structure to
get the desired estimates. Roughly speaking: The operator
$(-\Delta_1-\Delta_2+V_{\beta}(x_1-x_2))p_1p_2$ is hard to control for
large $\beta$ since
$ V_{\beta}(x_1-x_2)) $ is peaked for small
$|x_1-x_2|$. But since $f_{\beta_1,\beta}$ is the zero energy
scattering state of $-\Delta_1+(V_{\beta}-W_{\beta_1})/2$ it follows that
$(-\Delta_1-\Delta_2+V_{\beta}(x_1-x_2))f_{\beta_1,\beta}p_1p_2\Psi$ is
smoother.

To get good estimates we shall incorporate this idea in a very sensible way. How this can be done is easiest explained for a different functional, namely $\widetilde\alpha(\Psi,\phi)=\laa\Psi,\widehat{n}^2\Psi\raa=\laa\Psi,q_1\Psi\raa$ (see formula (\ref{partnumber})). Taking the time derivative and using that $q_1=1-p_1$ one gets among other terms
$$\dt\widetilde\alpha(\Psi_t,\phi_t)=i\sum_{j<k}N^{-1}\laa\Psi,[V_\beta(x_j-x_k),p_1]\Psi\raa\;.$$ Most of the interaction terms commute with $p_1$, only $i\sum_{1<k}N^{-1}\laa\Psi,[V_\beta(x_1-x_k),p_1]\Psi\raa$ remains. Hence the microscopic structure only for the interaction $V_\beta(x_1-x_k)$ matters.

This insight can also be used for our $\alpha_t$. Here many of the interactions cancel out due to  Lemma \ref{ableitung}) (d). 
Looking at Lemma \ref{hnorms} and considering $1/3<\beta<1$ for the moment, it is (\ref{alpha1b}) which we do not have good control of.
Consider the following functional
\begin{align}\label{adjalpha}&\alpha(\Psi,\phi)+N(N-1)\Im\left(\llaa\Psi
, q_1q_2g_{\beta_1,\beta}(x_{1}-x_{2}) (\widehat{n}-\widehat{n}_2)p_{1}p_{2}
\Psi\rraa\right)
\;.\end{align}
Taking the time derivative of this new functional one gets among other terms a
\be\label{adjust}N(N-1)\Im\left(\llaa\Psi
,q_1q_2 [-\Delta_1-\Delta_2,g_{\beta_1,\beta}(x_{1}-x_{2})] (\widehat{n}-\widehat{n}_2)p_{1}p_{2}
\Psi\rraa\right)\;.\ee
The commutator equals $(1-g_{\beta_1,\beta})(V_{\beta}(x_1-x_2)-W_{\beta_1})(x_1-x_2))$ plus mixed derivatives and one
sees, that the interactions in (\ref{alpha1b}) are ``replaced''
by $W_{\beta_1}$ for the price of new terms that have to be estimated.

Note that this adjustment  is not sufficient to get good control of the adjusted functional: Taking the time derivative of this new functional one gets terms where the potential $V_{\beta}$ does not cancel. As one shall see below one of these terms is not small for all $\beta<1$, but still it is better than (\ref{alpha1b}).
Therefore we make a similar adjustment as  before. We arrive at an iterative adjustment which leads to terms which we get better an better control of. With the iteration we will define below it turns out that five steps are enough to get sufficient control.

Guided by these ideas this section is organized as follows:
\begin{itemize}
 \item We need different weights $m^{j}$ for each step of the iterative adjustment. Note  that there is some freedom in choosing the starting point of our iteration. It does not necessarily have to be $\alpha(\Psi_0,\phi_0)$ but it should be close to a multiple of  $\alpha(\Psi_0,\phi_0)$.
Hence there are many possible choices for $m^{j}$.
All important properties the weights have to satisfy in order to generalize the Theorem are stated in Lemma \ref{mhut}. We prove the Lemma (i.e. the existence of a weight which satisfies all the important conditions) by construction.

Looking at (\ref{adjalpha}) one can already guess that starting with some weight $m^0$, $m^1$ has to satisfy $m^{1}(k,N)=m^{0}(k,N)-m^{0}(k+2,N)$. This explains (a) of Lemma \ref{mhut}. (b) ensures that the starting point of our iteration (which will be $\laa\Psi,\widehat{m}^0\Psi\raa$) is in fact close to a multiple of  $\alpha(\Psi_0,\phi_0)$. (b) and (c) of the Lemma are needed for the estimates.

\item Having defined the weights $m^{j}$ we construct some operators $R_{j,k}$, $S_{j,k}$ and $T_{j,k}$ (Definition \ref{RST}).
 These operators shall then be used to define the functionals $\gamma_{j,k}$ and $\xi_{j,k}$ (Definition \ref{defgammaxi}) as well as $\gamma_{j,k}'$ (Definition \ref{defgammastrich}) which are the basic elements of the iterative adjustment: \begin{itemize}
\item The $\gamma'_{j,k}$ are defined such, that $\dt\gamma_{j,k}(\Psi_t,\phi_t)=\gamma'_{j,k}(\Psi_t,\phi_t)$ (see Lemma \ref{firstadjlemma}).
\item
$\gamma_{0,0}(\Psi,\phi)=\laa\Psi,\widehat{m}^0\Psi\raa$ is the starting point of the iteration. $\xi_{0,0}$ plays the role of $\alpha_1'$.
Note that $\gamma_{0,0}'-\xi_{0,0}$ is small (Corollary \ref{newalpha2}).
\item
In the $l^{\text{th}}$ step of the adjustment we add $\sum_{j+k=l}\gamma_{j,k}$ to our functional.
The respective $\sum_{j+k=l}\gamma'_{j,k}$ cancels out the ``old'' term we have no sufficient control of (namely $\sum_{j+k=l-1}\xi_{j,k}$) but leads to a ``new'' term which --- as long as $j+k<5$ --- we have no sufficient control of (which is $\sum_{j+k=l}\xi_{j,k}$).
All other terms of $\sum_{j+k=l}\gamma'_{j,k}$ are controllable (Lemma \ref{firstadjest}) below.
\end{itemize}
\item
Finally we construct functionals $\Gamma$ and $\Gamma'$ with $\dt\Gamma(\Psi_t,\phi_t)=\Gamma'(\Psi_t,\phi_t)$ such that $\Gamma'(\Psi,\phi)$  can be estimated in terms of $\alpha(\Psi,\phi)$ (Corollary \ref{corproof2} (c)) and $\alpha(\Psi,\phi)$ can be estimated in terms of $\Gamma(\Psi,\phi)$ (Corollary \ref{corproof2} (b)). This allows to estimate $\Gamma(\Psi_t,\phi_t)$ via Gr\o nwall. Since $\alpha(\Psi,\phi)$ is controlled by $\Gamma(\Psi,\phi)$ (Corollary \ref{corproof2} (b)) we have good control of $\alpha(\Psi,\phi)$ for $\beta<1$.
\end{itemize}
\begin{lemma}\label{mhut}
There exists a set of weights $\{m^1,m^2,\ldots,m^5;m^j:\mathbb{N}^2\to\mathbb{R}^+ \}$

\begin{enumerate}
 \item For $0\leq  j<5$
$$m^{j+1}(k,N)=m^{j}(k,N)-m^{j}(k+2,N)$$
\item
For any $j>0$ there exist constants $c_j>0$ such that for any $k\geq 0$
\begin{align}\label{mabsch} c_j N^{-j} n^{1-2j}(k+2,N)  \leq  m^{j}(k,N)\leq&  N^{-j}n^{1-2j}(k+2,N)\;.
\end{align}

\item There exists a $C<\infty$ such that for $k\geq 0$
\begin{align*} |m^j(k,N)-m^j(k+1,N)|\leq&  C N^{-j-1}n^{-1-2j}(k+1,N)
\\ |m^j(k,N)-2m^j(k+1,N)+m^j(k+2,N)|\leq&  CN^{-j-2}n^{-3-2j}(k+1,N)\;.
\end{align*}

\end{enumerate}

\end{lemma}

\begin{proof}

Let $m^5(k,N):=N^{-1/2}(k+1)^{-5+1/2} $. We prove the Lemma by constructing $m^j$ for  $j\in\{0,1,\ldots,4\}$:
We define the functions $m^j$ for even $N+k$ via
\begin{align*} m^{j}(N,N)=&(N+2)^{-j}\\
m^{j}(k,N)=&m^{j+1}(k,N)+m^{j}(k+2,N)\;.\end{align*}
For odd $N+k$ we set $m^j(k)=(m^j(k-1)+m^j(k+1))/2$.
\begin{enumerate}
 \item

For even $N+k$ (a) follows by construction. For odd $N+k$ one has
\begin{align*} m^{j+1}(k,N)=&(m^{j+1}(k-1,N)+m^{j+1}(k+1,N))/2
\\=&(m^{j}(k-1,N)-m^{j}(k+1,N))/2\\&+(m^j(k+1,N)-m^j(k+3,N))/2
\\=&m^{j}(k,N)-m^{j}(k+2,N)\;.\end{align*}

\item It suffices to prove (b) for even $N+k$.  By construction it follows then also for odd $k+N$. We shall do so via
induction over $j$. For $j=5$ (\ref{mabsch}) follows directly from the definition of $m^5$.

Assume that (\ref{mabsch}) is satisfied for some $0<j\leq  5$. By construction we have for even $N+k$ that
\begin{align}
m^{j-1}(k,N)\nonumber=&m^{j-1}(N,N)+\sum_{k\leq  l< N}^{l \text{ even}} m^{j}(l,N) \\=&(N+2)^{-j}+\sum_{k\leq  l< N}^{l
\text{ even}} m^{j}(l,N)\label{nplus1}
\end{align}
By assumption there exist $c_j>0$, for later use we  assume that \be\label{cj}c_j<\frac{(2j-1)}{3^{j}}\;,\ee such that
\begin{align*}
 c_j\sum_{k\leq  l< N}^{l \text{ even}} (l+2)^{-j}n(l+2,N)\leq&  \sum_{k\leq  l< N}^{l \text{ even}}|m^{j}(l,N)|\\\leq&
\sum_{k\leq  l< N}^{l \text{ even}} (l+2)^{-j}n(l+2,N)\:.
\end{align*}

By monotonicity of the function $(\cdot)^{-j+1/2}$ for $j>0$ it follows that
\begin{align*}
c_j\frac{N^{-1/2}}{2}\int_{k}^{N-2} (x+2)^{-j+1/2}dx\leq&  \sum_{k\leq  l< N}^{l \text{ even}}|m^{j}(l,N)|\\\leq& \frac{N^{-1/2}}{2} \int_{k+2}^{N} (x+2)^{-j+1/2}dx
\end{align*}
\begin{align*}&\hspace{-2cm}\frac{c_jN^{-1/2}}{2j-1}\left((k+2)^{-j+1/2}-N^{-j+1/2}\right)
\\\leq& \sum_{k\leq  l< N}^{l \text{ even}}|m^{j}(l,N)|\leq  \frac{N^{-1/2}}{2j-1}(k+2)^{-j+1/2}\;.
\end{align*}
With (\ref{cj}) and  (\ref{nplus1}) we get for $N>0$
\begin{align*}&\hspace{-1cm}\frac{c_jN^{-1/2}}{2j-1}(k+2)^{-j+1/2}\\\leq& \frac{c_jN^{-1/2}}{2j-1}(k+2)^{-j+1/2}-\frac{c_jN^{-j}}{2j-1}+(N+2)^{-j}
\\\leq&  m^{j-1}(k,N)\leq  \frac{N^{-1/2}}{2j-1}(k+2)^{-j+1/2}\;.
\end{align*}
Since $(2j-1)^{-1}<1$ it follows that (\ref{mabsch}) holds for $j-1$.
Now (b) follows for even $N+k$ via induction.

\item Let us first prove (c) for even $N+k$. It follow that
\begin{align*}|m^j(k)-m^j(k+1)|=&|m^j(k)-(m^j(k)+m^j(k+2))/2|\\=&|m^j(k)-m^j(k+2)|/2=m^{j+1}(k)/2\;.\end{align*}
With (b) we get the first formula in (c). For the second formula we have
\begin{align*}|m^j(k)&-2m^j(k+1)+m^j(k+2)|\\=&|m^j(k)-(m^j(k)+m^j(k+2))+m^j(k+2)|=0\;.\end{align*}

For odd $N+k$ one has
\begin{align*}|m^j(k)-m^j(k+1)|=&|(m^j(k+1)+m^j(k-1))/2-m^j(k+1)|\\=&|m^j(k-1)-m^j(k+1)|/2=m^{j+1}(k-1)\end{align*}
and
\begin{align*}|m^j(k)&-2m^j(k+1)+m^j(k+2)|\\=&|(m^j(k+1)+m^j(k-1))/2-2m^j(k+1)\\&+(m^j(k+1)+m^j(k+3))/2|
\\=&\frac{1}{2}|m^j(k-1)-2m^j(k+1)+m^j(k+3)|\\=&\frac{1}{2}|m^{j+1}(k-1)-m^{j+1}(k+1)|=|m^{j+1}(k)|\;.\end{align*}
With (b) we get (c).
\end{enumerate}
\end{proof}

Having shown that there exist a weight satisfying the conditions of Lemma \ref{mhut} we use this weight to define some
operators that shall be used in our iterative adjustment of $\alpha$ below. Due to symmetry we have some arbitrariness in defining these operators, in particular in choosing the coordinates on which they act. For easier reference below we keep the coordinates $1,\ldots,4$ free, so the operators we shall define next act on the coordinates $x_5,x_6,\ldots$ only.
\begin{definition}\label{RST}
For any $ j,k\in\mathbb{N}_0$ we define the operators $Q_j$ via
\begin{align*} Q_j:=q_{2j+3}q_{2j+4}g_{1/4,\beta}(x_{2j+3}-x_{2j+4})p_{2j+3}p_{2j+4}
\end{align*}
and the operators $R_{j,k}$, $R_{j,k}$ and $R_{j,k}$ acting on $\LZN$ via
\begin{align*} R_{j,k}:=&\frac{N!}{(N-2j-2k)!k!j!}\prod_{l=1}^j Q_{l}\;\;\widehat{m}^{j+k}\prod_{l=1}^k Q^*_{l+j}
\\
S_{j,k}:=&\frac{N!}{(N-2j-2k-2)!j!k!}  \prod_{l=1}^j Q_{l}\;\;\widehat{m}^{j+k+1} \prod_{l=1}^k Q^*_{l+j}
\\
T_{j,k}:=&\frac{N!}{(N-2j-2k-2)!j!k!}  \prod_{l=1}^j Q_{l}\;\;(\widehat{m}^{j+k}-\widehat{m}^{j+k}_1) \prod_{l=1}^k Q^*_{l+j}\;.
\end{align*}
For later  use it is convenient to define $S_{-1,k}:=0$.

\end{definition}

\begin{lemma}\label{Rnorm}

Let $j,k\in\mathbb{N}_0$, $\mathcal{M}\subset\mathbb{N}$ with $|\mathcal{M}\cap\{5,6,\ldots,2j+2k+4\}|=M$. Then
\begin{align*}\|R_{j,k}\|_{\mathcal{M}}\leq  & C N^{-(j+k)/8}\|\phi\|_\infty^{j+k}\;\;&\text{ if }\mathcal{M}\cap\{5,6,\ldots,2j+2k+4\}=\emptyset\;,\\
\|R_{j,k}\|_{\mathcal{M}} \leq & N^{-1/2+M/2-(j+k)/8}\|\phi\|_\infty^{j+k}\;\;&\text{ if }\mathcal{M}\cap\{5,6,\ldots,2j+2k+4\}\neq\emptyset\;.
\end{align*}
Let furthermore $r:\mathbb{N}^2\to\mathbb{R}^+$ with $r\leq  n$, then
$$\|\widehat{r}\,S_{j,k}\|_{\mathcal{M}} \leq  C N^{1+M/2-(j+k)/8}\|\phi\|_\infty^{j+k}$$
$$\|\widehat{r}\,T_{j,k}\|_{\mathcal{M}} \leq  C N^{1+M/2-(j+k)/8}\|\phi\|_\infty^{j+k}\;.$$
\end{lemma}
\begin{proof}
First note that in view of Lemma \ref{kombinatorik} (e) and Lemma \ref{defAlemma} (b)
$$\|g_{1/4,\beta}(x_1-x_2)p_1\|_{op}\leq  C N^{-1-1/8}\|\phi\|_\infty\;.$$

Roughly estimating, $\|R_{j,k}\|_{op}$ has $j+k$ such factors, giving a power $N^{-(j+k)(1+1/8)}$. Furthermore we have
$2j+2k$ projectors $q$ in the definition of $R_{j,k}$, while $m^{j+k}$ is of order $n(k)k^{-j-k}$. So the projectors
$q$ give in view of Lemma \ref{kombinatorikb} together with $\widehat{m}^{j+k}$ a factor $N^{-j-k}$. Since
$\frac{N!}{(N-2j-2k)!}<N^{2j+2k}$ one gets by this rough estimate the result above.

In detail:
Let $\mathcal{M}\subset\mathbb{N}$.
Let $\Psi,\chi\in \mathcal{H}_{\mathcal{M}}$; $\|\Psi\|=\|\chi\|=1$. Using Lemma \ref{kombinatorik} (c)
 \begin{align*}
|\laa\Psi, R_{j,k}\chi\raa|
=&\frac{N!}{j!k!(N-2j-2k)!}\\& |\laa\Psi,
\widehat{n}^{-2j+1}_1\prod_{l=1}^{j}Q_{l}\;\;\widehat{n}^{2j-1}_{2j+1}\widehat{m}^{j+k}\widehat{n}^{2k}_{2k+1}
\prod_{l=1}^{k}Q^*_{l+j}\widehat{n}^{-2k}_1
\chi\raa|
\end{align*}
With Lemma \ref{mhut} (b) we have that $\|\widehat{n}^{2j-1}_{2j+1}\widehat{m}^{j+k}\widehat{n}^{2k}_{2k+1}\|_{op}\leq  CN^{-j-k}$, thus we get with Lemma \ref{kombinatorikb}
\begin{align*}
|\laa\Psi, R_{j,k}\chi\raa|\leq&  CN^{2j+2k} N^{-j-k}\|\widehat{n}^{-2j+1}_1\prod_{l=5}^{2j+4}q_l\Psi\|\;
\|\widehat{n}^{-2k}_1 \prod_{l=5+2j}^{2j+2k+4}q_l\chi\|
\\&\|g_{1/4,\beta}(x_1-x_2)p_1\|_{op}^{j+k}
\end{align*}
The estimate on $\|\widehat{n}^{-2j+1}_1\prod_{l=5}^{2j+4}q_l\Psi\|$ and $\|\widehat{n}^{-2k+1} \prod_{l=5+2j}^{2j+2k+4}q_l\chi\|$ now depends on the symmetry of $\Psi$ respectively $\chi$. If $\mathcal{M}\cap\{5,6,\ldots,2j+2k+4\}=\emptyset$ we can use Lemma \ref{kombinatorikb} for all $q_j$ with $j=5,\ldots,2j+2k+4$. Then it follows that $\|\widehat{n}^{-2j+1}_1\prod_{l=5}^{2j+4}q_l\Psi\|<C$ and $\|\widehat{n}^{-2k+1} \prod_{l=5+2j}^{2j+2k+4}q_l\chi\|<C$ and thus
\begin{align*}
|\laa\Psi, R_{j,k}\chi\raa|
\leq  CN^{-(j+k)/8}\|\phi\|_\infty^{j+k}\;.
\end{align*}
If $\mathcal{M}\subset\mathbb{N}$  with $|\mathcal{M}\cap\{5,6,\ldots,2j+2k+4\}|=M>0$ we can define
$\mathcal{M}_a:=\mathcal{M}\cap \{5,6,\ldots,2j+4\}$
and $\mathcal{M}_b:=\mathcal{M}\cap \{2j+5,2j+6,\ldots,2j+2k+4\}$ and assume without loss of generality that
$|\mathcal{M}_a|>0$. Then it follows that $\|\widehat{n}^{-2j+1}_1\prod_{l=5}^{2j+4}q_l\Psi\|<CN^{(|M_a|-1)/2}$ and $\|\widehat{n}^{-2k+1} \prod_{l=5+2j}^{2j+2k+4}q_l\chi\|<CN^{|M_b|/2}$ and thus
\begin{align*}
|\laa\Psi, R_{j,k}\chi\raa|
\leq  CN^{-(j+k)/8+(M-1)/2}\|\phi\|_\infty^{j+k}\;.
\end{align*}

$\widehat{r}S_{j,k}$ can be estimated in a similar way by
\begin{align*}
|\laa\Psi,& \widehat{r}\,S_{j,k}\chi\raa|\\\leq&   N^{2j+2k+2}
|\laa\Psi,
\widehat{r}\widehat{n}^{-2j-1}_1\prod_{l=1}^{j}Q_{l}\;\;\widehat{n}^{2j+1}_{2j+1}\widehat{m}^{j+k+1}\widehat{n}^{2k}_{2k+1}
\prod_{l=1}^{k}Q^*_{l+j}\widehat{n}^{-2k}_1
\chi\raa|
\\\leq& C N^{2j+2k+2-j-k-1}\|\widehat{r}\,\widehat{n}^{-2j-1}_{1}
\prod_{l=5}^{2j+4}q_l\Psi\|\;\prod_{l=5+2k}^{2j+2k+4}q_l\chi\|\\& \|g_{1/4,\beta}(x_1-x_2)p_1\|_{op}^{j+k} \|\widehat{n}^{-2k}_{1}\;.
\end{align*}
Using  Lemma \ref{kombinatorikb} $$\|\widehat{r}\widehat{n}^{-2j-1}_1\prod_{l=5}^{2j+4}q_l\Psi\|<CN^{|M_a|/2}\;\;\;\;\text{ and }\;\;\;\; \|\widehat{n}^{-2k+1} \prod_{l=5+2j}^{2j+2k+4}q_l\chi\|<CN^{|M_b|/2}\;,$$ thus
\begin{align*}
|\laa\Psi, \widehat{r}\,S_{j,k}\chi\raa|\leq&   C N^{j+k+1}N^{|M_a|/2}\|\phi\|_\infty^{j+k}N^{-9(j+k)/8}\\=&C N^{1+M/2-(j+k)/8}\|\phi\|_\infty^{j+k}\;.
\end{align*}

For the last equation  note  that in view of Lemma \ref{mhut} (c) $\widehat{m}^j-\widehat{m}^j_1\leq  \widehat{m}^j-\widehat{m}^j_2$, hence we get the same
estimate as for $\|\widehat{n}S_{j,k}\|_{\mathcal{M}}$.

\end{proof}

Using the operators defined in Definition \ref{RST} we now adjust the functional $\alpha$ as explained at the beginning of this section using the functionals $\gamma_{j,k}$ which we shall define next.
\begin{definition}\label{defgammaxi}
For any $j,k>0$ with $j+k\leq  5$ we define
\begin{align*}
\gamma_{j,k}(\Psi,\phi):=&\laa\Psi,R_{j,k}\Psi\raa
\\
\xi_{j,k}(\Psi,\phi):=&\laa\Psi,\potdiff_{\beta}(x_1,x_2)p_1p_2S_{j,k}\Psi\raa\text{ for }j,k\geq 0
\end{align*}

\end{definition}

As explained above we have after the $l^{\text{th}}$ step of our iteration a remainder $\sum_{j+k=l}\xi_{j,k}$. We wish
to show that this remainder is controllable after sufficiently many steps of iteration. It turns out that five steps
are enough:
\begin{proposition}\label{xiest}(Control of the remainder) There exists a $\kinf$ such that
 $$|\xi_{j,k}|\leq  CN^{1/2-(j+k)/8}\|\phi\|_\infty\mathcal{K}(\phi)\;.$$
\end{proposition}

\begin{proof}
Let us first prove the following three formulas which shall also be of use below.
\begin{align}\label{hilfe1}
\|\sqrt{|\potdiff_\beta(x_1,x_2)|}p_1\|_{op}\leq& C N^{-1/2} \|\phi\|_\infty
\\\label{hilfe2}
\|\sqrt{|\potdiff_\beta(x_1,x_2)|}\Psi\|\leq& C N^{-1/2}\left(1+\sqrt{\alpha(\Psi,\phi)}+\|\nabla\phi\|+\|\phi\|_\infty\right)
\\\label{hilfe3}
\|p_1\potdiff_\beta(x_1,x_2)\Psi\|\leq& C N^{-1} \|\phi\|_\infty \left(1+\sqrt{\alpha(\Psi,\phi)}+\|\nabla\phi\|+\|\phi\|_\infty\right)
\end{align}
(\ref{hilfe1}) follows from Lemma \ref{hnorms} (a) together with Lemma \ref{kombinatorik} (e):
\begin{align*}\|\sqrt{|\potdiff_\beta(x_1,x_2)|}p_1\|_{op}^2=&\|p_1\potdiff_\beta(x_1,x_2)p_1\|_{op}
\\&\hspace{-2cm}\leq \|p_1V_\beta(x_1-x_2)p_1\|_{op}+\frac{2a}{N-1}\|p_1(|\phi(x_1)|^2+|\phi(x_2)|^2)p_1\|_{op}
\\&\hspace{-2cm}\leq  CN^{-1}\|\phi\|_\infty^2\;.
\end{align*}
(\ref{hilfe2}) follows from Lemma \ref{totalE} together with Lemma \ref{hnorms} (a)
\begin{align*} \|\sqrt{|\potdiff_\beta(x_1,x_2)|}\Psi\|^2\leq&  \laa\Psi,(V_\beta(x_1-x_2)+\frac{2a}{N-1}|\phi(x_1)|^2+\frac{2a}{N-1}|\phi(x_2)|^2)\Psi\raa
\\\leq& N^{-1}\alpha(\Psi,\phi)+CN^{-1}\left(1+\|\nabla\phi\|^2+\|\phi\|_\infty^2\right)
\end{align*}
and (\ref{hilfe3}) is a direct consequence of (\ref{hilfe1}) and (\ref{hilfe2}).

 It follows that
 \begin{align*}|\xi_{j,k}|\leq&  \|p_1\potdiff_{\beta}(x_1,x_2)\Psi\|\;\|\widehat{n}_1^{-1}\|_{op} \|\widehat{n}_1S_{j,k}\|_{\{1,2\}}
\\\leq&  CN^{1/2-(j+k)/8}\|\phi\|_\infty^{j+k+1}\left(1+\sqrt{\alpha(\Psi,\phi)}+\|\nabla\phi\|+\|\phi\|_\infty\right)\;.\end{align*}
\end{proof}

\begin{definition}\label{defgammastrich}
For any $j,k>0$ let the functional $\gamma'_{j,k}:\LZN\otimes\LZ\to\mathbb{R}^+$ be given by
 $$\gamma'_{j,k}:=-2\Im\left(\sum_{j+k=l}\ajk+\bjk+\cjk+\frac{1}{4}\djk+\ejk+\fjk-\xi_{j-1,k}-\frac{1}{2}\xi_{j,k}\right)
\;.$$
where  the different summands are
\begin{enumerate}

  \item [(a)] The mixed derivative term
\begin{align*}\ajk(\Psi,\phi):=&j\laa\Psi,q_{1}q_{2}[H,g(x_{1}-x_{2})]p_{1}p_{2}S_{j-1,k}\Psi\raa
\\&+j\laa\Psi,q_{1}q_{2}((W_{1/4}-V_{\beta})f_{1/4,\beta}(x_1-x_2)p_{1}p_{2}S_{j-1,k}\Psi\raa
\;.\end{align*}

\item [(b)] The smoothed out interaction term
\begin{align*}\bjk(\Psi,\phi):=&-j\laa\Psi,q_{1}q_{2}((W_{1/4}-V_{\beta})f_{1/4,\beta}(x_1-x_2)p_{1}p_{2}S_{j-1,k}\Psi\raa
\\&-\xi_{j-1,k}(\Psi,\phi)
+j\laa\Psi ,\left[\potdiff_\beta(x_5,x_6),R_{j,k}\right] \Psi\raa
\end{align*}
\item [(c)] Three particle interactions
\begin{align*} \cjk(\Psi,\phi):=(N-2j-2k)j\laa\Psi , \left[\potdiff_\beta(x_1,x_5),R_{j,k}\right] \Psi\raa
\end{align*}

\item [(d)] Interaction terms of the correction first type
\begin{align*} \djk(\Psi,\phi):=&(N-2j-2k)(N-2j-2k-1)\laa\Psi ,[\potdiff_\beta(x_1,x_2),
R_{j,k}] \Psi\raa\\&-\xi_{j,k}(\Psi,\phi)+\xi_{k,j}^*(\Psi,\phi)\end{align*}

\item [(e)] Interaction terms of the correction second type
\begin{align*}&\ejk(\Psi,\phi)=j(j-1)\laa\Psi , \left[\potdiff_\beta(x_5,x_7),R_{j,k}\right] \Psi\raa
\end{align*}

\item [(f)] Interaction terms of the correction third type
\begin{align*}&\fjk(\Psi,\phi)=jk\Im\left(\laa\Psi,[\potdiff_\beta(x_5,x_{2j+5}),R_{j,k}]\Psi\raa\right)
\end{align*}

\end{enumerate}

\end{definition}

\begin{lemma}\label{firstadjlemma}

For all $0\leq  l \leq  5$
$$\sum_{j+k=l}\dt\gamma_{j,k}(\Psi_t,\phi_t)=\sum_{j+k=l}\gamma'_{j,k}(\Psi_t,\phi_t) \;. $$

\end{lemma}

\begin{proof} First note, that the $R_{j,k}$ are time dependent, since the operators $\widehat{m}^j$ depend on
$\phi_t$. The time derivative of $Q_j$ is
\begin{align*}\dot Q_j=&-\im[H^{GP},Q_j]+\im q_{2j+3}q_{2j+4}[H^{GP},g_{1/4,\beta}(x_{2j+3}-x_{2j+4})]p_{2j+3}p_{2j+4}
\end{align*}
thus by symmetry
\begin{align*}
\laa\Psi_t,(\dot R_{j,k})\Psi_t\raa=&-\im\laa\Psi_t,[H^{GP}, R_{j,k}]\Psi_t\raa
\\&+\im j\laa\Psi_t,q_{1}q_{2}[H^{GP},g_{1/4,\beta}(x_{1}-x_{2})]p_{1}p_{2}S_{j-1,k}\Psi_t\raa
\\&+\im k\laa\Psi_t,S_{j,k-1}p_{1}p_{2}[H^{GP},g_{1/4,\beta}(x_{1}-x_{2})]q_{1}q_{2}\Psi_t\raa\;.
\end{align*}
Note that after exchanging some variables the adjoint of $S_{j,k}$ equals $S_{k,j}$. Using symmetry and changing the label $k\to j$ in the last line
\begin{align*}
\dot\gamma_{j,k}(\Psi_t,\phi_t)=&\im\laa\Psi_t,[H-H^{GP},R_{j,k}]\Psi_t\raa\\&+\im
j\laa\Psi_t,q_{1}q_{2}[H^{GP},g_{1/4,\beta}(x_{1}-x_{2})]p_{1}p_{2}S_{j-1,k}\Psi_t\raa
\\&-\im j\laa\Psi_t,q_{1}q_{2}[H^{GP},g_{1/4,\beta}(x_{1}-x_{2})]p_{1}p_{2}S_{j-1,k}\Psi_t\raa^*\;.
\end{align*}
So the Lemma follows once we have shown that
\begin{align}\label{toprove}
\sum_{j+k=l}\gamma'_{j,k}(\Psi,\phi)=&\sum_{j+k=l}\im\laa\Psi,[\sum_{1\leq  l<m\leq
N}\potdiff_\beta(x_l,x_m),R_{j,k}]\Psi\raa
\nonumber\\&-2j\Im \left(\laa\Psi,q_{1}q_{2}[H^{GP},g_{1/4,\beta}(x_{1}-x_{2})]p_{1}p_{2}S_{j-1,k}\Psi\raa\right)
%
\;.
\end{align}

As above we want to get rid of the sum $1\leq  l<m\leq  N$ using that many summands are equal because of symmetry.
$R_{j,k}=\frac{N!}{(N-2j-2k)!} A_j\widehat{m}^{j+k}B_{j,k}$ breaks some of the symmetry but it is still symmetric in exchanging
any two variables with indices in $\mathcal{M}_a=\{5,6,\ldots,2j+4\}$ as well as in exchanging
any two variables with indices  in $\mathcal{M}_b=\{2j+5,2j+6,\ldots,2j+2k+4\}$ and in exchanging
any two variables with indices  in $\mathcal{M}_c=\{1,2,3,4,2j+2k+5.2j+2k+6,\ldots,N\}$.

We arrive at three different cases for the variable $x_l$: $x_l\in\mathcal{M}_a$, $x_l\in\mathcal{M}_b$ and $x_l\in\mathcal{M}_c$.
For the case $x_l\in\mathcal{M}_c$ we arrive at three different cases for the variable $x_m$. For the case $x_l\in\mathcal{M}_a$
more symmetry is broken via the factor $q_{l}q_{l\pm1}g_{1/4,\beta}(x_{l}-x_{l\pm1})p_{l}p_{l\pm1}$  (+ if $l$ is odd, $-$ if l is even)
appearing in $R_{j,k}$ (see definition \ref{RST}).
Similar for the case the case $x_l\in\mathcal{M}_b$. Hence we arrive at the following eight different summands:

\begin{align*}
\im\laa\Psi,&[\sum_{1\leq  l<m\leq
N}\potdiff_\beta(x_l,x_m),R_{j,k}]\Psi\raa
\\=&\frac{\im}{2}(N-2j-2k)(N-2j-2k-1)\laa\Psi,[\potdiff_\beta(x_1,x_2),R_{j,k}]\Psi\raa
\\&+\im(N-2j-2k)j\laa\Psi,[\potdiff_\beta(x_1,x_5),R_{j,k}]\Psi\raa
\\&+\im j\laa\Psi,[\potdiff_\beta(x_5,x_6),R_{j,k}]\Psi\raa
\\&+\im(N-2j-2k)k\laa\Psi,[\potdiff_\beta(x_1,x_{2j+5}),R_{j,k}]\Psi\raa
\\&+\im k\laa\Psi,[\potdiff_\beta(x_{2j+5},x_{2j+6}),R_{j,k}]\Psi\raa
\\&+\im j(2j-1)\laa\Psi,[\potdiff_\beta(x_5,x_7),R_{j,k}]\Psi\raa
\\&+\im j2k\laa\Psi,[\potdiff_\beta(x_5,x_{2j+5}),R_{j,k}]\Psi\raa
\\&+\im k(2k-1)\laa\Psi,[\potdiff_\beta(x_{2j+5},x_{2j+7}),R_{j,k}]\Psi\raa\;.
\end{align*}
Using  that  after exchanging some variables the adjoint of $R_{j,k}$ equals $R_{k,j}$
\begin{align*}
\im\laa\Psi,&[\sum_{1\leq  l<m\leq
N}\potdiff_\beta(x_l,x_m),R_{j,k}]\Psi\raa
\\=&\frac{\im}{4}(N-2j-2k)(N-2j-2k-1)\laa\Psi,[\potdiff_\beta(x_1,x_2),R_{j,k}]\Psi\raa
\\&-\frac{\im}{4}(N-2j-2k)(N-2j-2k-1)\laa\Psi,[\potdiff_\beta(x_1,x_2),R_{k,j}]\Psi\raa^*
\\&+\im(N-2j-2k)j\laa\Psi,[\potdiff_\beta(x_1,x_5),R_{j,k}]\Psi\raa
\\&+\im j\laa\Psi,[\potdiff_\beta(x_5,x_6),R_{j,k}]\Psi\raa
\\&-\im(N-2j-2k)k\laa\Psi,[\potdiff_\beta(x_1,x_{5}),R_{k,j}]\Psi\raa^*
\\&-\im k\laa\Psi,[\potdiff_\beta(x_5,x_6),R_{k,j}]\Psi\raa^*
\\&+\im j(2j-1)\laa\Psi,[\potdiff_\beta(x_5,x_7),R_{j,k}]\Psi\raa
\\&+\im jk\laa\Psi,[\potdiff_\beta(x_5,x_{2j+5}),R_{j,k}]\Psi\raa
\\&-\im jk\laa\Psi,[\potdiff_\beta(x_{2j+5},x_5),R_{k,j}]\Psi\raa^*
\\&-\im k(2k-1)\laa\Psi,[\potdiff_\beta(x_{5},x_{7}),R_{k,j}]\Psi\raa^*
\end{align*}
It follows that $\sum_{j+k=l}\gamma'_{j,k}(\Psi,\phi)$ equals
\begin{align*}
=&-\frac{1}{2}\sum_{j+k=l}(N-2j-2k)(N-2j-2k-1)\Im\left(\laa\Psi,[\potdiff_\beta(x_1,x_2),R_{j,k}]\Psi\raa\right)
\\&-2\sum_{j+k=l} (N-2j-2k)j\Im\left(\laa\Psi,[\potdiff_\beta(x_1,x_5),R_{j,k}]\Psi\raa\right)
\\&-2\sum_{j+k=l} j\Im\left(\laa\Psi,[\potdiff_\beta(x_5,x_6),R_{j,k}]\Psi\raa\right)
\\&-2\sum_{j+k=l} j(2j-1)\Im\left(\laa\Psi,[\potdiff_\beta(x_5,x_7),R_{j,k}]\Psi\raa\right)
\\&-2\sum_{j+k=l} jk\Im\left(\laa\Psi,[\potdiff_\beta(x_5,x_{2j+5}),R_{j,k}]\Psi\raa\right)
\\&-2\sum_{j+k=l} j\Im\left(\laa\Psi,q_{1}q_{2}[H,g_{1/4,\beta}(x_{1}-x_{2})]p_{1}p_{2}S_{j-1,k}\Psi\raa\right)\;.
\end{align*}
Adding $$-2\sum_{j+k=l}j\Im\left(\laa\Psi,q_{1}q_{2}((W_{1/4}-V_{\beta})f_{1/4,\beta}(x_1-x_2)p_{1}p_{2}S_{j-1,k}\Psi\raa\right)$$
to the  last line
 and subtracting it form the third line, as well as subtracting $-2\sum_{j+k=l}\Im(\xi_{j-1,k})$ from the third line and adding it to the total and subtracting $-2\sum_{j+k=l}\Im(\xi_{j,k})/2=-2\sum_{j+k=l}\Im((\xi_{j,k}-\xi_{k,j}^*)/4)$ from the first line and adding it to the total gives that the right hand side of (\ref{toprove}) equals
\begin{align*}
-2\Im\left(\sum_{j+k=l}\frac{1}{4}\djk+\cjk+\bjk+\ejk+\fjk+\ajk+\xi_{j-1,k}+\frac{1}{2}\xi_{j,k}\right)
\end{align*}
which proves the Lemma.
\end{proof}
Having proven that the functionals $\gamma'$ can be understood as the time derivative of the functionals $\gamma$ our next step is to control the functionals $\gamma$. To start with $\gamma_{0,0}$.
\begin{corollary}\label{newalpha2} Let $\beta<1$. Then there exists a $\kinf$ and a $\eta>0$ such that for
$\xi:=\gamma'_{0,0}-\Im\left(\xi_{0,0}\right)$
$$\xi(\Psi,\phi) \leq \mathcal{K}(\phi)(\|\phi\|_\infty+(\ln N)^{1/3}\|\nabla\phi\|_{6,loc} )
(\laa\Psi,\widehat{n}\Psi\raa+\|\nabla_1q_1\Psi\|^2+N^{-\eta})\;.$$
\end{corollary}
\begin{proof}
By Definition \ref{defgammastrich} one sees, that for $\gamma'_{0,0}-\Im\left(\xi_{0,0}\right)$ only (d) remains. And here the problematic term in fact cancels out
In view of Lemma \ref{kombinatorik} (d) and Lemma \ref{mhut} (a)
\begin{align}\nonumber\xi(\Psi,\phi)=&N(N-1)\left(
\laa\Psi,[\potdiff_{\beta}(x_1,x_2),\widehat{m}^0]\Psi\raa-\laa\Psi,[\potdiff_{\beta}(x_1,x_2),p_1p_2\widehat{m}^1]\Psi\raa\right)
\\=&\label{xiformel}2N(N-1)
\laa\Psi,[\potdiff_{\beta}(x_1,x_2),p_1q_2(\widehat{m}^0-\widehat{m}^0_1)]\Psi\raa
%
\\\nonumber=&2N(N-1)\Im\left(\laa\Psi,p_1p_2\potdiff_{\beta}(x_1,x_2)p_1q_2(\widehat{m}^0-\widehat{m}^0_1)\Psi\raa\right)
\\\nonumber&+2N(N-1)\Im\left(\laa\Psi,q_1q_2\potdiff_{\beta}(x_1,x_2)p_1q_2(\widehat{m}^0-\widehat{m}^0_1)\Psi\raa\right)
\;.
\end{align}
Since $\beta<1$ this is controlled by Lemma \ref{hnorms} (b) and (d) using the bounds from Lemma \ref{mhut} (c).
\end{proof}

\begin{lemma}\label{firstadjest}
For any $1/3\leq \beta\leq  1$, $l>0$ there exists a $\kinf$, $\eta>0$ such that
\begin{align*}&\sum_{j+k=l} \gamma'_{j,k}(\Psi,\phi)+2\Im\left(\xi_{j-1,k}(\Psi,\phi)\right)+\Im\left(\xi_{j,k}(\Psi,\phi)\right)
\\&\hspace{2cm}\leq  \Cphi\C\;.\end{align*}

\end{lemma}

\begin{proof}

To prove the Lemma we shall estimate the imaginary parts of $\ajk,\bjk,\ldots,\fjk$ separately.

\begin{enumerate}

 \item [(a)] The commutator in $\ajk$ equals
\begin{align}\label{commu}[H,g_{1/4,\beta}(x_1-x_2)]=&-[H,f_{1/4,\beta}(x_1-x_2)]\\\nonumber=&[\Delta_1+\Delta_2,f_{1/4,\beta}(x_1-x_2)]
\\\nonumber=&(\Delta_1+\Delta_2)f_{1/4,\beta}(x_1-x_2)\\\nonumber&+(\nabla_1f_{1/4,\beta}(x_1-x_2))\nabla_1-(\nabla_2f_{1/4,\beta}(x_1-x_2))\nabla_2
\\\nonumber=&-(W_{1/4}-V_{\beta})f_{1/4,\beta}(x_1-x_2)\\\nonumber&+(\nabla_1g_{1/4,\beta}(x_1-x_2))\nabla_1-(\nabla_2g_{1/4,\beta}(x_1-x_2))\nabla_2\;.
\end{align}
This and integration by parts gives
\begin{align*}
|\ajk(\Psi,\phi)|\leq& 2j|\laa\Psi ,q_{1}q_{2}(\nabla_1g_{1/4,\beta}(x_1-x_2))\nabla_1
p_1p_2S_{j-1,k}\Psi\raa|
\nonumber\\\leq& 2j|\laa\nabla_1q_{1}q_{2}\widehat{n}^{-1}_1\Psi ,(g_{1/4,\beta}(x_1-x_2))\nabla_1
p_1p_2\widehat{n}_3S_{j-1,k}\Psi\raa
\nonumber\\&+2j|\laa\Psi ,\widehat{n}^{-1}_1q_{1}q_{2}(g_{1/4,\beta}(x_1-x_2))\Delta_1
p_1p_2\widehat{n}_3S_{j-1,k}\Psi\raa|
\nonumber\\&\hspace{-1.7cm}\leq \;\;\;2j\|\nabla_1q_1q_{2}\widehat{n}^{-1}_1\Psi\|\;\|\nabla\phi\|\;\|g_{1/4,\beta}(x_1-x_2))p_2\|_{op}
\|\widehat{n}_3S_{j-1,k}\|_{\{1,2\}}
\\&\hspace{-1cm}+2j\|\widehat{n}^{-1}_1q_{1}\Psi\|\;\|\Delta\phi\|\; \|g_{1/4,\beta}(x_1-x_2))p_2\|_{op}\|q_{2}S_{j-1,k}\|_{\{1,2\}}\:.
\end{align*}
Using (\ref{nablamitm}), Lemma \ref{kombinatorik} (e) and Lemma \ref{kombinatorikb}  the latter is bounded by
$$CN^{-(j+k)/8}\|\phi\|_\infty^{j+k}(\|\nabla\phi\|+\|\Delta\phi\|)\;.$$

\item [(b)] Using symmetry it follows that
$$\laa\Psi,\potdiff_\beta(x_5,x_6)R_{j,k}\Psi\raa
=\laa\Psi, \potdiff_\beta(x_1,x_2)q_{1}q_{2}g_{1/4,\beta}(x_{1}-x_{2})p_{1}p_{2}S_{j-1,k}\Psi\raa$$
It follows that
\begin{align*}
\bjk(\Psi,\phi)=&-j\laa\Psi,q_{1}q_{2}((W_{1/4}-V_{\beta})f_{1/4,\beta}(x_1-x_2)p_{1}p_{2}S_{j-1,k}\Psi\raa
\\&-j\laa\Psi,\potdiff_{\beta}(x_1,x_2)p_1p_2S_{j-1,k}\Psi\raa
\\&+j\laa\Psi ,\potdiff_\beta(x_1,x_2)q_{1}q_{2}g_{1/4,\beta}(x_{1}-x_{2})p_{1}p_{2}S_{j-1,k}\Psi\raa
\\&-j\laa\Psi ,R_{j,k}\potdiff_\beta(x_5,x_6)\Psi\raa
\\=&-j\laa\Psi,q_{1}q_{2}\left((W_{1/4}-V_{\beta})f_{1/4,\beta}\right)(x_1-x_2)p_{1}p_{2}S_{j-1,k}\Psi\raa
\\&-j\laa\Psi,\potdiff_{\beta}(x_1,x_2)f_{1/4,\beta}(x_1-x_2)p_1p_2S_{j-1,k}\Psi\raa
\\&-j\laa\Psi ,\potdiff_{\beta}(x_1,x_2)(1-q_{1}q_{2})g_{1/4,\beta}(x_{1}-x_{2})p_{1}p_{2}S_{j-1,k}\Psi\raa
\\&-j\laa\Psi ,R_{j,k}\potdiff_\beta(x_5,x_6)\Psi\raa
\end{align*}
and thus
\begin{align}
\nonumber|&\sum_{j+k=l}\Im(\bjk(\Psi,\phi))|
\\\leq&\label{bjka}\sum_{j+k=l}j|\laa\Psi,q_{1}q_{2}(W_{1/4}(x_1-x_2)-\frac{4a}{N-1}|\phi(x_1)^2|)\\\nonumber&\hspace{5cm}f_{1/4,\beta}(x_1-x_2)p_{1}p_{2}S_{j-1,k}\Psi\raa|
\\&\label{bjkb}+|\sum_{j+k=l}j\;\Im\laa\Psi ,p_1p_2\potdiff_\beta(x_1,x_2) f_{1/4,\beta}(x_{1}-x_{2})p_{1}p_{2}S_{j-1,k}\Psi\raa|
\\&\label{bjkc}+\sum_{j+k=l}2j|\laa\Psi ,p_1q_2\potdiff_\beta(x_1,x_2) f_{1/4,\beta}(x_{1}-x_{2})p_{1}p_{2}S_{j-1,k}\Psi\raa|
\\&+\label{bjkd}\sum_{j+k=l}j|\laa\Psi ,\potdiff_\beta(x_1,x_2)(1-q_{1}q_{2})g_{1/4,\beta}(x_{1}-x_{2})p_{1}p_{2}S_{j-1,k}\Psi\raa|
\\&+\label{bjke}\sum_{j+k=l}j|\laa\Psi ,R_{j,k}\potdiff_{\beta}(x_5,x_6) \Psi\raa|
\;.
\end{align} If $(j,k)=(1,0)$ we get in view of Lemma \ref{hnorms} (c) (recall than $Nm^0\leq  Cn^{-1}$) that there exists a $\kinf$
\begin{align*}
|(\ref{bjka})|\leq&
\sum_{j+k=l}jN^2|\laa\Psi,q_{1}q_{2}W_{1/4}f_{1/4,\beta}(x_1-x_2)p_{1}p_{2}\widehat{m}^0\Psi\raa|
\\&+\sum_{j+k=l}jN^2|\laa\Psi,q_{1}q_{2}\frac{4a}{N-1}|\phi(x_1)^2|g_{1/4,\beta}(x_1-x_2)p_{1}p_{2}\widehat{m}^0\Psi\raa|
\\\leq&  \mathcal{K}(\phi)\|\phi\|_\infty(\laa\Psi,\widehat{n}\Psi\raa+CN^{-\eta})\|\phi\|_\infty^3\;.
\end{align*}
For $(j+k)>1$ we shall use  Lemma \ref{hnorms} (c) to control $(\ref{bjka})$. Note, that although for $\chi:=S_{j-1,k}\Psi$ some symmetry is broken,  still
$\chi\in\mathcal{H}_{\mathcal{M}}$ for  $\mathcal{M}=\{5,6,\ldots,2j+2k+4\}$.
We write
\begin{align*}
|(\ref{bjka})|=&\sum_{j+k=l}j|\laa\Psi,\widehat{n}^{-1}_1q_{1}q_{2}(W_{1/4}(x_1-x_2)-\frac{4a}{N-1}|\phi(x_1)^2|)\\&f_{1/4,\beta}(x_1-x_2)\widehat{n}_3p_{1}p_{2}S_{j-1,k}\Psi\raa|
\\\leq&
\sum_{j+k=l}j\|\widehat{n}^{-1}_1q_{1}q_{2}W_{1/4}f_{1/4,\beta}(x_1-x_2)p_{1}p_{2}\|_{\mathcal{M}}\|\widehat{n}_3S_{j-1,k}\|_{\{1,2\}}
\\&+\sum_{j+k=l}\frac{2aj}{N-1}\|\widehat{n}^{-1}_1q_{1}\Psi\|\;\|\phi\|_\infty^2 \|\widehat{n}_3 S_{j-1,k}\|_{\{1,2\}}\;.
\end{align*}
Due to Lemma  Lemma \ref{defAlemma} $W_{1/4} f_{1/4,\beta}\in\mathcal{V}_{1/4}$.
Thus Lemma \ref{hnorms} (c) and Lemma \ref{Rnorm} imply that $|(\ref{bjka})|$ has the right bound.

Note that $S_{j,k}$ becomes after exchanging some of the variables the adjoint of $S_{k,j}$ selfadjoint, thus $(\ref{bjkb})=0$ for all $l$.

For $(\ref{bjkc})$ we use Lemma \ref{hnorms} (b) and Lemma \ref{Rnorm}
\begin{align*}
|(\ref{bjkc})|\leq& \sum_{j+k=l}2j|\laa\Psi ,\widehat{n}^{-1}_1p_1q_2(V_\beta f_{1/4,\beta}(x_{1}-x_{2})-\frac{2a}{N-1}|\phi(x_1)|^2%
\\&\hspace{3cm}-\frac{2a}{N-1}|\phi(x_2)|^2 )p_{1}p_{2}\widehat{n}_3S_{j-1,k}\Psi\raa|
\\&+2|\laa\Psi ,\widehat{n}^{-1}_1p_1q_2g_{1/4,\beta}(x_{1}-x_{2})(\frac{2a}{N-1}|\phi(x_1)|^2\\ &\hspace{3cm}+\frac{2a}{N-1}|\phi(x_2)|^2  )p_{1}p_{2}\widehat{n}_3S_{j-1,k}\Psi\raa|\;.
\end{align*}
\begin{align*}
\leq&C \|\widehat{n}^{-1}_1p_1q_2(V_\beta f_{1/4,\beta}(x_{1}-x_{2})-\frac{2a}{N-1}|\phi(x_2)|^2 )p_{1}p_{2}\|_{\mathcal{M}}
\|\widehat{n}_3S_{j-1,k}\|_{\{1,2\}}
\\&+C \|\widehat{n}^{-1}_1q_2\Psi\|\;\|p_1g_{1/4,\beta}(x_{1}-x_{2})\|_{op} N^{-1}\|\phi\|_\infty^2\|\widehat{n}_3S_{j-1,k}\|_{\{1,2\}}
\\\leq& CN^{-(j+k)/8}\|\phi\|_\infty^{1+j+k}(\|\phi\|_\infty+1)\;.
\end{align*}

For $(\ref{bjkd})$ note, that $(1-q_{1}q_{2})=p_1p_2+p_1q_2+q_1p_2$. Since all factors in $(\ref{bjkd})$ are symmetric in exchanging $x_1$ with $x_2$ it is sufficient to control
\be\label{suffb}
\laa\Psi ,\potdiff_\beta(x_1,x_2)p_1r_2g_{1/4,\beta}(x_{1}-x_{2})p_{1}p_{2}S_{j-1,k}\Psi\raa
\ee
for $r_2\in\{p_2,q_2\}$ to get good control of $|(\ref{bjkd})|$. Using (\ref{hilfe3}) line (\ref{suffb}) is bounded  by
\begin{align*} \|p_1&\potdiff_\beta(x_1,x_2)\Psi\|\;\|p_1g_{1/4,\beta}(x_{1}-x_{2})p_{1}\|_{op}\|\widehat{n}_1^{-1}\|_{op}
\|\widehat{n}_1S_{j-1,k}\|_{\{1,2\}}
\\\leq& CN^{-2-1/2+1/2}N^{1-(j+k)/8}\|\phi\|_{\infty}^{2+j+k}\leq  CN^{-1}\|\phi\|_{\infty}^{2+j+k}\;.
\end{align*}

For $(\ref{bjke})$ we use that $p_j=p_j^2$, thus
\begin{align*}
(\ref{bjke})=&\sum_{j+k=l}j\laa\Psi ,R_{j,k}p_5p_6\potdiff_\beta(x_5-x_6) \Psi\raa
\end{align*}
With Lemma \ref{Rnorm}  it follows that
\begin{align*}
\\|(\ref{bjke})|\leq& \sum_{j+k=l}\|R_{j,k}\|_{\{5,6\}}\|p_1\potdiff_\beta(x_1,x_2)\Psi\|
\\\leq& CN^{-1/2+1-(j+k)/8} \|\phi\|_\infty^{j+k}
N^{-1}\|\phi\|_\infty
\\\leq&  C\|\phi\|_\infty^{j+k+1} N^{-1/2}\;.
\end{align*}

\item [(c)]
Using $p_5^2=p_5$ we have $R_{j,k}=R_{j,k}p_5$, thus

\begin{align} |\cjk(\Psi,\phi)|\leq& Nj|\laa\Psi , R_{j,k}p_5\potdiff_\beta(x_1,x_5)\Psi\raa |\label{cjka}
\\\label{cjkb}&+ Nj|\laa\Psi , \potdiff_\beta(x_1,x_5)R_{j,k} \Psi\raa|
\;.
\end{align}

For $(\ref{cjka})$ we use Lemma \ref{Rnorm} and (\ref{hilfe3})

\begin{align*}
|(\ref{cjka})|\leq&  Nj\|R_{j,k}\|_{\{1,5\}}\|p_5V_\beta(x_1-x_5)\Psi\|
\\\leq& C\|\phi\|_\infty^{j+k+1} N^{-(j+k)/8}\;.
\end{align*}

For $(\ref{cjkb})$ we write using $q_3=p_1q_3+q_1q_3=p_1q_3+q_1-q_1p_3$
\begin{align}
\nonumber(\ref{cjkb})=&N|\laa\Psi ,
\potdiff_\beta(x_1,x_{3})q_{3}q_{4}g_{1/4,\beta}(x_{3}-x_{4})p_{3}p_{4}S_{j-1,k} \Psi\raa|
\\\leq&\label{cjkc} N|\laa\Psi ,
\potdiff_\beta(x_1,x_{3})p_1\widehat{n}_1^{-1}q_{3}q_{4}g_{1/4,\beta}(x_{3}-x_{4})p_{3}p_{4}\widehat{n}_3S_{j-1,k} \Psi\raa|
\\\label{cjkd}&+N|\laa\Psi ,
\potdiff_\beta(x_1,x_{3})q_1q_{4}g_{1/4,\beta}(x_{3}-x_{4})p_{3}p_{4}S_{j-1,k} \Psi\raa|
\\\label{cjke}&+N|\laa\Psi ,
\potdiff_\beta(x_1,x_{3})q_1p_{3}q_{4}g_{1/4,\beta}(x_{3}-x_{4})p_{3}p_{4}S_{j-1,k} \Psi\raa|
\;.
\end{align}

For $(\ref{cjkc})$ note that $p_1\potdiff_\beta(x_1,x_{3})\Psi\in\mathcal{H}_{\{1,3\}}$, so with
Lemma \ref{kombinatorik} (b)
$$\|\widehat{n}_1^{-1}q_4p_1\potdiff_\beta(x_1-x_{3})\Psi\|\leq  C \|p_1\potdiff_\beta(x_1-x_{3})\Psi\|\leq  CN^{-1}\|\phi\|_\infty\;.$$
Using also Lemma
 \ref{Rnorm}  we get
\begin{align*}
(\ref{cjkc})
\leq& C\|\phi\|_\infty
\; \| g_{1/4,\beta}(x_{3}-x_{4})p_{3}\|_{op}\|\widehat{n}_3S_{j-1,k}\|_{\{1,3,4\}}
\\\leq& C\|\phi\|_{\infty}^{j+k+1} N^{-(j+k)/8}\;.
\end{align*}

For $(\ref{cjkd})$ we use Lemma \ref{Rnorm}
\begin{align*}
|(\ref{cjkd})|\leq& N|\laa\sqrt{|\potdiff_\beta(x_1,x_{3})|}\Psi ,\\ &
q_{4}g_{1/4,\beta}(x_{3}-x_{4})p_{4}\sqrt{|\potdiff_\beta(x_1,x_{3})|}p_{3}
S_{j-1,k}\widehat{n}_1 \widehat{n}^{-1}_1q_1\Psi\raa|
\\\leq& CN\|\sqrt{|\potdiff_\beta(x_1,x_{3})|}\Psi\|\;\|g_{1/4,\beta}(x_{3}-x_{4})p_{4}\|_{op}\\&
\|\sqrt{|\potdiff_\beta(x_1,x_{3})|}p_{3}\|_{op}\|\widehat{n}_1S_{j-1,k}\|_{\{1,3,4\}}\|\widehat{n}^{-1}_1q_1\Psi\|
%
%
\\=&CN^{-(j+k)/8}\|\phi\|_\infty^{j+k+1}\;.
\end{align*}
Again using  Lemma \ref{Rnorm}   we get for the last term  in (c)
\begin{align*}
(\ref{cjke})\leq& 
%
N\|p_3\potdiff_\beta(x_1,x_{3})\Psi\|
\|g_{1/4,\beta}(x_{3}-x_{4})p_{4}\|_{op}
\\&\|\widehat{n}_1S_{j-1,k}\|_{\{1,3,4\}}\;\|\widehat{n}^{-1}_1q_1\Psi\|
\\\leq& C\|\phi\|_\infty^{j+k+1}N^{-(j+k)/8}\;.
\end{align*}

\item [(d)] Using Lemma \ref{kombinatorik} (d) we get that
\begin{align*} (N-2j-2k)&(N-2j-2k-1)[\potdiff_\beta(x_1,x_2),
R_{j,k}]\\=&[\potdiff_\beta(x_1,x_2),
p_1p_2S_{j,k}]+[\potdiff_\beta(x_1,x_2),
p_1q_2T_{j,k}]\\&+[\potdiff_\beta(x_1,x_2),
q_1p_2T_{j,k}]\;.
\end{align*}
Hence
\begin{align*} \djk(\Psi,\phi)=&
-2\laa\Psi ,[\potdiff_\beta(x_1,x_2),
p_1q_2T_{j,k}]\;, \Psi\raa\end{align*} thus

\begin{align*} |\djk(\Psi,\phi)|\leq&
2\laa\Psi ,\potdiff_\beta(x_1,x_2)
p_1T_{j,k}\widehat{n}_1\widehat{n}_1^{-1}q_2\Psi\raa\\&+2\laa\Psi ,
q_2\widehat{n}_1^{-1}\widehat{n}_1T_{j,k}p_1\potdiff_\beta(x_1,x_2)\Psi\raa\end{align*}
  Lemma \ref{Rnorm}  gives
\begin{align*} |\djk(\Psi,\phi)|\leq&
C\|\widehat{n}_1T_{j,k}\|_{\{1,2\}}\|p_{1}\potdiff_\beta(x_1,x_2)\Psi\|\;
\|\widehat{n}_1^{-1}q_2\Psi\|
\\\leq& C\|\phi\|_\infty^{j+k+1} N^{1-(j+k)/8-1}\;.
\end{align*}

\item[(e)] Again using  $R_{j,k}=R_{j,k}p_5$ as well as $q_1=1-p_1$ \begin{align} |\ejk(\Psi,\phi)|\leq&j(j-1)|\laa\Psi , R_{j,k}p_5\potdiff_\beta(x_5,x_7) \Psi\raa|\label{ejka}
\\&\hspace{-1cm}+(j-1)|\laa\Psi , \potdiff_\beta(x_1,x_5)p_1q_2g_{1/4,\beta}(x_1-x_2)p_1p_2S_{j-1,k} \Psi\raa|\label{ejkb}
\\&\hspace{-1cm}+(j-1)\laa\Psi , \potdiff_\beta(x_1,x_5)q_2g_{1/4,\beta}(x_1-x_2)p_1p_2S_{j-1,k} \Psi\raa|\label{ejkc}
\;.
\end{align}
 Lemma \ref{Rnorm} gives
\begin{align*}
|(\ref{ejka})|\leq& C \|R_{j,k}\|_{\{5,7\}} \|p_5\potdiff_\beta(x_5,x_7) \Psi\|\\\leq&  C\|\phi\|_\infty^{j+k+1} N^{1/2-(j+k)/8-1}
\end{align*}
as well as
\begin{align*}
|(\ref{ejkb})|\leq&  C \|p_1\potdiff_\beta(x_1,x_5) \Psi\|\;\|p_1g_{1/4,\beta}(x_1-x_2)p_1\|_{op}
\\&\|\widehat{n}_1S_{j-1,k}\|_{\{1,2,5\}}\|\widehat{n}^{-1}_1\|_{op}
\\
\leq&  C\|\phi\|^{j+k+2}_\infty N^{-1-1-1/2+3/2-(j-1+k)/8+1/2}
\end{align*}
and
\begin{align*}
|(\ref{ejkc})|=&- (j-1)\laa\sqrt{\potdiff_\beta(x_1,x_5)}\Psi ,\\&q_2 g_{1/4,\beta}(x_1-x_2)p_2\sqrt{\potdiff_\beta(x_1,x_5)}p_1S_{j-1,k} \Psi\raa
\\|(\ref{ejkc})|\leq& C\|\sqrt{\potdiff_\beta(x_1,x_5)}\Psi\|\; \|g_{1/4,\beta}(x_1-x_2)p_2\|_{op}\\&
 \|\sqrt{\potdiff_\beta(x_1,x_5)}p_1\|_{op}\|\widehat{n}_1S_{j-1,k}\|_{\{1,2,5\}}\|\widehat{n}^{-1}_1\|_{op}
\\\leq&  C\|\phi\|^{j+k+1}_\infty N^{-1/2-1-1/8-1/2+3/2-(j-1+k)/8+1/2}\;.
\end{align*}

\item[(f)]
Using as above $R_{j,k}=R_{j,k}p_5=p_{2j+5}R_{j,k}$

\begin{align*}|\fjk(\Psi,\phi)|\leq& jk\left|\laa\Psi,R_{j,k}p_{5}\potdiff_\beta(x_5,x_{2j+5})\Psi\raa\right|
\\&+jk\left|\laa\Psi,\potdiff_\beta(x_5,x_{2j+5})p_{2j+5}R_{j,k}\Psi\raa\right|
\\\leq& C \|R_{j,k}\|_{\{5,2j+5\}}\|p_{5}\potdiff_\beta(x_5,x_{2j+5})\Psi\|
\\\leq& C\|\phi\|_\infty^{j+k+1} N^{1/2-(j+k)/8-1} \;.
\end{align*}

\end{enumerate}

\end{proof}

\subsection{Proof of the Theorem for $1/3\leq \beta<  1$}\label{secproof2}

Summarizing the last section we get the following Corollary, which directly gives the Theorem.
\begin{corollary}\label{corproof2}
Let  $0<\beta<1$. There exists a functional $\Gamma:\LZN\otimes\LZ\to\mathbb{R}^+$, a functional
$\Gamma':\LZN\otimes\LZ\to\mathbb{R}$ and a $c>0$ such that
\begin{enumerate}
 \item $$|\dt\Gamma(\Psi_t,\phi_t)|\leq |\Gamma'(\Psi_t,\phi_t)|\;.$$

\item $$ c\alpha(\Psi,\phi)-CN^{-\eta}\leq  \Gamma(\Psi,\phi)\leq   \alpha(\Psi,\phi)+CN^{-\eta}$$
uniform in $\Psi,\phi$
\item There exists a functional $\kinf$ such that $$|\Gamma'(\Psi,\phi)|\leq  \CphiA\C$$ uniform in $\Psi,\phi$.
\end{enumerate}

\end{corollary}

\begin{proof}
Set \begin{align*}\Gamma(\Psi,\phi):=&\sum_{ j+k\leq  5}2^{-j-k}\gamma_{j,k}(\Psi,\phi)+|\mathcal{E}(\Psi)-\mathcal{E}^{GP}(\phi)|\hspace{1cm}\text{ and}
\\ \Gamma'(\Psi,\phi):=&\sum_{ j+k\leq  5}2^{-j-k}\gamma_{j,k}'(\Psi,\phi)+\dt|\mathcal{E}(\Psi)-\mathcal{E}^{GP}(\phi)|\;.\end{align*}

\begin{enumerate}
 \item follows from Lemma \ref{firstadjlemma} with (\ref{ablediff}).
\item
\begin{align*}
\Gamma(\Psi,\phi)=&\laa\Psi,\widehat{m}^0\Psi\raa+|\mathcal{E}(\Psi)-\mathcal{E}^{GP}(\phi)|+\sum_{1\leq  j+k\leq  5}\laa\Psi,R_{j,k}\Psi\raa+\Gamma(\Psi,\phi)\end{align*}

In view of Lemma \ref{mhut} we have that
$$ c_0\alpha(\Psi,\phi)\leq \laa\Psi,\widehat{m}^0\Psi\raa+|\mathcal{E}(\Psi)-\mathcal{E}^{GP}(\phi)|\leq
\alpha(\Psi,\phi)\;.$$
The other summands are in view of Lemma \ref{Rnorm} bounded by $CN^{-\eta}$ and (b) follows.

\item Recall that $\xi_{-1,k}=0$, thus
\begin{align*}
\Gamma'(\Psi,\phi)=&\sum_{ j+k\leq  5}2^{-j-k}\left(\gamma_{j,k}'(\Psi,\phi)+2\Im(\xi_{j-1,k})-\Im(\xi_{j,k})\right)
\\&+\sum_{j+k=5}2^{-5}\Im(\xi_{j,k})
+\dt|\mathcal{E}(\Psi)-\mathcal{E}^{GP}(\phi)|
\end{align*}
The first line is controlled by Lemma \ref{firstadjest}. The second line is bounded by Proposition \ref{xiest} and (\ref{ablediff}).

\end{enumerate}
\end{proof}
From (b) and (c)   it follows that $$\Gamma'(\Psi,\phi)\leq  \CphiA\Cgamma$$ and we get via Gr\o nwall
$$\Gamma(\Psi_t,\phi_t)\leq  e^{\int_0^t (\|\phi_s\|_\infty+(\ln N)^{1/3}\|\nabla\phi_s\|_{6,loc}+\|\dot A_s\|_\infty)K(\phi_s)ds}(\Gamma(\Psi_0,\phi_0)+ N^{-\eta})\;.$$
For $\phi\in\mathcal{G}$ we have that $\sup_{s\in\mathbb{R}}\{\mathcal{K}(\phi_s)\}<\infty$.
Again using (b) we get the bound on $\alpha(\Psi_t,\phi_t)$ as stated in Theorem \ref{theorem}.

\section{Generalizing to $\beta=1$}\label{secb3}

 Lemma \ref{firstadjest} holds for $\beta=1$. When generalizing the Theorem to $\beta=1$ below one ``only'' has to adjust $\Gamma$ in Corollary \ref{corproof2} such that the $\xi$    in Corollary \ref{newalpha2}  becomes controllable.

Recall that (c.f.  Corollary \ref{newalpha2}) $$\xi=-2N(N-1)\Im\left(\laa\Psi,\potdiff_{\beta}(x_1,x_2)p_1q_2(\widehat{m}^0-\widehat{m}^0_1)\Psi\raa\right)\;.$$
The method we use is similar as above: We add a functional $\gamma$ to $\Gamma$ such that the interaction term in $\xi$ is smoothed out by the cost of additional terms. It turns out that all the additional terms are controllable so in contrast to section \ref{secadj1} this first adjustment will be sufficient.

The ``new'' interaction term will scale moderately and can --- using Lemma \ref{hnorms} (d) --- be controlled in terms of $\|q_1\nabla_1\Psi\|$. But we only got good control of $\|\nabla_1q_1\Psi\|$ for $\beta<1$. Hence we have to generalize our estimates on $\|\nabla_1q_1\Psi\|$ to $\beta=1$ first.

\subsection{Controlling $\|\nabla_1q_1\Psi\|$ for $\beta=1$}\label{secabl2}

For $\beta=1$ a relevant part of the kinetic energy is used to form the microscopic structure.
That part of the kinetic energy is concentrated around the scattering centers. Hence $\|\nabla_1q_1\Psi\|$ will in fact {\it not} be small.

The microscopic structure is --- neglecting three
particle interactions --- given by Lemma \ref{defAlemma}.
So we shall first cutoff three particle interactions, i.e. we define a cutoff function which
does not depend on $x_1$ and cuts off all parts of the wave function
where two particles $x_j, x_k$ with $j\neq k$, $j,k\neq 1$ come to
close (see $\mathcal{B}_1$ in  Definition \ref{hdetail}).

After that we shall cutoff that part of the kinetic energy which
is used to form the microscopic structure. The latter is
concentrated around the scattering centers (i.e. on the sets
$\overline{\mathcal{A}}_j$ given by Definition \ref{hdetail}).

Then we show  that $\|\mathds{1}_{\mathcal{A}_1}\nabla_1q_1\Psi\|$ is small (see Lemma \ref{kineticenergy} below).

Having good control on $\|\mathds{1}_{\mathcal{A}_1}\nabla_1q_1\Psi\|$ instead of $\|\nabla_1q_1\Psi\|$
Lemma \ref{hnorms} (d) has of course to be changed appropriately. This will be done in Lemma \ref{qqqterm}.

\begin{definition}\label{hdetail}
For any $j,k=\mathbb{N}$ let \be\label{defkleins}a_{j,k}:=\{(x_1,x_2,\ldots,x_N)\in
\mathbb{R}^{3N}: |x_j-x_k|<N^{-26/27}\}\ee

$$\overline{\mathcal{A}}_j:=\bigcup_{k\neq j}a_{j,k}\;\;\;\;\;\;\;\mathcal{A}_j:=\mathbb{R}^{3N}\backslash \overline{\mathcal{A}}_j
\;\;\;\;\;\;\;\overline{\mathcal{B}}_{j}:=\bigcup_{k,l\neq j}a_{k,l}\;\;\;\;\;\;\;\mathcal{B}_{j}:=\mathbb{R}^{3N}\backslash
\overline{\mathcal{B}}_{j}\;.$$

\end{definition}

\begin{proposition}\label{propo}
\begin{enumerate}
 \item

$$\|\mathds{1}_{\overline{\mathcal{A}}_{j}}p_1\|_{op}\leq  C\|\phi\|_\infty N^{-17/18}\;,$$
\item

$$\|\mathds{1}_{\overline{\mathcal{A}}_{j}}\Psi\|\leq  CN^{-17/27}\|\nabla_j\Psi\|\;,$$
\item
$$\|\mathds{1}_{\overline{\mathcal{B}}_{j}}\Psi\|\leq  CN^{-7/54}\|\nabla_j\Psi\|\;.$$
\end{enumerate}\end{proposition}

\begin{proof}
\begin{enumerate}
 \item

 \begin{align*}
 \|\mathds{1}_{\overline{\mathcal{A}}_{j}}p_1\|_{op}\leq& \|\phi\|_\infty\|\mathds{1}_{\overline{\mathcal{A}}_{j}}\|
 \\=&\|\phi\|_\infty  |\overline{\mathcal{A}}_{j}|^{1/2} \leq  \|\phi\|_\infty N^{(1-26/9)/2}
 \end{align*}

 \item
Using H\"older and Sobolev under the $x_k$-integration we get 
\begin{align*} \|
\mathds{1}_{\overline{\mathcal{A}}_{j}}\Psi\|^2=&\|\mathds{1}_{\overline{\mathcal{A}}_{j}}\Psi^2\|_1\leq  \|\mathds{1}_{\overline{\mathcal{A}}_{j}}\|_{3/2}\|\Psi^2\|_3
\leq |\overline{\mathcal{A}}_{j}|^{2/3}\|\nabla_j\Psi\|^2
\\\leq& \|\nabla_1\Psi \|^2(N N^{-26/9})^{2/3}\leq
N^{-34/27}\|\nabla_1 \Psi\|^2\;.
 \end{align*}
Since $\|\nabla_1\Psi \|<C$ (b) follows.

 \item We use that $\overline{\mathcal{B}}_{j}\subset\bigcup_{k=1}\overline{\mathcal{A}}_{k}$.
Hence one can find pairwise disjoint sets $ \mathcal{C}_k\subset\mathcal{A}_{k}$, $k=1,\ldots,N$
such that $\overline{\mathcal{B}}_{j}\subset\bigcup_{k=1} \mathcal{C}_{k}$. Since the sets $ \mathcal{C}_k$ are pairwise disjoint,
the $\mathds{1}_{ \mathcal{C}_{k}}\Psi $ are pairwise orthogonal and we get
\begin{align*}
\|\mathds{1}_{\overline{\mathcal{B}}_{j}}\Psi \|^2=\sum_{k=1}\|\mathds{1}_{ \mathcal{C}_{k}}\Psi \|^2
\leq \sum_{k=1}\|\mathds{1}_{\overline{\mathcal{A}}_{k}}\Psi \|^2\leq  CN^{-7/27}\|\nabla_j\Psi\|^2\;.
\end{align*}
\end{enumerate}
\end{proof}

Next we prepare estimates of some energy terms we shall need below.

\begin{corollary}\label{potdiffer2}
Let $V_1\in\mathcal{V}_1$, $0<\beta_1<1$. Then there exist a $\kinf$ and a $\eta>0$ such that for all
$\Psi\in\mathcal{H}_{\emptyset}$ and all $\phi\in\LZ$
$$\laa\Psi ,\left(2a|\phi(x_1)|^2-(N-1)\mathds{1}_{\mathcal{B}_{1}}W_{\beta_1}(x_1-x_2)\right)\Psi \raa\leq
\C$$

\end{corollary}

\begin{proof}
\begin{align}\nonumber
\laa\Psi ,&\left(2a|\phi(x_1)|^2-(N-1)\mathds{1}_{\mathcal{B}_{1}}W_{\beta_1}(x_1-x_2)\right)\Psi \raa
\\\label{split1}=&\laa\mathds{1}_{\mathcal{B}_{1}}\Psi ,p_1p_2\left(2a|\phi(x_1)|^2-(N-1)W_{\beta_1}(x_1-x_2)\right)\mathds{1}_{\mathcal{B}_{1}}p_1p_2\Psi \raa
\\\label{split2}&+2a\laa\Psi , |\phi(x_1)|^2- \mathds{1}_{\mathcal{B}_{1}}p_1p_2|\phi(x_1)|^2p_1p_2\mathds{1}_{\mathcal{B}_{1}}\Psi \raa
\\\label{split3}&+2(N-1)\laa\Psi ,\mathds{1}_{\mathcal{B}_{1}}p_1p_2W_{\beta_1}(x_1-x_2)(1-p_1p_2)\mathds{1}_{\mathcal{B}_{1}}\Psi \raa
\\\label{split4}&-(N-1)\laa\Psi ,\mathds{1}_{\mathcal{B}_{1}}(1-p_1p_2)W_{\beta_1}(x_1-x_2)(1-p_1p_2)\mathds{1}_{\mathcal{B}_{1}}\Psi \raa\;.
\end{align}
Using (\ref{faltungorigin}) and (\ref{formelfuerkin}) we get that (\ref{split1}) is bounded by
$$|(\ref{split1})|\leq \C\;.$$
Due to Lemma \ref{kombinatorikb}
$$|\laa\Psi , |\phi(x_1)|^2\Psi\raa-\langle\phi,|\phi^2|\phi\rangle|\leq  C\|\phi\|_\infty^2 \alpha(\phi,\phi)\;.$$
On the other hand we have
\begin{align*}
\laa\Psi,\mathds{1}_{\mathcal{B}_{1}}p_1p_2|\phi(x_1)|^2p_1p_2\mathds{1}_{\mathcal{B}_{1}}\Psi\raa
=\langle\phi,|\phi^2|\phi\rangle\|p_1p_2\mathds{1}_{\mathcal{B}_{1}}\Psi\|^2
\end{align*}
and with Proposition \ref{propo} that
\begin{align*}
|\|p_1p_2\mathds{1}_{\mathcal{B}_{1}}\Psi\|^2-1|\leq& \left| \|p_1p_2\mathds{1}_{\mathcal{B}_{1}}\Psi\|^2-\|p_1p_2\Psi\|^2\right|
+\left|\|p_1p_2\Psi\|^2-1\right|
\\\leq& CN^{-7/54}+\alpha(\Psi,\phi)\;.
\end{align*}
Thus
\begin{align*}|(\ref{split2})|
\leq \C\;.
\end{align*}

For (\ref{split3}) we use that the support of $W_{\beta_1}(x_1-x_2)$ and $\overline{\mathcal{A}}_1$ are disjoint, thus
\begin{align} \label{split5}(\ref{split3})=&-
2(N-1)\laa\Psi ,\mathds{1}_{\mathcal{B}_{1}}p_1p_2W_{\beta_1}(x_1-x_2)p_1p_2\mathds{1}_{\mathcal{B}_{1}}\mathds{1}_{\overline{\mathcal{A}}_{1}}\Psi \raa
\\\label{split6}&+
2(N-1)\laa\Psi ,\mathds{1}_{\mathcal{B}_{1}}\mathds{1}_{\overline{\mathcal{A}}_{1}}p_1p_2W_{\beta_1}(x_1-x_2)(1-p_1p_2)\mathds{1}_{\mathcal{B}_{1}}\mathds{1}_{\mathcal{A}_{1}}\Psi \raa
\\\label{split7}&+2(N-1)\laa\Psi ,\mathds{1}_{\mathcal{B}_{1}}\mathds{1}_{\mathcal{A}_{1}}p_1p_2W_{\beta_1}(x_1-x_2)(1-p_1p_2)\mathds{1}_{\mathcal{B}_{1}}\mathds{1}_{\mathcal{A}_{1}}\Psi \raa
\;.
\end{align}
Using Proposition \ref{propo} and Lemma \ref{kombinatorik} (e) we have
$$|(\ref{split5})|\leq 2N\|p_1W_{\beta_1}(x_1-x_2)p_1\|_{op}\|\mathds{1}_{\overline{\mathcal{A}}_{1}}\Psi \|\leq
 C\|\phi\|_\infty^2N^{-17/27}\;.$$
For (\ref{split6}) we have using  Proposition \ref{propo} and Lemma \ref{kombinatorik} (e)
\begin{align*}
|\ref{split6}|\leq& 2N \|\mathds{1}_{\overline{\mathcal{A}}_{1}}\Psi\|\;\|\mathds{1}_{\overline{\mathcal{A}}_{1}}p_1\|_{op} \|p_1W_{\beta_1}(x_1-x_2)\|_{op}
\\\leq&  CN N^{-17/24} \|\phi\|_\infty N^{-17/18} \|\phi\|_\infty \|W_{\beta_1}\|
\\\leq& CN^{-11/72+3/2(\beta-1)} \|\phi\|_\infty^2\;.
\end{align*}
Note that $\mathds{1}_{\mathcal{B}_{1}}\mathds{1}_{\mathcal{A}_{1}}\mathbb{R}^{3N}\backslash\bigcup_{j\neq k}a_{j,k}$, thus $\mathds{1}_{\mathcal{B}_{1}}\mathds{1}_{\mathcal{A}_{1}}\Psi\in\mathcal{H}_{\emptyset}$. Therefore
(\ref{split7}) is controlled by Lemma \ref{potdiffer} (a) and also bounded by the right hand side of the Corollary.

Having good control of (\ref{split1}), (\ref{split2}) and (\ref{split3}) and using that (\ref{split4}) is negative the Lemma follows.

\end{proof}

\begin{lemma}\label{kineticenergy}
Let  $V_1\in\mathcal{V}_1$. Then there exists a $\eta>0$ and a $\kinf$ such that for any $\Psi \in \mathcal{H}_\emptyset$ and any
$\phi\in\LZ$
\begin{enumerate}
\item
\begin{align*}\|\mathds{1}_{\mathcal{A}_{1}}\nabla_1q_1
\Psi \|^2 \leq  \C
 \end{align*}
\item
\begin{align*} (N-1)\|\mathds{1}_{\overline{\mathcal{B}}_{1}}\sqrt{V_1}(x_1-x_2)\Psi \|^2\leq \C
 \end{align*}

\item
\begin{align*}\|\mathds{1}_{\overline{\mathcal{B}}_{1}}\nabla_1q_1\Psi \|^2\leq& \C
\end{align*}
\end{enumerate}
\end{lemma}
\begin{proof}
\begin{itemize}
 \item [(a)+(b)]
For any $0<\beta_1<1$ we have
\begin{align*}
\|\nabla_1&q_1\Psi \|^2+\laa\Psi ,((N-1)V_1(x_1-x_2)-2a|\phi (x_1)|^2)\Psi \raa
\nonumber\\=&\|\mathds{1}_{\mathcal{A}_{1}}\nabla_1q_1\Psi
\|^2
 + \|\mathds{1}_{\overline{\mathcal{B}}_{1}}\mathds{1}_{\overline{
\mathcal{A}}_{1}}\nabla_1q_1\Psi \|^2
+ \|\mathds{1}_{\mathcal{B}_{1}}\mathds{1}_{\overline{\mathcal{A}}_{1}}\nabla_1q_1\Psi \|^2
\\&+(N-1)\|\mathds{1}_{\overline{\mathcal{B}}_{1}}\sqrt{V}_1(x_1-x_2)\Psi \|^2\\&+\laa\Psi ,\sum_{j\neq
1}\mathds{1}_{\mathcal{B}_{1}}\left(V_1-W_{\beta_1}\right)(x_1-x_j)\Psi \raa
\\& +\laa\Psi ,\left(\sum_{j\neq
1}\mathds{1}_{\mathcal{B}_{1}}W_{\beta_1}(x_1-x_j)-2a|\phi(x_1)|^2\right)\Psi \raa\;.
\end{align*}
Using that $q_1=1-p_1$ and symmetry gives (after reordering)
\begin{align} =&(N-1)\|\mathds{1}_{\overline{\mathcal{B}}_{1}}\sqrt{V}_1(x_1-x_2)\Psi \|^2
+ \|\mathds{1}_{\mathcal{A}_{1}}\nabla_1q_1\Psi \|^2
\label{vielezeilen}\\\label{vielezeilen0}&+ \|\mathds{1}_{\overline{\mathcal{B}}_{1}}\mathds{1}_{\overline{
\mathcal{A}}_{1}}\nabla_1q_1\Psi \|^2+ \|\mathds{1}_{ \mathcal{B}_{1}}\mathds{1}_{\overline{
\mathcal{A}}_{1}}\nabla_1p_1\Psi \|^2
\\\label{vielezeilena}&-2 \Re\left(\laa\nabla_1q_1\Psi ,
\mathds{1}_{\mathcal{B}_{1}}\mathds{1}_{\overline{\mathcal{A}}_{1}}
\nabla_1p_1\Psi \raa\right)
\\\label{vielezeilenb}&+ \|\mathds{1}_{\mathcal{B}_{1}}\mathds{1}_{\overline{\mathcal{A}}_{1}}\nabla_1\Psi \|^2
+\laa\Psi ,\sum_{j\neq
1}\mathds{1}_{\mathcal{B}_{1}}\left(V_1-W_{\beta_1}\right)(x_1-x_j)\Psi \raa
\\&\label{vielezeilenc}+\laa\Psi ,\left(\sum_{j\neq
1}\mathds{1}_{\mathcal{B}_{1}}W_{\beta_1}(x_1-x_j)-2a|\phi(x_1)|^2\right)\Psi \raa
\end{align}

Proposition \ref{propo} (b) and (\ref{ablabsch}) yields that for some $\kinf$
\begin{align*}
|(\ref{vielezeilena})|\leq& 2\left|\Re\left(\laa \nabla_1 q_1\Psi ,
\mathds{1}_{\overline{\mathcal{A}}_{1}}\nabla_1p_1\mathds{1}_{\mathcal{B}_{1}}\Psi \raa\right)\right|
\\\leq& \|\nabla_1q_1\Psi \|\;\|\nabla\phi \|\|\mathds{1}_{\mathcal{B}_{1}}\Psi \|
\\\leq& \mathcal{K}(\phi) N^{-7/54}\;.
\end{align*}
Choosing $\beta_1<1$ large enough the support of the
potentials $V_{1}(x_1-x_j)$ and $W_{\beta_1}(x_1-x_j)$ are
subsets of
$\overline{\mathcal{A}}_{1}:=\mathbb{R}^{3N}\backslash\mathcal{A}_{1}$ (c.f. Definition \ref{hdetail}).
Furthermore we have
that the support of the potentials
$$\mathds{1}_{\mathcal{B}_{1}}\left(V_{1}(x_1-x_j)-W_{\beta_1}(x_1-x_j)\right)$$
are pairwise disjoint for different $j$. It follows with Lemma
\ref{defAlemma} (c) that
 (\ref{vielezeilenb}) is positive.

Corollary \ref{potdiffer2} gives a abound on (\ref{vielezeilenc}) it follows that
\begin{align*}(\ref{vielezeilen})+(\ref{vielezeilen0})
\leq& \C\;.
\end{align*}

Since all the summands in (\ref{vielezeilen}) and (\ref{vielezeilen0}) are positive, it follows that each of them is
bounded by $\C$
and we get (a), (b) as well as
\be\label{gleich}\|\mathds{1}_{\overline{\mathcal{B}}_{1}}\mathds{1}_{\overline{\mathcal{A}}_{1}}\nabla_1q_1\Psi
\|^2\leq \C\;.\ee

\item[(c)] (\ref{gleich}) and (a) give
\begin{align*}
 \|\mathds{1}_{\overline{\mathcal{B}}_{1}}\nabla_1q_1\Psi \| ^2
\leq&  \|\mathds{1}_{\overline{\mathcal{B}}_{1}}\mathds{1}_{\overline{\mathcal{A}}_{1}}\nabla_1q_1\Psi \|^2
+ \|\mathds{1}_{\overline{\mathcal{B}}_{1}}\mathds{1}_{\mathcal{A}_{1}}\nabla_1q_1\Psi \|^2
\\\leq&  \|\mathds{1}_{\overline{\mathcal{B}}_{1}}\mathds{1}_{\overline{\mathcal{A}}_{1}}\nabla_1q_1\Psi \|^2
+ \|\mathds{1}_{\mathcal{A}_{1}}\nabla_1q_1\Psi \|^2
\\\leq& \C\;.
\end{align*}

\end{itemize}
\end{proof}

\subsection{Generalizing Lemma \ref{hnorms} (d)}\label{secqqq}

For $\beta=1$ our plan is again to smoothen out the interaction using the microscopic structure and use Lemma \ref{hnorms}
for the smoothed interaction terms. To be able to do so we first have to estimate Lemma \ref{hnorms} (d)  in terms of
$\|\mathds{1}_{\mathcal{A}_{1}}\nabla_1q_1
\Psi\|$ (Remember that for $\beta=1$  $\|\nabla_1q_1
\Psi\|$ is not small).

\begin{lemma}\label{qqqterm}
 Let for $0<\beta<1$
$V_\beta\in \mathcal{V}_\beta$, $g\in L^2$ and $m:\mathbb{N}^2\to\mathbb{R}^+$ with $m\leq  n^{-1}$.
Then there exists a $\kinf$ and a $\eta>0$ such that for any $\Psi \in\mathcal{M}$ with
$\{1,2\}\subset\mathcal{M}$
\begin{align}\label{l24}&N\laa \Psi q_1p_2V_\beta(x_1-x_2)\widehat{m}q_1q_2\Psi\raa
\\\nonumber&\hspace{0.2cm}\leq
\mathcal{K}(\phi)(\|\phi\|_\infty+(\ln N)^{1/3}\|\nabla\phi\|_{6,loc})
\left(\laa\Psi,\widehat{n}\Psi\raa+\|\mathds{1}_{\mathcal{A}_{1}}\nabla_1q_1\Psi\|^2+N^{-\eta}\right)\end{align}


\end{lemma}

\begin{proof}

We first prove the Lemma for some small $0<\beta$ and generalize to all $0<\beta<1$ thereafter.

In view of Definition \ref{udef}
\begin{align}\nonumber N|\laa\Psi ,q_1  p_2 V_{\beta}(x_1-&x_2) \widehat{m}q_1 q_2
\Psi\raa|=N|\laa\Psi ,q_1  p_2 U_{0,\beta} (x_1-x_2)\widehat{m}q_1 q_2 \Psi\raa|
\\& +N|\laa\Psi ,\widehat{m}_1q_1  p_2 (\Delta
h_{0,\beta})(x_1-x_2) q_1 q_2 \Psi\raa|\nonumber
\\ \leq&  N|\laa\Psi ,q_1  p_2 U_{0,\beta} (x_1-x_2)\widehat{m}q_1 q_2 \Psi\raa|\label{saa}
\\& +N|\laa\Psi ,q_1  p_2\widehat{m}_1 (\nabla_1
h_{0,\beta}(x_1-x_2))q_2\mathds{1}_{\overline{\mathcal{A}}_1}\nabla_1q_1
\Psi\raa|\label{sab}
\\& +N|\laa\mathds{1}_{\overline{\mathcal{A}}_1}\nabla_1q_1  \Psi ,p_2
(\nabla_1 h_{0,\beta}(x_1-x_2)) q_1 q_2\widehat{m} \Psi\raa|\label{sac}
\\& +
N|\laa\Psi ,q_1  p_2 \widehat{m}_1 (\nabla_1
h_{0,\beta}(x_1-x_2))q_2\mathds{1}_{\mathcal{A}_1}\nabla_1q_1
\Psi\raa|\label{sad}
\\& +N|\laa\mathds{1}_{\mathcal{A}_1}\nabla_1q_1  \Psi ,p_2
(\nabla_1 h_{0,\beta}(x_1-x_2)) q_1 q_2\widehat{m}\Psi\raa|\label{sae}
\;. \end{align}
For $(\ref{saa})$ we have
\begin{align*}
(\ref{saa})\leq  N\|q_1\Psi\|\;\| U_{0,\beta}\|_\infty\| \widehat{m}q_1 q_2 \Psi\|
\leq  CN\alpha(\Psi,\phi)\|V_\beta\|_1\;.
\end{align*}
For $(\ref{sab})$ and $(\ref{sac})$.
\begin{align*}
(\ref{sab})+(\ref{sac})\leq& N\|\mathds{1}_{\overline{\mathcal{A}}_1}q_2(\nabla_1
h_{0,\beta}(x_1-x_2))\widehat{m}_1 q_1  p_2\Psi\|\;\|\nabla_1q_1
\Psi\|
\\&+N\|\nabla_1q_1  \Psi\|\;\| \mathds{1}_{\overline{\mathcal{A}}_1}p_2
(\nabla_1 h_{0,\beta}(x_1-x_2)) q_1 q_2\widehat{m} \Psi\|
\end{align*}
Using Proposition \ref{propo} and Lemma \ref{ulemma}
\begin{align*}
(\ref{sab})+(\ref{sac})\leq& CN^{10/27} \|\nabla_1 q_2(\nabla_1
h_{0,\beta}(x_1-x_2))\widehat{m}_1 q_1  p_2\Psi\|\;\|\nabla_1q_1
\Psi\|
\\&+CN^{10/27}\|\nabla_1q_1 \Psi\|\; \|\nabla_1p_2
(\nabla_1 h_{0,\beta}(x_1-x_2)) q_1 q_2\widehat{m} \Psi\|
\\\leq& CN^{10/27} \| q_2(\Delta_1
h_{0,\beta}(x_1-x_2))\widehat{m}_1 q_1  p_2\Psi\|
\\&+CN^{10/27} \| q_2(\nabla_1
h_{0,\beta}(x_1-x_2))\nabla_1\widehat{m}_1 q_1  p_2\Psi\|
\\&+CN^{10/27}\|p_2
(\Delta_1 h_{0,\beta}(x_1-x_2)) q_1 q_2\widehat{m} \Psi\|
\\&+CN^{10/27}\|p_2
(\nabla_1 h_{0,\beta}(x_1-x_2)) \nabla_1q_1 q_2\widehat{m} \Psi\|
\end{align*}

Note that for any $m\leq  Cn^{-1}$ and any $\Psi\in\mathcal{H}_{\{1\}}$
\begin{align}\label{nablacomm}
\|\nabla_1q_1q_2\widehat{m}\Psi\|\leq& \|p_1\nabla_1q_1q_2\widehat{m}\Psi\|+\|q_1\nabla_1q_1q_2\widehat{m}\Psi\|
\nonumber\\=&\|\widehat{m}_{1}q_2p_1\nabla_1q_1\Psi\|+\|\widehat{m}q_2q_1\nabla_1q_1\Psi\|
\nonumber\\\leq& C\|\nabla_1q_1\Psi\|
\end{align}
and similarly
\begin{align*}
\|\nabla_1q_1p_2\widehat{m}\Psi\|
\leq& C\|\widehat{m}\|_{op}\|\nabla_1q_1\Psi\|\;.
\end{align*}
Since $\|\nabla_1q_1\Psi\|$ is bounded (see (\ref{ablabsch}))
\begin{align*}
(\ref{sab})+(\ref{sac})\leq& CN^{10/27}\|\phi\|_\infty\| \big(\| \Delta_1
h_{0,\beta}\|
\\& +\| \nabla_1
h_{0,\beta} \|N^{-1/2}
 + \|
\Delta_1 h_{0,\beta}\|
+
\|\nabla_1 h_{0,\beta} \|\big)
\end{align*}
Using the bounds on $\|
\Delta_1 h_{0,\beta}\|$ and $\| \nabla_1
h_{0,\beta} \|$ from Lemma \ref{ulemma} it follows that for $\beta$ small enough we can find a $\eta>0$ and a $\kinf$ such that
$$(\ref{sab})+(\ref{sac})\leq  \mathcal{K}(\phi)\|\phi\|_\infty N^{-\eta}\;.$$

\begin{align*}
(\ref{sad})\leq& \frac{N}{N-1} |\laa\Psi ,\sum_{k=2}^N q_1  p_k \widehat{m}_1 (\nabla_1
h_{0,\beta}(x_1-x_k))q_k\mathds{1}_{\mathcal{A}_1}\nabla_1q_1
\Psi\raa|
\\\leq&  \frac{N}{N-1}\|\sum_{k=2}^N    q_k(\nabla_1
h_{0,\beta}(x_1-x_k))\widehat{m}_1q_1  p_k\Psi\|\;\|\mathds{1}_{\mathcal{A}_1}\nabla_1q_1
\Psi\|
\end{align*}
Due to (\ref{ablabsch}) $\|\mathds{1}_{\mathcal{A}_1}\nabla_1q_1
\Psi\|$ is bounded.
For the other factor we write
\begin{align}
\nonumber\|&\sum_{k=2}^N    q_k(\nabla_1
h_{0,\beta}(x_1-x_k))\widehat{m}_1q_1  p_k\Psi\|^2
\\=&\sum_{2\leq  j<k\leq  N}\laa \Psi,  \widehat{m}_1q_1  p_j(\nabla_1
h_{0,\beta}(x_1-x_j))q_jq_k(\nabla_1
h_{0,\beta}(x_1-x_k))\widehat{m}_1q_1  p_k\Psi\raa\label{saf}
\\&+\sum_{k=2}^N    \laa \Psi,  \widehat{m}_1q_1  p_k(\nabla_1
h_{0,\beta}(x_1-x_k))q_k(\nabla_1
h_{0,\beta}(x_1-x_k))\widehat{m}_1q_1  p_k\Psi\raa\label{sag}
 \;.
\end{align}
For $(\ref{saf})$ we use that for any $2\leq  k\leq  N$ $\nabla_1
h_{0,\beta}(x_1-x_k)=\nabla_k
h_{0,\beta}(x_1-x_k)$. Partial integrations  yield
\begin{align*}
(\ref{saf})&\\\leq& \sum_{2\leq  j<k\leq  N}|\laa \nabla_k\nabla_jp_j\widehat{m}_1q_k\Psi,  q_1
h_{0,\beta}(x_1-x_j)
h_{0,\beta}(x_1-x_k)\widehat{m}_1q_1  p_kq_j\Psi\raa|
\\&+\sum_{2\leq  j<k\leq  N}|\laa  \nabla_jp_j\widehat{m}_1q_k\Psi,  q_1
h_{0,\beta}(x_1-x_j)
h_{0,\beta}(x_1-x_k)\nabla_k\widehat{m}_1q_1  p_kq_j\Psi\raa|
\\&+\sum_{2\leq  j<k\leq  N}|\laa \nabla_k\widehat{m}_1q_k\Psi,  q_1  p_j
h_{0,\beta}(x_1-x_j)
h_{0,\beta}(x_1-x_k)\nabla_j\widehat{m}_1q_1  p_kq_j\Psi\raa|
\\&+\sum_{2\leq  j<k\leq  N}|\laa \widehat{m}_1q_k\Psi,  q_1  p_j
h_{0,\beta}(x_1-x_j)
h_{0,\beta}(x_1-x_k)\nabla_j\nabla_k\widehat{m}_1q_1  p_kq_j\Psi\raa|\:.
\end{align*}
Due to symmetry the first and the fourth line are equal and
\begin{align*}
(\ref{saf})\leq& 2N^2|\laa \nabla_2p_2\nabla_3\widehat{m}_1q_3\Psi,  q_1
h_{0,\beta}(x_1-x_3)p_3h_{0,\beta}(x_1-x_2)\widehat{m}_1q_1  q_2\Psi\raa|
\\&+N^2|\laa  \nabla_2\widehat{m}_1p_2q_3\Psi,  q_1
h_{0,\beta}(x_1-x_2)h_{0,\beta}(x_1-x_3)\nabla_3\widehat{m}_1q_1p_3  q_2\Psi\raa|
\\&+N^2|\laa \nabla_3\widehat{m}_1q_3\Psi,  q_1
h_{0,\beta}(x_1-x_3) p_3 p_2h_{0,\beta}(x_1-x_2)\nabla_2\widehat{m}_1q_1 q_2\Psi\raa|
\\\leq& 2N^2\| p_3 h_{0,\beta}(x_1-x_3)\nabla_3\widehat{m}_1q_1q_3\Psi\|
\|h_{0,\beta}(x_1-x_2)\nabla_2p_2 \|_{op}\|\widehat{m}_1q_1  q_2\Psi\|
\\&+N^2\|h_{0,\beta}(x_1-x_2)\nabla_2p_2
\|_{op}^2\|\widehat{m}_1q_1q_3\Psi\|^2
\\&+N^2\|p_2 h_{0,\beta}(x_1-x_2) \nabla_2\widehat{m}_1q_1q_2\Psi\|^2
\\\leq& CN\|p_2 h_{0,\beta}(x_1-x_2)
\nabla_2\widehat{m}_1q_1q_2\Psi\|\;\|\nabla\phi\|_{6,loc}(\ln N)^{1/3}\sqrt{\alpha(\Psi,\phi)}
\\&+C\|\nabla\phi\|_{6,loc}^2(\ln N)^{2/3}
\alpha(\Psi,\phi)
\\&+N^2\|p_2 h_{0,\beta}(x_1-x_2) \nabla_2\widehat{m}_1q_1q_2\Psi\|^2
\;.
\end{align*}
Let us next control the factor $\|p_2 h_{0,\beta}(x_1-x_2) \nabla_2\widehat{m}_1q_1q_2\Psi\|$. Using
$1=\mathds{1}_{\mathcal{A}_1}+\mathds{1}_{\overline{\mathcal{A}}_1}$ and exchanging the variables $x_1$ and $x_2$ we have
\begin{align*}
\|p_2 h_{0,\beta}(x_1-x_2) \nabla_2\widehat{m}_1q_1q_2\Psi\|
\leq& \|p_1 h_{0,\beta}(x_1-x_2)\|_{op} \|\nabla_1q_1q_2\widehat{m}_1\mathds{1}_{\mathcal{A}_1}\Psi\|
\\&+\|p_1 h_{0,\beta}(x_1-x_2) \nabla_1\|_{op}\|q_2\widehat{m}_1\mathds{1}_{\overline{\mathcal{A}}_1}q_1\Psi\|\;.
\end{align*}
Using formula (\ref{nablacomm}) on the first and Lemma \ref{kombinatorikb} on the second summand (note, that $\mathds{1}_{\mathcal{A}_1}\Psi\in\mathcal{H}_{\{1\}}$)  we get
\begin{align*}
\|p_2 h_{0,\beta}(x_1-x_2) \nabla_2\widehat{m}_1q_1q_2\Psi\|
\leq& C\|p_1 h_{0,\beta}(x_1-x_2)\|_{op} \|\nabla_1q_1\mathds{1}_{\mathcal{A}_1}q_1\Psi\|
\\&+\|p_1 h_{0,\beta}(x_1-x_2) \nabla_1\|_{op}\|\mathds{1}_{\overline{\mathcal{A}}_1}\Psi\|\;.
\end{align*}
The first summand is in view of Lemma \ref{kombinatorik} (e) and Lemma \ref{ulemma} bounded by
$$CN^{-1}\|\phi\|_{\infty} \|\mathds{1}_{\mathcal{A}_1}\nabla_1q_1\Psi\|\:.$$

Partial integration, Proposition (\ref{propo}) and again  Lemma \ref{kombinatorik} (e) with Lemma \ref{ulemma} give for the second summand
\begin{align*}
\|&p_1 h_{0,\beta}(x_1-x_2) \nabla_1\|_{op}\|\mathds{1}_{\overline{\mathcal{A}}_1}q_1\Psi\|
\\\leq& CN^{-17/27}\left(\| h_{0,\beta}(x_1-x_2) \nabla_1p_1\|_{op}+\| (\nabla_1h_{0,\beta}(x_1-x_2)) p_1\|_{op}\right)\|\nabla_1q_1\Psi\|
\\\leq& CN^{-17/27}\left(N^{-1}\| \nabla\phi\|_{6,loc}(\ln N)^{1/3}+N^{-1}\|\phi\|_{\infty}\right)\|\nabla_1q_1\Psi\|
\end{align*}
 Using (\ref{ablabsch}) and choosing $\beta$ sufficiently small there exists a $\eta>0$ and a $\kinf$ such that
\begin{align}\label{superform}
&N\|p_2 h_{0,\beta}(x_1-x_2)
\nabla_2\widehat{m}_1q_1q_2\Psi\|
\\\nonumber&\hspace{1cm}\leq \mathcal{K}(\phi)(\|\phi\|_\infty+(\ln N)^{1/3}\|\nabla\phi\|_{6,loc})
\left( \|\mathds{1}_{\mathcal{A}_{1}}\nabla_1q_1\Psi\|^2+N^{-\eta}\right)\;.
\end{align}
Thus
$(\ref{saf})$ is bounded by the right hand side of (\ref{l24}). For $(\ref{sag})$ we have
\begin{align*}
(\ref{sag})\leq& N \| \widehat{m}_1q_1\Psi\|^2\;\|   p_2(\nabla_1
h_{0,\beta}(x_1-x_2))\|_{op}^2
\\\leq& CN \|\phi\|_\infty^2 \|\nabla_1
h_{0,\beta}\|^2 \leq  C N^{-1/2+2\beta }\|\phi\|_\infty^2
\end{align*}
For small enough $\beta$
  that there exists a $\eta>0$ such that $(\ref{sag})\leq  CN^{-\eta}\|\phi\|_\infty^2$ and thus
$ (\ref{sad})$ is bounded by the right hand side of (\ref{l24}) for some $\kinf$.

Again using that $\nabla_1 h_{0,\beta}(x_1-x_2)=-\nabla_2 h_{0,\beta}(x_1-x_2)$ and partial integration yields for
$(\ref{sae})$
\begin{align*}
(\ref{sae})\leq& N|\laa h_{0,\beta}(x_1-x_2)\nabla_2p_2\mathds{1}_{\mathcal{A}_1}\nabla_1q_1  \Psi ,
  q_1 q_2\widehat{m}\Psi\raa|
\\&+N|\laa\mathds{1}_{\mathcal{A}_1}\nabla_1q_1  \Psi ,p_2
h_{0,\beta}(x_1-x_2) q_1 \nabla_2q_2\widehat{m}\Psi\raa|
\\\leq& N\|h_{0,\beta}(x_1-x_2)\nabla_2p_2\|_{op}\|\mathds{1}_{\mathcal{A}_1}\nabla_1q_1  \Psi\|
 \|q_1 q_2\widehat{m}\Psi\|
\\&+N\|\mathds{1}_{\mathcal{A}_1}\nabla_1q_1  \Psi\|\;\|p_2
h_{0,\beta}(x_1-x_2) q_1 \nabla_2q_2\widehat{m}\Psi\|\;.
\end{align*}
Using (\ref{superform}) it follows that $N(\ref{sae})\leq  C$ which completes the proof of the Lemma for  $0<\beta\leq \beta_0$ for some $0<\beta_0<1$.

To prove the Lemma for $\beta_0<\beta<1$ we write
\begin{align*}|\laa\Psi ,q_1  p_2 V_{\beta}(x_1-x_2) \widehat{m}q_1 q_2
\Psi\raa|=&|\laa\Psi ,q_1  p_2 U_{\beta_0,\beta} \widehat{m}q_1 q_2 \Psi\raa|
\\&\hspace{-2cm}+N\laa\Psi p_1q_2(V_\beta(x_1-x_2)-U_{\beta_1,\beta}(x_1-x_2))\widehat{m}q_1q_2\Psi\raa\;.\end{align*}
In view of Lemma \ref{ulemma} $U_{\beta_0,\beta}\in\mathcal{V}_{\beta_0}$. The case  $\beta=\beta_0$ has just been shown, so
 the first summand is bounded by the right hand side of (\ref{l24}). The second summand is controlled by (\ref{hnormd})  and the Lemma
follows in full generality.

\end{proof}

\subsection{Second adjustment: Making $\alpha'_2$ controllable}\label{secadj2}

Next we adjust the functional $\Gamma$ from Corollary \ref{corproof2}  such that $\xi$  (defined in Corollary \ref{newalpha2}) becomes controllable.

Recall that (see (\ref{xiformel}))
\begin{align}\label{xixi}\xi=&-2N(N-1)\Im\left(\laa\Psi,\potdiff_{\beta}(x_1,x_2)p_1q_2(\widehat{m}^0-\widehat{m}
^0_1)\Psi\raa\right)
%
\;.
\end{align}

Following the ideas in section \ref{secadj1} we define now a functional which smoothens the ``bad'' interaction term in $\xi$.
\begin{definition}\label{lambda2}
Let  $V_{1}\in\mathcal{V}_1$, let $\widehat{m}=\widehat{m}^0-\widehat{m}^0_1$.

We define the functional $\gamma:\LZN\otimes\LZ\to\mathbb{R}^+$ by \begin{align*}\gamma(\Psi,\phi) :=&-N(N-1)\Im\left(\llaa\Psi
, q_{1}q_{2}g_{8/9,1}(x_{1}-x_{2}) \widehat{m}p_{1}q_{2}
\Psi\rraa\right) \end{align*}
and the functional $\gamma':\LZN\otimes\LZ\to\mathbb{R} $ by
$$\gamma'(\Psi,\phi):=\as+\bs+\cs+\ds+\xi\;,$$
where the different summands are
\begin{enumerate}
 \item [(a)] The mixed derivative term
\begin{align*} \as=&-N(N-1)\Im\left(\laa\Psi ,q_{1}q_{2}\left[H,
g_{8/9,1}(x_{1}-x_{2})
\right]\widehat{m}p_{1}q_{2}
\Psi\raa\right)\\&-N(N-1)\Im\left(\laa\Psi ,
q_{1}q_{2}(W_{8/9}-V_{\beta})f_{8/9,1}(x_1-x_2)\widehat{m}p_{1}q_{2}
\Psi\raa\right)\;.\end{align*}

\item [(b)] The new interaction term
\begin{align*}\bs=&-\xi+N(N-1)\Im\left(\laa\Psi ,\left[\potdiff_\beta(x_1,x_2),q_{1}q_{2} g_{8/9,1}(x_{1}-x_{2})
\widehat{m}p_{1}q_{2}\right] \Psi\raa\right)
\\&+N(N-1)\Im\left(\laa\Psi ,q_{1}q_{2}
(W_{8/9}-V_{\beta})f_{8/9,1}(x_1-x_2)\widehat{m}p_{1}q_{2}
\Psi\raa\right)
\end{align*}
\item [(c)] Three particle interactions
\begin{align*}&\cs=N(N-1)(N-2)\\&\;\;\;\;\;\;\;\;\;\;\Im\left(\laa\Psi ,
\left[\potdiff_\beta(x_1,x_3)+\potdiff_\beta(x_2,x_3),q_{1}q_{2}g_{8/9,1}(x_{1}-x_{2})
\widehat{m}p_{1}q_{2}\right] \Psi\raa\right)\end{align*}

\item [(d)] Interaction terms of the correction
\begin{align*}&\ds=N(N-1)(N-2)(N-3)\\&\;\;\;\;\;\;\;\;\;\;\Im\left(\laa\Psi ,q_{1}q_{2}
g_{8/9,1}(x_{1}-x_{2})p_{1}q_{2}\left[\potdiff_\beta(x_3,x_4),
\widehat{m}\right] \Psi\raa\right)\end{align*}

\end{enumerate}

\end{definition}

Lemma \ref{hnorms} (a) together with Lemma \ref{defAlemma} (a) imply directly
\begin{corollary}\label{replacealpa2}
There exists a $\eta>0$ such that
$$|\gamma(\Psi,\phi)|\leq  C N^{-\eta}\|\phi\|_\infty\;.$$
\end{corollary} Next we show that in fact  $\gamma'$ satisfies
$\gamma'(\Psi_t,\phi_t)=\dt \gamma(\Psi_t,\phi_t)$:
\begin{lemma}\label{secondadjlemma}
$$\dt\gamma(\Psi_t,\phi_t)=\gamma'(\Psi_t,\phi_t) $$

\end{lemma}

\begin{proof}

\begin{align*}\dt\gamma(\Psi_t,\phi_t)=&-N(N-1)\Im\left(\llaa\Psi ,q_{1}q_{2}\left[H,
g_{8/9,1}(x_{1}-x_{2})
\right]\widehat{m}p_{1}q_{2}
\Psi\rraa\right)
\\&+N(N-1)\Im\left(\llaa\Psi , \left[H-H^{GP},q_{1}q_{2}g_{8/9,1}(x_{1}-x_{2})
\widehat{m}p_{1}q_{2}\right] \Psi\rraa\right)
\end{align*}
Using symmetry
\begin{align*}
=&-N(N-1)\Im\left(\llaa\Psi ,q_{1}q_{2}\left[H,
g_{8/9,1}(x_{1}-x_{2})
\right]\widehat{m}p_{1}q_{2}
\Psi\rraa\right)
\\&+N(N-1)\Im\left(\laa\Psi ,\left[\potdiff_\beta(x_1,x_2),q_{1}q_{2} g_{8/9,1}(x_{1}-x_{2})
\widehat{m}p_{1}q_{2}\right] \Psi\raa\right)
\\&+N(N-1)(N-2)\Im\left(\laa\Psi ,
\left[\potdiff_\beta(x_1,x_3)+\potdiff_\beta(x_2,x_3),q_{1}q_{2}g_{8/9,1}(x_{1}-x_{2})
\widehat{m}p_{1}q_{2}\right] \Psi\raa\right)
\\&+N(N-1)(N-2)(N-3)\Im\left(\laa\Psi ,q_{1}q_{2}
g_{8/9,1}(x_{1}-x_{2})p_{1}q_{2}\left[\potdiff_\beta(x_3,x_4),
\widehat{m}\right] \Psi\raa\right)
\;.
\end{align*}

\end{proof}

Subtracting
$$N(N-1)\Im\left(\laa\Psi ,
q_{1}q_{2}(W_{8/9}-V_{\beta})f_{8/9,1}(x_1-x_2)\widehat{m}p_{1}q_{2}
\Psi\raa\right)$$
from the first line
 and adding it to the second line, as well as subtracting $\xi$ from the second line
and adding it to the total gives that the right hand side of (\ref{toprove}) equals
$\as+\bs+\cs+\ds+\xi$
which proves the Lemma.

\begin{lemma}\label{secondadjcontrol}
Let $\beta=1$. There exists a $\kinf$ such that
$$ \gamma'(\Psi,\phi)-\xi\leq \Cphi\C\;.$$

\end{lemma}

\begin{proof}
\begin{enumerate}
 \item
 Using (\ref{commu}) and $\nabla_1g_{8/9,1}(x_1-x_2)=-\nabla_2g_{8/9,1}(x_1-x_2)$
and integrating by parts we have
\begin{align}\nonumber
|\as(\Psi,\phi)|\leq& N^2|\laa\Psi
,q_{1}q_{2}\widehat{m}_{-1}(\nabla_2g_{8/9,1}(x_1-x_2))\nabla_2
p_1q_2\Psi\raa|
\\\nonumber&+N^2|\laa\Psi
,q_{1}q_{2}\widehat{m}_{-1}(\nabla_1g_{8/9,1}(x_1-x_2))\nabla_1
p_1q_2\Psi\raa|
\\\leq&\label{asa} 2N^2|\laa\Psi
,q_{1}q_{2}\widehat{m}_{-1}g_{8/9,1}(x_1-x_2)\nabla_1\nabla_2
p_1q_2\Psi\raa|
\\&\label{asb}+N^2|\laa\nabla_1\widehat{m}_{-1}q_{1}q_{2}\Psi
,g_{8/9,1}(x_1-x_2)\nabla_2
p_1q_2\Psi\raa|
\\&\label{asc}+N^2|\laa\nabla_2q_{1}q_{2}\Psi
,g_{8/9,1}(x_1-x_2)\nabla_1
\widehat{m}p_1q_2\Psi\raa|
\;.
\end{align}

We use that  for any $\chi\in\LZN$ by H\"older- and Sobolev's inequality
\begin{align*}
\|\mathds{1}_{\{(x_1-x_2)\leq
RN^{-8/9}\}}\chi\|^2=&\laa\chi,\mathds{1}_{\{(x_1-x_2)\leq
RN^{-8/9}\}}\chi\raa\\\leq&  \|\mathds{1}_{\{(x_1-x_2)\leq
RN^{-8/9}\}}\|_{3/2}\|\chi^2\|_3
\\\leq&  C N^{-16/9}\|\nabla_1\chi\|^2\;.
\end{align*}
This, (\ref{nablamitm}), Lemma \ref{mhut} and Lemma \ref{kombinatorik} (e) give
\begin{align*}
(\ref{asa})\leq&
N^2\|\mathds{1}_{\{(x_1-x_2)\leq
RN^{-8/9}\}}\widehat{m}_{-1}q_1q_2\Psi\|\\&\|g_{8/9,1}(x_1-x_2)\nabla_1p_1\|_{op}\|\nabla_2
q_2\Psi\|
\\\leq& CN^{2-8/9}\|q_2\nabla_1\widehat{m}_{-1}q_1\Psi\|\;\|g_{8/9,1}\|_3\|\nabla\phi\|_{6,loc}
\\\leq& CN^{-8/9}(\ln N)^{1/3}\|\nabla\phi\|_{6,loc}\;.
\end{align*}
For $(\ref{asb})$ we get
\begin{align*}
(\ref{asb})\leq& N^2\|\widehat{m}_{-1}\nabla_1q_{1}q_{2}\Psi\|\;
\|g_{8/9,1}(x_1-x_2)p_1\|_{op}\|\nabla_2
q_2\Psi\|
\\\leq& CN^{-4/9}\|\phi\|_{\infty}\;.
\end{align*}
To control $(\ref{asc})$ we use symmetry and write
\begin{align*}
(\ref{asc})=&\frac{N^2}{N-1}|\laa\nabla_2q_{2}\Psi
,\sum_{j\neq 2}q_{j}g_{8/9,1}(x_j-x_2)\nabla_j
\widehat{m}p_jq_2\Psi\raa|
\\\leq& \frac{N^2}{N-1}\|\nabla_2q_{2}\Psi\|
\|\sum_{j=2}^Nq_{j}g_{8/9,1}(x_j-x_1)\nabla_j
\widehat{m}p_jq_1\Psi\|
\end{align*}
As above $\|\nabla_2q_{2}\Psi\|\leq  C$. For the second factor we write
\begin{align}\nonumber
\|&\sum_{j=2}^Nq_{j}g_{8/9,1}(x_j-x_1)\nabla_j
\widehat{m}p_jq_1\Psi\|^2\\=&
\sum_{j=2}^N\laa\nabla_jp_j\widehat{m}q_1\Psi,g_{8/9,1}(x_j-x_1)q_{j}g_{8/9,1}(x_j-x_1)\nabla_j
p_j\widehat{m}q_1\Psi\raa\label{asd}
\\&+\sum_{j\neq k=2}^N\laa\nabla_jp_j\widehat{m}q_1\Psi,g_{8/9,1}(x_j-x_1)q_{j}q_kg_{8/9,1}(x_k-x_1)\nabla_k
p_k\widehat{m}q_1\Psi\raa\label{ase}
\;.
\end{align}
For $(\ref{asd})$ we use symmetry, Lemma \ref{kombinatorik} (e) and Lemma \ref{kombinatorikb} (recall that in view of Lemma \ref{mhut} $m(k,N)\leq  CN^{-1}n^{-1}(k+1,N)$) and get
\begin{align*}
(\ref{asd})\leq& (N-1) \laa\nabla_2p_2\widehat{m}q_1\Psi,g_{8/9,1}(x_j-x_1) g_{8/9,1}(x_2-x_1)\nabla_2
p_2\widehat{m}q_1\Psi\raa
\\\leq& N\| g_{8/9,1}(x_2-x_1)\nabla_2
p_2\|_{op}^2\|\widehat{m}q_1\Psi\|^2
\\\leq& N^{-3}(\ln N)^{2/3}\|\nabla\phi\|_{6,loc}^2\;.
\end{align*}
For $(\ref{ase})$ we get
\begin{align*}
|(\ref{ase})|\leq& N^2 \laa\nabla_2p_2\mathds{1}_{\{(x_1-x_3)\leq
RN^{-8/9}\}}\widehat{m}q_1q_3\Psi,\\&g_{8/9,1}(x_2-x_1)g_{8/9,1}(x_3-x_1)\nabla_3
p_3\mathds{1}_{\{(x_1-x_2)\leq
RN^{-8/9}\}}q_1q_{2}\Psi\raa
\\\leq& N^2\|g_{8/9,1}(x_2-x_1)\nabla_2p_2\|_{op}^2\|\mathds{1}_{\{(x_1-x_2)\leq
RN^{-8/9}\}}\widehat{m}q_1q_{2}\Psi\|^2\;.
\end{align*}
Since $\{(x_1-x_2)\leq
RN^{-8/9}\}\subset \mathcal{A_1}$ we get with Proposition \ref{propo} and (\ref{nablacomm})
\begin{align*}
\|\mathds{1}_{\{(x_1-x_2)\leq
RN^{-8/9}\}}\widehat{m}q_1q_{2}\Psi\|\leq&
CN^{-17/27}\|\nabla_1\widehat{m}q_1q_{2}\Psi\|\\\leq&  CN^{-1-17/27}\|\nabla_1q_1\Psi\|
\end{align*}
It follows that
\begin{align*}|(\ref{ase})| \leq  CN^{-2-34/27}\|\nabla\phi\|^2_{6,loc}(\ln N)^{2/3}\|\nabla_1q_1\Psi\|^2\;,
\end{align*}
thus
$$\|\sum_{j=2}^Nq_{j}g_{8/9,1}(x_j-x_1)\nabla_j
\widehat{m}p_jq_1\Psi\|\leq  C N^{-3/2}\|\nabla\phi\|_{6,loc}(\ln N)^{1/3}(1+\|\nabla_1q_1\Psi\|)\:.$$
With (\ref{ablabsch}) it follows that also $|(\ref{asc})|$ has the right bound and (a) follows.

\item
For $\bs$ we can write in view of (\ref{xixi}) and using $q_1|\phi(x_2)|^2p_1=0$
\begin{align}&\bs(\Psi,\phi)\leq \label{bsa}N^2|\laa\Psi ,
q_{1}q_{2}\widehat{m}_{-1}g_{8/9,1}(x_{1}-x_{2})
p_{1}q_{2} \potdiff_1(x_1,x_2)\Psi\raa|
\\\label{bsb}&+N^2|\laa\Psi ,  \potdiff_1(x_1,x_2)
(q_1q_2-1)g_{8/9,1}(x_{1}-x_{2})
p_{1}q_{2}\widehat{m}\Psi\raa|
\\\label{bsc}&+N^2|\laa\Psi ,q_{1}q_{2}
(W_{8/9}(x_1-x_2)-V_1(x_1-x_2)+\potdiff_1(x_1,x_2))\\&\hspace{6.0cm}f_{8/9,1}(x_{1}-x_{2})\widehat{m}
p_1q_2
\Psi\raa|\nonumber
\;.
\end{align}
For $(\ref{bsa})$ we have
\begin{align*}
|(\ref{bsa})|\leq& N^2\|\widehat{m}_{-1}q_2\Psi\|\;\|g_{8/9,1}(x_{1}-x_{2})
p_{1}\|_{op}\|p_1\potdiff_1(x_1,x_2)\Psi\|\\\leq&  CN^{-4/9} \|\phi\|_\infty^2\;.
\end{align*}
For $(\ref{bsb})$
\begin{align*}
(\ref{bsb})\leq& N^2|\laa\Psi ,  \potdiff_1(x_1,x_2)
p_1p_2g_{8/9,1}(x_{1}-x_{2})\widehat{m}
p_{1}q_{2}\Psi\raa|
\\&+N^2|\laa\Psi ,  \potdiff_1(x_1,x_2)
p_1q_2g_{8/9,1}(x_{1}-x_{2})\widehat{m}
p_{1}q_{2}\Psi\raa|
\\&+N^2|\laa\Psi ,
q_1q_2g_{8/9,1}(x_{1}-x_{2})\widehat{m}
p_{1}q_{2}\Psi\raa|
\\\leq& N^2\|p_1\potdiff_1(x_1,x_2)\Psi\|\;\|p_1g_{8/9,1}(x_{1}-x_{2})
p_{1}\|_{op}\|\widehat{m}q_2\Psi\|
\\&+N^2 \|p_1\potdiff_1(x_1,x_2)\Psi\|\;\|p_1g_{8/9,1}(x_{1}-x_{2})
p_{1}\|_{op}\|\widehat{m}q_2\Psi\|
\\&+N^2\|p_1\potdiff_1(x_1,x_2)\Psi\|\;\|g_{8/9,1}(x_{1}-x_{2})
p_{1}\|_{op}\;.\|\widehat{m}q_2\Psi\|
\end{align*}
Recall that $m(k,N)\leq  C N^{-1}n(k+1,N)^{-1}$, thus
\begin{align*}
(\ref{bsb}) \leq  C N^{2-2-16/9}\|\phi\|^3_\infty+C N^{2-2-4/9}\|\phi\|^2_\infty\;.
\end{align*}
 For $(\ref{bsc})$ we have
\begin{align*} &(\ref{bsc})\leq
N |\laa\Psi ,q_{1}q_{2}
W_{8/9}(x_1-x_2)f_{8/9,1}(x_{1}-x_{2})N\widehat{m}
p_1q_2
\Psi\raa|
\\&+\frac{N^2}{N-1}|\laa\Psi ,
 q_1q_2\widehat{m}_{-1} \left(2a|\phi(x_1)|^2+2a|\phi(x_2)|^2\right)f_{8/9,1}(x_{1}-x_{2})
p_{1}q_{2}\Psi\raa|
\;.
\end{align*}
We get with Lemma \ref{qqqterm} that the first line is bounded by
$$\mathcal{K}(\phi)(\|\phi\|_\infty+(\ln N)^{1/3}\|\nabla\phi\|_{6,loc})
\left(\laa\Psi,\widehat{n}\Psi\raa+\|\mathds{1}_{\mathcal{A}_{1}}\nabla_1q_1\Psi\|^2+N^{-\eta}\right)\;.$$

For the second line recall that $\|f_{8/9,1}\|_\infty=1$, thus it is controlled by
$$C N\|q_1q_2\widehat{m}_{-1}\Psi\|\;\|\phi\|_\infty^2\|q_{2}\Psi\|\leq
C\|\phi\|_\infty^2\laa\Psi,\widehat{n}\Psi\raa\;.$$

\item

Using $q_2=1-p_2$ the left hand side of (c) is bounded by
\begin{align}
|\cs(\Psi,\phi)|\leq&
\label{csa}N^3|\Im\laa\Psi ,\left[\potdiff_1(x_2,x_3),
q_{1}q_2\widehat{m}_{-1}g_{{8/9},1}(x_{1}-x_{2}) p_1p_2\right] \Psi\raa|
\\\label{csb} &+N^3|\Im\laa\Psi ,\left[\potdiff_1(x_2,x_3),
q_{1}q_2\widehat{m}_{-1} g_{{8/9},1}(x_{1}-x_{2}) p_1\right] \Psi\raa|
\\\label{newcsc} &+N^3|\Im\laa\Psi ,\left[\potdiff_1(x_1,x_3)
q_{1}q_2g_{{8/9},1}(x_{1}-x_{2})  \widehat{m}p_1q_2\right] \Psi\raa|
\;.
\end{align}
Using symmetry $(\ref{csa})$  can be controlled like in the proof of Lemma \ref{firstadjlemma} (c). For easier reference we repeat the formulas here: Using $1=p_3+q_3$ and $q_1=1-p_1$
\begin{align*}|(\ref{csa})|
\leq& N^3|\laa\Psi ,
q_{1}q_2\widehat{m}_{-1}g_{8/9,1}(x_{1}-x_{2})p_1p_2  \potdiff_1(x_2,x_3) \Psi\raa|
\\&+N^3|\laa\Psi ,\potdiff_1(x_2,x_3)p_3\
q_{1}q_2\widehat{m}_{-1}g_{8/9,1}(x_{1}-x_{2})p_1p_2  \Psi\raa|
\\&+N^3|\laa\Psi ,\potdiff_1(x_2,x_3)
p_{1}q_2g_{8/9,1}(x_{1}-x_{2})\widehat{m}p_1p_2q_3\widehat{m}  \Psi\raa|
\\&+N^3|\laa\Psi ,\sqrt{\potdiff_1(x_2,x_3)}
q_2g_{8/9,1}(x_{1}-x_{2})p_2\sqrt{\potdiff_1(x_2,x_3)}p_1q_3\widehat{m}  \Psi\raa|
\\\leq& N^3\|\widehat{m}_{-1}q_1\Psi\|\;\|g_{8/9,1}(x_{1}-x_{2})p_1\|_{op}\|p_1\potdiff_1(x_2,x_3) \Psi\|
\\&+N^3 \|q_2\widehat{m}_{-1}p_3\potdiff_1(x_2,x_3) \Psi\|\;\|g_{8/9,1}(x_{1}-x_{2})p_1\|_{op}
\\&+N^3\|p_1\potdiff_1(x_2,x_3) \Psi\|\;\|g_{8/9,1}(x_{1}-x_{2})p_1\|_{op}\|q_3\widehat{m}  \Psi\|
\\&+N^3\|\sqrt{\potdiff_1(x_2,x_3)}\Psi\|\;\|g_{8/9,1}(x_{1}-x_{2})p_2\|_{op}\|\sqrt{\potdiff_1(x_2,x_3)}p_1\|_{op}\|q_3\widehat{m}  \Psi\|
\\\leq& CN^3 \|\phi\|^2_\infty N^{-1-1-4/9-1}=C N^{-4/9} \|\phi\|^2_\infty\;.
\end{align*}
Using $q_2=1-p_2$ we can write for $(\ref{csb})$
\begin{align}
(\ref{csb})\leq& N^3|\Im\laa\Psi ,  \left[\potdiff_1(x_2,x_3),
q_1\widehat{m}_{-1}g_{8/9,1}(x_{1}-x_{2}) p_1\right]\Psi\raa|\label{csd}
\\&+N^3|\Im\laa\Psi ,  q_1p_2\widehat{m}_{-1}g_{8/9,1}(x_{1}-x_{2}) p_1\potdiff_1(x_2,x_3)\Psi\raa|\label{cse}
\\&+N^3|\Im\laa\Psi ,\potdiff_1(x_2,x_3)
q_1p_2\widehat{m}_{-1}g_{8/9,1}(x_{1}-x_{2}) p_1\Psi\raa|\label{csf}
\;.
\end{align}
Using that $q_1g_{8/9,1}(x_{1}-x_{2}) p_1$ commutes with $\potdiff_1(x_2,x_3)$ and then
Lemma \ref{kombinatorik} (d) we have
\begin{align*}
(\ref{csd})\leq& N^3|\Im\laa\Psi ,  \left[\potdiff_1(x_2,x_3),\widehat{m}_{-1}\right]
q_1g_{8/9,1}(x_{1}-x_{2}) p_1\Psi\raa|
\\\leq& N^3|\Im\laa\Psi ,  \left[\potdiff_1(x_2,x_3),p_2p_3(\widehat{m}_{-1}-\widehat{m}_{1})\right]
q_1g_{8/9,1}(x_{1}-x_{2}) p_1\Psi\raa|
\\&+
N^3|\Im\laa\Psi ,  \left[\potdiff_1(x_2,x_3),p_2q_3(\widehat{m}_{-1}-\widehat{m}_{1})\right]
q_1g_{8/9,1}(x_{1}-x_{2}) p_1\Psi\raa|
\\&+
N^3|\Im\laa\Psi ,  \left[\potdiff_1(x_2,x_3),q_2p_3(\widehat{m}_{-1}-\widehat{m}_{1})\right]
q_1g_{8/9,1}(x_{1}-x_{2}) p_1\Psi\raa|
\\\leq& N^3|\Im\laa\Psi ,  \potdiff_1(x_2,x_3)p_2p_3(\widehat{m}_{-1}-\widehat{m}_{1})
q_1g_{8/9,1}(x_{1}-x_{2}) p_1\Psi\raa|
\\&+N^3|\Im\laa\Psi ,  q_1(\widehat{m}_{-1}-\widehat{m}_{1})p_2p_3
\potdiff_1(x_2,x_3)g_{8/9,1}(x_{1}-x_{2}) p_1\Psi\raa|
\\&+
N^3|\Im\laa\Psi ,  \potdiff_1(x_2,x_3)p_2q_3(\widehat{m}_{-1}-\widehat{m}_{1})
q_1g_{8/9,1}(x_{1}-x_{2}) p_1\Psi\raa|
\\&+
N^3|\Im\laa\Psi , q_1(\widehat{m}_{-1}-\widehat{m}_{1}) p_2q_3\potdiff_1(x_2,x_3)
g_{8/9,1}(x_{1}-x_{2}) p_1\Psi\raa|
\\&+
N^3|\Im\laa\Psi , \potdiff_1(x_2,x_3)q_2p_3(\widehat{m}_{-1}-\widehat{m}_{1})
q_1g_{8/9,1}(x_{1}-x_{2}) p_1\Psi\raa|
\\&+
N^3|\Im\laa\Psi ,q_1(\widehat{m}_{-1}-\widehat{m}_{1}) q_2p_3\potdiff_1(x_2,x_3)
g_{8/9,1}(x_{1}-x_{2}) p_1\Psi\raa|
\end{align*}
Using Lemma \ref{kombinatorikb}
\begin{align*}
\leq& CN^2\|p_2\potdiff_1(x_2,x_3)\Psi\|\;\|g_{8/9,1}(x_{1}-x_{2}) p_1\|_{op}
\\&+CN^2 \|p_2\sqrt{
\potdiff_1(x_2,x_3)}\|_{op}\|g_{8/9,1}(x_{1}-x_{2}) p_1\|_{op}\|\sqrt{
\potdiff_1(x_2,x_3)}\Psi\|
\\&+
CN^2\|p_2\potdiff_1(x_2,x_3)\Psi\|\;  \|g_{8/9,1}(x_{1}-x_{2}) p_1\|_{op}
\\&+
CN^2\|p_2\sqrt{\potdiff_1(x_2,x_3)}\|_{op}
\|g_{8/9,1}(x_{1}-x_{2}) p_1\|_{op}\|\sqrt{
\potdiff_1(x_2,x_3)}\Psi\|
\\&+
CN^2\|p_3\potdiff_1(x_2,x_3)\Psi\|\;\|g_{8/9,1}(x_{1}-x_{2}) p_1\|_{op}
\\&+
CN^2\|p_3\sqrt{\potdiff_1(x_2,x_3)}\|_{op}\|
g_{8/9,1}(x_{1}-x_{2}) p_1\|_{op} \|\sqrt{
\potdiff_1(x_2,x_3)}\Psi\|
\\\leq& CN^{2-2-4/9}\|\phi\|_\infty^2=CN^{-4/9}\|\phi\|_\infty^2\;.
\end{align*}
Using Lemma \ref{kombinatorikb} $(\ref{cse})$ is bounded by
\begin{align*}
(\ref{cse})=&
N^3|\Im\laa\Psi ,  q_1\widehat{m}_{-1}p_2\sqrt{\potdiff_1(x_2,x_3)}g_{8/9,1}(x_{1}-x_{2}) p_1\sqrt{\potdiff_1(x_2,x_3)}\Psi\raa|
\\\leq& CN^2 \|p_2\sqrt{\potdiff_1(x_2,x_3)}\|_{op}\|g_{8/9,1}(x_{1}-x_{2}) p_1\|_{op}\|\sqrt{\potdiff_1(x_2,x_3)}\Psi\|
\\\leq& CN^{-4/9}\|\phi\|_\infty^2\;.
\end{align*}
For $(\ref{csf})$ we have again with Lemma \ref{kombinatorikb}
\begin{align*}(\ref{csf})\leq  CN^2\|p_2\potdiff_1(x_2,x_3)\Psi\|
\|g_{8/9,1}(x_{1}-x_{2}) p_1\|_{op}\leq& CN^{-4/9}\|\phi\|_\infty^2\;.
\end{align*}
Having controlled $(\ref{csa})$ and all terms in $(\ref{csb})$ we split up $(\ref{newcsc})$ using $1=p_3+q_3$ and $q_1=1-p_1$
\begin{align}
(\ref{newcsc})=&
N^3|\laa\Psi ,
q_{1}q_{2} \widehat{m}_{-1}g_{8/9,1}(x_{1}-x_{2})p_1q_2 \potdiff_1(x_1,x_3) \Psi\raa|\label{csg}
\\&+\frac{N^3}{N-1}|\laa\Psi ,a(|\phi(x_1)|^2+|\phi(x_3)|^2)
q_1q_{2} g_{8/9,1}(x_{1}-x_{2}) \widehat{m}p_1q_2\Psi\raa|\label{csh}
\\&+N^3|\laa\Psi ,V_1(x_1,x_3)
q_1q_{2} g_{8/9,1}(x_{1}-x_{2}) \widehat{m}p_1q_2p_3\Psi\raa|\label{csi}
\\&+N^3|\laa\Psi ,V_1(x_1,x_3)
p_1q_{2} g_{8/9,1}(x_{1}-x_{2}) \widehat{m}p_1q_2q_3\Psi\raa|\label{csj}
\\&+N^3|\laa\Psi ,V_1(x_1,x_3)
p_{2} g_{8/9,1}(x_{1}-x_{2}) \widehat{m}p_1q_2q_3\Psi\raa|\label{csk}
\\&+N^3|\laa\Psi ,V_1(x_1,x_3)
 q_2 g_{8/9,1}(x_{1}-x_{2}) \widehat{m}p_1q_2q_3 \mathds{1}_{\overline{\mathcal{B}_2}}\Psi\raa|\label{csl}
\\&+N^3|\laa\Psi ,V_1(x_1,x_3)
 q_2g_{8/9,1}(x_{1}-x_{2}) \widehat{m}p_1q_2q_3 \mathds{1}_{\mathcal{B}_2}\Psi\raa|\label{csm}
 \;.
\end{align}
Using Lemma \ref{kombinatorikb}
$(\ref{csg})$, $(\ref{csh})$, $(\ref{csi})$, $(\ref{csj})$ and $(\ref{csk})$ are bounded by
\begin{align*}
&CN^3\| \widehat{m}_{-1}q_{2}\Psi\|;\|g_{8/9,1}(x_{1}-x_{2})p_1\|_{op}\|p_1 \potdiff_1(x_1,x_3) \Psi\|\;,
\\&CN^2\|\phi\|_\infty^2\|g_{8/9,1}(x_{1}-x_{2})p_1\|_{op} \|\widehat{m}q_2\Psi\|\;,
\\&
N^3|\|p_3V_1(x_1,x_3)\Psi\|\;\| g_{8/9,1}(x_{1}-x_{2})p_1\|_{op} \|\widehat{m}q_2\Psi\|\;,
\\&
N^3\|p_1V_1(x_1,x_3)\Psi\|\;
 \|g_{8/9,1}(x_{1}-x_{2})p_1\|_{op}\| \widehat{m}q_2\Psi\|\;,
\\&N^3\|\sqrt{V_1}(x_1,x_3)\Psi\|
p_{2} g_{8/9,1}(x_{1}-x_{2})\|_{op}\|\sqrt{V_1}(x_1,x_3)p_1\|_{op} \|\widehat{m}q_2\Psi\|
\:.
\end{align*}
All these are smaller than $C\|\phi\|^2_{\infty}N^{-4/9}$.

Next we turn to $(\ref{csl})$. Since the support of $g_{8/9,1}(x_{1}-x_{2})$ is a subset of $\overline{\mathcal{B}}_3$ we get
that $(\ref{csl})$ is bounded by \begin{align*}
N^3\|\mathds{1}_{\overline{\mathcal{B}}_3}\sqrt{V}_1(x_1,x_3)\Psi\|
 \|g_{8/9,1}(x_{1}-x_{2}) \sqrt{V}_1(x_1,x_3)\widehat{m}p_1q_2q_3\mathds{1}_{\overline{\mathcal{B}}_2} \Psi\|
\end{align*}
The first factor  is controlled by Lemma \ref{kineticenergy} (b)
\begin{align*}
(\ref{csl})\leq&  N^{5/2}\left(\C\right)^{1/2}
\\& \|g_{8/9,1}(x_{1}-x_{2}) \sqrt{V}_1(x_1,x_3)\widehat{m}p_1q_2q_3\mathds{1}_{\overline{\mathcal{B}}_2} \Psi\|
\;.
\end{align*}
 For the remaining factor we use for any fixed $x_1,x_2,\ldots,x_N$ H\"older and Sobolev under the $x_2$-integral. Setting $\chi:=\sqrt{V}_1(x_1,x_3)p_1q_2q_3\widehat{m}\mathds{1}_{\overline{\mathcal{B}}_2} \Psi$
\begin{align*}\|g_{8/9,1}(x_{1}-x_{2})\chi\|^2\leq& \|g_{8/9,1}^2\|_{3/2}\left\|\|\chi^2\|_{3\text{ in }x_2}\right\|
\\=&\|g_{8/9,1}\|_3^2\|\left\|\|\chi\|^2_{6\text{ in }x_2}\right\|\leq \|g_{8/9,1}\|_3^2\|\nabla_2\chi\|^2\;.
\end{align*}
With Lemma \ref{defAlemma} and (\ref{nablacomm}) we get
\begin{align*}
\|g_{8/9,1}(x_{1}-x_{2}) & \sqrt{V}_1(x_1,x_3)\widehat{m}p_1q_2q_3\mathds{1}_{\overline{\mathcal{B}}_2} \Psi\|
\\\leq& \|g_{8/9,1}\|_3\|\nabla_2\sqrt{V}_1(x_1,x_3)p_1q_2q_3\widehat{m}\mathds{1}_{\overline{\mathcal{B}}_2} \Psi\|
\\\leq& CN^{-1}(\ln N)^{1/3} \|\sqrt{V}_1(x_1,x_3)p_1\|_{op}\|\nabla_2q_2q_3\widehat{m}\mathds{1}_{\overline{\mathcal{B}}_2} \Psi\|
\\\leq& CN^{-5/2}(\ln N)^{1/3} \|\nabla_2q_2\mathds{1}_{\overline{\mathcal{B}}_2} \Psi\|
\end{align*}
Since \begin{align*}\|\nabla_2q_2\mathds{1}_{\overline{\mathcal{B}}_2} \Psi\|
\leq& 2\|\nabla_2\mathds{1}_{\overline{\mathcal{B}}_2} \Psi\|
+ 2\|\nabla_2p_2\mathds{1}_{\overline{\mathcal{B}}_2} \Psi\|
\\\leq& 2\|\nabla_2\mathds{1}_{\overline{\mathcal{B}}_2} \Psi\|
+ 2\|\nabla\phi\|\;\|\mathds{1}_{\overline{\mathcal{B}}_2} \Psi\|
\end{align*}
we get with  Lemma \ref{kineticenergy} (c) and Proposition \ref{propo} that the latter is bounded by $\left(\C\right)^{1/2}$.

Using symmetry $(\ref{csm})$ is bounded by
\begin{align}
CN^2&|\laa\Psi ,V_1(x_1,x_3)\sum_{j=4}^N
 g_{8/9,1}(x_{1}-x_{j}) \widehat{m}p_1q_3q_j \mathds{1}_{\mathcal{B}_j}\Psi\raa|\nonumber
\\=&CN^2|\laa\Psi ,\sqrt{V_1}(x_1,x_3)(p_1p_3+p_1q_3+q_1p_3+q_1q_3)\\&\sqrt{V_1}(x_1,x_3)\sum_{j=4}^N
 g_{8/9,1}(x_{1}-x_{j}) \widehat{m}p_1q_3q_j \mathds{1}_{\mathcal{B}_j}\Psi\raa|\nonumber
\\=&CN^2|\laa\Psi ,\sqrt{V_1}(x_1,x_3)(p_1p_3\widehat{m}_1+p_1q_3\widehat{m}+q_1p_3\widehat{m}+q_1q_3\widehat{m}_{-1})\nonumber
\\&\hspace{1cm}\sum_{j=4}^N
 g_{8/9,1}(x_{1}-x_{j})\sqrt{V_1}(x_1,x_3) p_1q_3q_j \mathds{1}_{\mathcal{B}_j}\Psi\raa|\nonumber
\\\leq& CN^{3}\|\sqrt{V_1}(x_1,x_3)\Psi\|\:\|\widehat{m}_1\|_{op}
\| p_1\sqrt{g_{8/9,1}}(x_{1}-x_{2})\|_{op}\label{csn}
\\&\hspace{1cm}\|p_3\sqrt{V_1}(x_1,x_3)\|_{op}\|\sqrt{g_{8/9,1}}(x_{1}-x_{2})p_1\|_{op}\nonumber
\\&+CN^{2}\|(p_1q_3\widehat{m}+q_1p_3\widehat{m}+q_1q_3\widehat{m}_{-1})\sqrt{V_1}(x_1,x_3)\Psi\|\label{cso}
\\&\hspace{1cm}\|\sum_{j=4}^N
 g_{8/9,1}(x_{1}-x_{j}) \sqrt{V_1}(x_1,x_3)p_1q_3q_j\mathds{1}_{\mathcal{B}_j}\Psi\|\nonumber
\;.
\end{align}
$(\ref{csn})$ is bounded by
\begin{align*}
CN^{1/2-16/9}\|\phi\|_\infty^3\;.
\end{align*}
Using Lemma \ref{kombinatorikb} and Lemma \ref{totalE} the first factor of $(\ref{cso})$ is bounded by $CN^{-3/2}$.
Using the abbreviation $\chi_j=\sqrt{V_1}(x_1,x_3)p_1q_3q_j\mathds{1}_{\mathcal{B}_j}\Psi$ we can write for the second factor
\begin{align*}
\|\sum_{j=4}^N
  g_{8/9,1}(x_{1}-x_{j})  \|^2
\leq& \nonumber
N^2\laa\chi_4 g_{8/9,1}(x_{1}-x_{4}) g_{8/9,1}(x_{1}-x_{5})\chi_5\raa
\\&+N\laa\chi_4(g_{8/9,1}(x_{1}-x_{4}))^2 \chi_4\raa
\end{align*}
Since $g_{8/9,1}(x_{1}-x_{5})$ and $g_{8/9,1}(x_{1}-x_{4})$ commute we get
\begin{align*}
\leq& N^2\|g_{8/9,1}(x_{1}-x_{4}) \chi_5\|^2
+N\|g_{8/9,1}(x_{1}-x_{4})) \chi_4\|^2
\end{align*}
Using this and $q_3=1-p_3$ it follows that
\begin{align}\label{threelines}
(\ref{cso})\leq& CN^{3/2}\|g_{8/9,1}(x_{1}-x_{4}) \sqrt{V_1}(x_1,x_3)p_1q_5\mathds{1}_{\mathcal{B}_5}\Psi\|
\\&\nonumber+CN^{3/2}\|g_{8/9,1}(x_{1}-x_{4}) \sqrt{V_1}(x_1,x_3)p_1p_3q_5\mathds{1}_{\mathcal{B}_5}\Psi\|
\\&\nonumber+CN\|g_{8/9,1}(x_{1}-x_{4})) \sqrt{V_1}(x_1,x_3)p_1q_3q_4\mathds{1}_{\mathcal{B}_4}\Psi\|
\end{align}
Note that for large enough $N$ the function $ \sqrt{V_1}(x_1,x_3)g_{8/9,1}(x_{1}-x_{4}) $ is different from zero only inside the set $a_{3,4}$. Since $\mathcal{B}_5$ and $a_{3,4}$ are by definition disjoint it follows that
\begin{align*}
g_{8/9,1}&(x_{1}-x_{4}) V_1(x_1,x_3) g_{8/9,1}(x_{1}-x_{5})p_1q_5 \mathds{1}_{\mathcal{B}_5}\Psi
\\=&g_{8/9,1}(x_{1}-x_{4}) V_1(x_1,x_3) g_{8/9,1}(x_{1}-x_{5})p_1q_5 \mathds{1}_{a_{3,4}}\mathds{1}_{\mathcal{B}_5}\Psi
\\=&0\;.
\end{align*}
Thus the second summand in (\ref{threelines}) is zero.
Using as above H\"older and Sobolev under the $x_2$-integral of the third  summand in (\ref{threelines})  we get that $(\ref{cso})$ is bounded by
\begin{align*}
&CN^{-23/18}\|\phi\|_\infty^3+
CN^{3/2}\| \sqrt{V_1}(x_1,x_3)p_3\|_{op}\|g_{8/9,1}(x_{1}-x_{4})p_1\|_{op}
\\&+N\|g_{8/9,1}\|_3^2\|\nabla_4 \sqrt{V_1}(x_1,x_3) g_{8/9,1}(x_{1}-x_{4})\widehat{m}p_1q_3q_4 \mathds{1}_{\mathcal{B}_4}\Psi\|^2
\\\leq&  CN^{-4/9}\|\phi\|_\infty^2+CN^{-1}(\ln N)^{2/3}\| \sqrt{V_1}(x_1,x_3)p_1\|_{op}^2 \|\nabla_4\widehat{m}q_3q_4 \mathds{1}_{\mathcal{B}_4}\Psi\|^2
\\\leq&  CN^{-4/9}\|\phi\|_\infty^2(\ln N)^{2/3}\;.
\end{align*}


\item

Using symmetry and Lemma \ref{kombinatorik} (d) $\ds$ is bounded by
\begin{align*}\ds\leq& N^4 \Im\left(\laa\Psi ,q_{1}q_{2} g_{8/9,1}(x_{1}-x_{2})p_{1}q_{2}\left[\potdiff_1(x_3,x_4),
p_3p_4(\widehat{m}-\widehat{m}_2)\right] \Psi\raa\right)
\\&+2N^4\Im\left(\laa\Psi ,q_{1}q_{2} g_{8/9,1}(x_{1}-x_{2})p_{1}q_{2}\left[\potdiff_1(x_3,x_4),
p_3q_4(\widehat{m}-\widehat{m}_1)\right] \Psi\raa\right)
\\\leq& N^4 \Im\left(\laa\Psi
,\potdiff_1(x_3,x_4)p_3p_4q_{1}q_{2}\widehat{n}^{-2}_1g_{8/9,1}(x_{1}-x_{2})p_{1}q_{2}\widehat{n}^2_2
(\widehat{m}-\widehat{m}_2)  \Psi\raa\right)
\\&+N^4 \Im\left(\laa\Psi ,q_{1}q_{2}\widehat{n}^{-2}_1g_{8/9,1}(x_{1}-x_{2})p_{1}q_{2} \widehat{n}^2_2
(\widehat{m}-\widehat{m}_2)p_3p_4\potdiff_1(x_3,x_4)  \Psi\raa\right)
\\&+2N^4\Im\left(\laa\Psi
,\potdiff_1(x_3,x_4)p_3q_4q_{1}q_{2}\widehat{n}^{-2}_1g_{8/9,1}(x_{1}-x_{2})p_{1}q_{2}\widehat{n}^2_2
(\widehat{m}-\widehat{m}_1)  \Psi\raa\right)
\\&+2N^4\Im\left(\laa\Psi ,q_{1}q_{2}\widehat{n}^{-2}_1 g_{8/9,1}(x_{1}-x_{2})p_{1}q_{2} \widehat{n}^2_2
(\widehat{m}-\widehat{m}_1)p_3q_4\potdiff_1(x_3,x_4)  \Psi\raa\right)
\end{align*}
With Lemma \ref{kombinatorikb} it follows that
\begin{align*}\ds\leq& CN^4 \|p_3\potdiff_1(x_3,x_4)\Psi\|
\;\|g_{8/9,1}(x_{1}-x_{2})p_{1}\|_{op}\|\widehat{n}\widehat{n}^2_2
(\widehat{m}-\widehat{m}_2)  \Psi\|
\\&+CN^4 \|g_{8/9,1}(x_{1}-x_{2})p_{1}\|_{op} \| \widehat{n}\widehat{n}^2_2
(\widehat{m}-\widehat{m}_2)\|_{op}\|p_3\potdiff_1(x_3,x_4)  \Psi\|
\\&+CN^4\|p_3\potdiff_1(x_3,x_4)\Psi\|\;\|g_{8/9,1}(x_{1}-x_{2})p_{1}\|_{op}\|\widehat{n}\widehat{n}^2_2
(\widehat{m}-\widehat{m}_1)  \Psi\|
\\&+CN^4\|g_{8/9,1}(x_{1}-x_{2})p_{1}\|_{op}\|\widehat{n} \widehat{n}^2_2
(\widehat{m}-\widehat{m}_1)\|_{op}\|p_3\potdiff_1(x_3,x_4)  \Psi\|\;.
\end{align*}

Recall that $\widehat{m}=\widehat{m}^0-\widehat{m}^0_1$. Due to Lemma \ref{mabsch}
\begin{align*} m(k)-m(k+1)=m^0(k)-2m^0(k+1)+m^0(k+2)
\leq  CN^{-2}n(k+1)^{-3}\end{align*} and \begin{align*} m(k)-m(k+2)=&m^0(k)-m^0(k+1)-m^0(k+2)+m^0(k+3)\\\leq&  CN^{-2}n(k+1)^{-3}\;.\end{align*} It follows that
\begin{align*}
\ds\leq& CN^4\|p_3\potdiff_1(x_3,x_4)\Psi\|\;\|g_{8/9,1}(x_{1}-x_{2})p_{1}\|_{op} N^{-2}
\;.
\end{align*}
In view of (\ref{hilfe3}) we get that also $\ds$ is bounded by the right hand side of (\ref{secondadjcontrol}).

\end{enumerate}

\end{proof}

\subsection{Proof of the Theorem for $\beta=  1$}\label{secproof3}

\begin{corollary}\label{corproof3}
Let  $\beta=1$. There exists a functional $\Gamma:\LZN\otimes\LZ\to\mathbb{R}^+$, a functional
$\Gamma':\LZN\otimes\LZ\to\mathbb{R}$ and a constant $c>0$ such that
\begin{enumerate}
 \item $$|\dt\Gamma(\Psi_t,\phi_t)|\leq |\Gamma'(\Psi_t,\phi_t)|\;.$$

\item $$ c\alpha(\Psi,\phi)-CN^{-\eta}\leq  \Gamma(\Psi,\phi)\leq  \alpha(\Psi,\phi)+CN^{-\eta}$$
uniform in $\Psi,\phi$
\item There exists a functional $\kinf$ such that $$|\Gamma'(\Psi,\phi)|\leq  \CphiA\C$$ uniform in $\Psi,\phi$.
\end{enumerate}

\end{corollary}

\begin{proof}
Set \begin{align*}\Gamma(\Psi,\phi):=&\gamma(\Psi,\phi)+\sum_{ j+k\leq  5}2^{-j-k}\gamma_{j,k}(\Psi,\phi)+|\mathcal{E}(\Psi)-\mathcal{E}^{GP}(\phi)|\hspace{1cm}\text{ and}
\\ \Gamma'(\Psi,\phi):=&\gamma'(\Psi,\phi)+\sum_{ j+k\leq  5}2^{-j-k}\gamma_{j,k}'(\Psi,\phi)+\dt|\mathcal{E}(\Psi)-\mathcal{E}^{GP}(\phi)|\;.\end{align*}

\begin{enumerate}
 \item follows from Lemma \ref{firstadjlemma} with Lemma \ref{secondadjlemma} and (\ref{ablediff}).
\item follows from Corollary \ref{corproof2} (b) together with Corollary \ref{replacealpa2}\;.
\item Remember that $\xi:=\gamma'_{0,0}-\Im\left(\xi_{0,0}\right)$ (see Corollary \ref{newalpha2}). Thus
\begin{align*}
\Gamma'(\Psi,\phi)=&\gamma'(\Psi,\phi)-\xi(\Psi,\phi)\\&+\sum_{1\leq  j+k\leq  5}2^{-j-k}\left(\gamma_{j,k}'(\Psi,\phi)+2\Im(\xi_{j-1,k})-\Im(\xi_{j,k})\right)
\\&+\sum_{j+k=5}2^{-5}\Im(\xi_{j,k})
+\dt|\mathcal{E}(\Psi)-\mathcal{E}^{GP}(\phi)|
\end{align*}
The first line is controlled by Lemma \ref{secondadjcontrol}.
The second line  by Lemma \ref{firstadjest}. The  third line is bounded by Proposition \ref{xiest} and (\ref{ablediff}).

\end{enumerate}
\end{proof}
From (b) and (c)   it follows that $$\Gamma'(\Psi,\phi)\leq  \CphiA\Cgamma$$ and we get via Gr\o nwall
$$\Gamma(\Psi_t,\phi_t)\leq  e^{\int_0^t (\|\phi_s\|_\infty+(\ln N)^{1/3}\|\nabla\phi_s\|_{6,loc}+\|\dot A_s\|_\infty)K(\phi_s)ds}(\Gamma(\Psi_0,\phi_0)+ N^{-\eta})\;.$$
For $\phi\in\mathcal{G}$ we have that $\sup_{s\in\mathbb{R}}\{\mathcal{K}(\phi_s)\}<\infty$.
Again using (b) we get the bound on $\alpha(\Psi_t,\phi_t)$ as stated in Theorem \ref{theorem}.

\section*{Acknowledgments}

I wish to thank Jean-Bernard Bru, Detlef D\"urr, J\"urg Fr\"ohlich, Antti Knowles and
Jakob Yngvason for many helpful comments on the paper.

\end{document}